\documentclass[12pt, twoside]{article}
\usepackage{amsmath,amsthm,amssymb}
\usepackage{times}
\usepackage{enumerate}
\usepackage{subfigure}
\usepackage{xcolor}
\usepackage{graphicx,color}
\usepackage[colorlinks,pdfdisplaydoctitle]{hyperref}
\usepackage{mathrsfs}
\usepackage{tikz}
\usepackage{tikz-cd}
\usetikzlibrary{matrix,arrows,decorations.pathmorphing,positioning}

\tikzset{
recfmt/.style={
	rectangle,
	rounded corners,
	draw=black, very thick,
	minimum height=1em,
	text centered
	      }	      
}

\theoremstyle{definition}
\newtheorem{thm}{Theorem}[section]

\newtheorem{lem}[thm]{Lemma}

\numberwithin{equation}{section}

\newcommand{\cF}{\mathcal{F}}

\newcommand{\gt}{\rightarrow}
\renewcommand{\O}{\mathcal{O}}
\newcommand{\R}{\mathbb{R}}
\newcommand{\E}{\mathbb{E}}
\newcommand{\Q}{\mathbb{Q}}
\newcommand{\I}{\mathbb{I}}
\newcommand{\PP}{\mathbb{P}}
\newcommand{\C}{\mathcal{C}}
\newcommand{\G}{\mathcal{G}}
\newcommand{\B}{\mathcal{B}}
\newcommand{\BB}{\mathscr{B}}
\newcommand{\N}{\mathcal{N}}
\newcommand{\T}{\mathcal{T}}

\newcommand{\RR}{\mathcal{R}}
\newcommand{\M}{\mathcal{M}}
\newcommand{\HH}{\mathcal{H}}

\newcommand{\dd}{ \,\textrm{d}}
\newcommand{\PAR}[1]{\partial #1}
\newcommand{\OL}[1]{\overline{#1}}

\newcommand{\be}{\begin{equation}}
\newcommand{\en}{\end{equation}}

\def\eref#1{(\ref{#1})}
\def\Fref#1{Figure \ref{#1}}

\newtheorem{assumptions}[thm]{Assumptions}

\newtheorem{remarks}[thm]{Remarks}
\newtheorem{algorithm}[thm]{Algorithm}
\newtheorem{proposition}[thm]{Proposition}

\begin{document}


\baselineskip=17pt


\title{A Bayesian Level Set Method\\for Geometric Inverse Problems}

\author{Marco A. Iglesias$^\dag$, Yulong Lu$^\ast$ and Andrew M. Stuart$\ddag$\\
$^\dag$ School of Mathematical Sciences, University of Nottingham, Nottingham NG7 2RD, UK\\
$^\ast$ Mathematics Institute, University of Warwick, Coventry CV4 7AL, UK\\
$\ddag$ Mathematics Institute, University of Warwick, Coventry CV4 7AL, UK\\
$^\ast$ corresponding author. yulong.lu@warwick.ac.uk}

\date{}

\maketitle


\begin{abstract}
We introduce a level set based  approach to Bayesian geometric inverse 
problems. In these problems the interface between different domains 
is the key unknown, and is realized as the level set of a function. This
function itself becomes the object of the inference. Whilst the level set 
methodology has been widely used for the solution of geometric inverse problems,
the Bayesian formulation that we develop here contains two significant
advances: firstly it leads to a well-posed inverse problem in which the
posterior distribution is Lipschitz with respect to the observed data;
and secondly it leads to computationally expedient algorithms in which
the level set itself is updated implicitly via the MCMC methodology
applied to the level set function -- no explicit velocity field
is required for the level set interface. Applications are numerous
and include medical imaging, modelling of subsurface formations 
and the inverse source problem; our theory is illustrated with
computational results involving the last two applications.
\end{abstract}

\section{Introduction}
\label{sec:intro}

Geometric inverse problems, in which the interfaces between different
domains are the primary unknown quantities, are ubiquitous in
applications including medical imaging problems such as EIT
\cite{B02} and subsurface flow \cite{armstrong:11}; they also have an 
intrinsically interesting mathematical structure \cite{I90}.
In many such applications the data is sparse, so that the problem is
severely under-determined, and noisy. For both these reasons
the Bayesian approach to the inverse problem is attractive as
the probabilistic formulation allows for regularization of
the under-determined, and often ill-posed, inversion via the
introduction of priors, and simultaneously deals with the presence 
of noise in the observations \cite{KS05,S10}. The level set method
has been a highly successful methodology for the solution of
classical, non-statistical, inverse problems for interfaces since the seminal
paper of Santosa \cite{S96}; see for example \cite{CCT05,Bur01,Bur03,VC02,LLT06,LLT206,TC04,BO05,DL06,DL09,wang2015binary,IM11,ABH04,astrakova:14}
and for related Bayesian level set approaches see
\cite{xie2011uncertainty,Lorentzen2012,Lorentzen2012B,Ping2014}.

In this paper we marry the level set approach with the Bayesian approach
to geometric inverse problems. This leads to two significant
advances: firstly it leads to a well-posed inverse problem in which the
posterior distribution is Lipschitz with respect to the observed data,
in the Hellinger metric -- there is no analogous well-posedness
theory for classical level set inversion; and secondly it leads to 
computationally expedient algorithms in which the interfaces are updated 
implicitly via the Markov chain Monte Carlo (MCMC) 
methodology applied to the level set function -- 
no explicit velocity field is required for the level set interface. 
We highlight that the
recent paper \cite{wang2015binary} demonstrates the potential for working
with regularized data misfit minimization in terms of a level set
function, but is non-Bayesian in its treatment of the problem, using
instead simulated annealing within an optimization framework.
On the other hand the paper \cite{xie2011uncertainty} adopts a
Bayesian approach and employs
the level set method, but requires a velocity field for propagation of
the interface and does not have the conceptual and implementational
simplicity, as well as the underlying theoretical basis, 
of the method introduced here. 
The papers \cite{Lorentzen2012,Lorentzen2012B,Ping2014}, 
whilst motivated by the Bayesian approach,
use the ensemble Kalman filter and are therefore not strictly Bayesian --
the method does not deliver a provably reliable approximation of the posterior
distribution except for linear Gaussian inverse problems. 

The key idea which underpins our work is this. Both the theory and
computational practice of the level set method for geometric inverse
problems is potentially hampered by the fact that the mapping from
the space of the level set function to the physical parameter space
is discontinuous. This discontinuity occurs when the level set
function is flat at the critical levels, and in particular where
the desired level set has non-zero Lebesgue measure. This is dealt
with in various {\em ad hoc} ways in the applied literature. The
beauty of the Bayesian approach is that, with the right choice of
prior in level set space, these discontinuities have probability
zero. As a result a well-posedness theory (the posterior is Lipschitz
in the data) follows automatically, and computational algorithms
such as MCMC may be formulated in level set space. We thus have
practical algorithms which are simultaneously founded on a 
sound theoretical bedrock. 

In section \ref{sec:blip} we set up the inverse problem of interest,
describe the level set formulation, and state assumptions under which
we have a well-posed inverse problem for the level set function.
We also characterize the discontinuity set of the level set map
and demonstrate the existence of Gaussian priors for which this
discontinuity set is a probability zero event. In section
\ref{sec:exa} we describe two examples -- inverse gravimetry
and permeability determination in groundwater flow -- which
can be shown to satisfy the theoretical framework of the 
preceding section and hence for which there is a well-posed
inverse problem for the level set function. Section \ref{sec:num}
contains numerical experiments for both of these examples,
demonstrating the potential of the methodology, and also
highlighting important questions for future research. We
conclude in section \ref{sec:con}, and then the two appendices contain
various technical details and proofs which have been deferred
in order to maintain the flow of ideas in the main body of the article.

\section{Bayesian Level Set Inversion}
\label{sec:blip}

\subsection{The Inverse Problem}
\label{ssec:tip}

This paper is concerned with inverse problems of the following type:
recover function $\kappa \in X:=L^{q}(D;\R)$, $D$ a bounded open set in
$\R^2$, from a finite set of
noisily observed linear functionals $\{\O_j\}_{j=1}^J$ of $p \in V$, for some Banach space
$V$, where $p=G(u)$ for nonlinear operator $G \in C(X,V).$
Typically, for us, $\kappa$ will represent input data for a 
partial differential equation (PDE),
$p$ the solution of the PDE and $G$ the solution operator mapping
the input $\kappa$ to the solution $p$. Collecting the linear functionals
into a single operator $\O: V \to \R^J$ and assuming additive noise
$\eta$ we obtain the inverse problem of finding $\kappa$ from $y$ where
\begin{equation}
\label{eq:IP1}
y=(\O \circ G)(\kappa)+\eta.
\end{equation}
This under-determined inverse problem is well-adapted to both the 
classical \cite{EHN96} and Bayesian \cite{S13} approaches
to regularized inversion, because $\O \circ G \in C(X,\R^J).$

\subsection{Level Set Parameterization}
\label{ssec:lsp}

There are many inverse problems where $\kappa$ is known
{\em a priori} to have the form
\begin{equation}
\label{eq:kd}
 \kappa(x) = \sum_{i=1}^n \kappa_i \I_{D_i}(x);
\end{equation}
here $\I_D$ denotes the indicator function of subset $B \subset \R^2$,
$\{D_i\}_{i=1}^n$ are subsets of $D$ such that $\cup_{i=1}^n \OL{D_i} = \OL{D}$ and $D_i \cap D_j = \varnothing$, and the
$\{\kappa_i\}_{i=1}^n$ are known positive constants.\footnote{Generalization to the $\kappa_i$ being unknown constants, or
unknown smooth functions on each domain $D_i$, are possible
but will not be considered explicitly in this paper.
Our focus is on the geometry of the interfaces implied by the $D_i.$}
In this setting the $D_i$ become the primary unknowns and the level set
method is natural. Given integer $n$ fix the constants $c_i\in \R$ for $i = 0, \cdots, n$ with $-\infty = c_0 < c_1 < \cdots <  c_n = \infty$ and consider a real-valued
continuous {\em level set  function} $u:D \to \R$. We can then define the $D_i$
by 
\be\label{eq_lss}
\ D_i = \{x\in D\ |\ c_{i-1} \leq u(x) < c_i\}.
\en
It follows that $\ D_i \cap D_{j} = \varnothing$ for $i,j \geq 1,\ i \neq j$. 
For later use define the $i$-th level set $D^0_i = \OL{D_i} \cap \OL{D_{i+1}} = \{x\in D\ |\ u(x) = c_i\}$. Let $U=C(\OL{D};\R)$ and, given the positive constants 
$\{\kappa_i\}_{i=1}^n$,
we define the level set map $F: U \to X$ by
\begin{equation}\label{eq_lsf}
(Fu)(x) \gt \kappa(x) = \sum_{i=1}^n \kappa_i\, \I_{D_i}(x).
\end{equation}
We may now define $\G=\O \circ G \circ F: U \to \R^J$ and reformulate
the inverse problem in terms of the level set function $u$: find
$u$ from $y$ where
\begin{equation}
\label{eq:IP2}
y=\G(u)+\eta.
\end{equation}
However, because $F: U \to X$, and hence $\G: U \to \R^J$, is discontinuous, 
the classical regularization theory for this form of inverse problem
is problematic.
Whilst the current state of the art for Bayesian regularization
assumes continuity of $\G$ for inverse problems of the form
\eref{eq:IP2}, we will demonstrate that the Bayesian setting
can be generalized to level set inversion. This will be achieved
by a careful understanding of the discontinuity set for $F$,
and an understanding of probability measures for which this set
is a measure zero event.

\subsection{Well-Posed Bayesian Inverse Problem}
\label{ssec:wpb}

In the Bayesian approach to inverse problems, all quantities in \eref{eq:IP2}
are treated as random variables. Let $U$ denote a separable Banach space
and define a complete\footnote{Complete probability space is defined at
the start of Appendix 2} probability space $(U,\Sigma,\mu_0)$ for the unknown 
$u$. (In our applications $U$ will be the space $C(\OL{D};\R)$ but we
state our main theorem in more generality). Assume that the noise $\eta$ is a 
random draw from the centered Gaussian $\Q_0 :=\N(0,\Gamma)$; 
we also assume that $\eta$ and $u$ are independent.\footnote{Allowing for 
non-Gaussian $\eta$ is also possible, as is dependence between $\eta$ and $u$;
however we elect to keep the presentation simple.} We may now define the 
joint random variable $(u,y) \in U \times \R^J$.
The Bayesian approach to the inverse problem \eref{eq:IP2} consists of seeking
the posterior probability distribution $\mu^y$ of the random variable $u|y$, 
$u$ given $y$. This measure $\mu^y$ describes our probabilistic knowledge about $u$
on the basis of the measurements $y$ given by \eref{eq:IP2} and the prior 
information $\mu_0$ on $u$.

We define the {\em least squares function} $\Phi: U \times \R^J \to \R^+$ by
\be\label{Phi}
\Phi(u; y) = \frac{1}{2} |y - \G(u)|^2_\Gamma
\en
with $|\cdot|:=|\Gamma^{-\frac12}\cdot|.$
We now state a set of assumptions under which the posterior distribution
is well-defined via its density with respect to the prior distribution,
and is Lipschitz in the Hellinger metric, with respect to data $y$.

\begin{assumptions}
\label{a:1}
The least squares function $\Phi: U\times \R^J\gt \R$ and probability
measure $\mu_0$ on the measure space $(U,\Sigma)$ 
satisfy the following properties:
\begin{enumerate}

\item for every $r > 0$ there is a $K = K(r)$ such that, for all $u\in U$ and 
all $y\in \R^J$ with $|y|_\Gamma < r$,
 \[
 0 \leq \Phi(u; y) \leq K;
 \]

\item for any fixed $y\in \R^J$, $\Phi(\cdot;y): U\gt \R$, is continuous
$\mu_0$-almost surely on the complete probability space $(U, \Sigma, \mu_0)$;

\item for $y_1, y_2\in \R^J$ with $\max\{|y_1|_\Gamma, |y_2|_\Gamma\} < r$, there exists a $C = C(r)$ such that, for all $u\in U$,
 \[
 |\Phi(u; y_1) - \Phi(u; y_2)| \leq C|y_1 - y_2|_\Gamma.
 \]
\end{enumerate}
\end{assumptions}

The Hellinger distance between $\mu$ and $\mu^\prime$ is defined as
  \[
  d_{{\rm Hell}}(\mu, \mu^\prime) = \left(\frac{1}{2} \int_{U} \left(\sqrt{\frac{\dd \mu}{\dd \nu}} - \sqrt{\frac{\dd \mu^{\prime}}{\dd \nu}}\right)^2 \dd \nu \right)^{\frac{1}{2}}
  \]
for any measure $\nu$ with respect to which $\mu$ and $\mu^\prime$ are
absolutely continuous; the Hellinger distance is, however, independent
of which reference measure $\nu$ is chosen. We have the following: 

\begin{thm} 
\label{t:main}
Assume that the least squares function $\Phi: U \times \R^J
\to \R$ and the probability measure $\mu_0$ on the measure space
$(U,\Sigma)$ satisfy Assumptions \ref{a:1}.  Then 
$\mu^y \ll \mu_0$ with Radon-Nikodym derivative
  \be\label{posterior}
  \frac{\dd \mu^y}{\dd \mu_0} = \frac{1}{Z} \exp(-\Phi(u;y))
  \en
  whera, for $y$ almost surely,
  \[
  Z := \int_U \exp(-\Phi(u; y))\mu_0 (du) > 0.
\]
Furthermore $\mu^y$ is locally Lipschitz with respect to $y$, in the Hellinger distance:
for all $y, y^\prime$ with $\max\{|y|_\Gamma, |y^\prime|_\Gamma\} < r$, there exists a $C = C(r) > 0$ such that
  \[
  d_{{\rm Hell}}(\mu^y, \mu^{y^\prime}) \leq C|y - y^\prime|_\Gamma.
  \]
This implies that, for all $f \in L^2_{\mu_0}(U;S)$ for separable
Banach space $S$,
$$\|\E^{\mu^y} f(u)-\E^{\mu^{y^\prime}} f(u)\|_{S} \le C|y-y^\prime|.$$
\end{thm}

\begin{remarks}
\begin{itemize}

\item The {\em interpretation} of this result is very natural,
linking the Bayesian picture with least squares minimisation:
the posterior measure is large on sets where the least squares
function is small, and vice-versa, all measured relative to the
prior $\mu_0.$

\item The key technical advance in this theorem over existing theories
overviewed in \cite{S13} is that $\Phi(\cdot;y)$ is only continuous
$\mu_0-$almost surely; existing theories typically use that
$\Phi(\cdot;y)$ is continuous everywhere on $U$ and that $\mu_0(U)=1$;
these existing theories cannot be used in the level set inverse problem,
because of discontinuities in the level set map.
Once the technical Lemma \ref{lem_asc} has been 
established, which uses $\mu_0-$almost sure continuity to establish
measurability,
the {\em proof} of the theorem is a straightforward application of
existing theory; we therefore defer it to Appendix 1. 

\item What needs to be
done to {\em apply} this theorem in our level set context is to identify
the sets of discontinuity for the map $\G$, and hence $\Phi(\cdot;y)$,
and then construct prior measures $\mu_0$ for which these sets have measure 
zero.  We study these questions in general terms in the next two subsections,
and then, in the next section,
demonstrate two test model PDE inverse problems where
the general theory applies.

\item The {\em consequences} of this result are wide-ranging, and we
name the two primary ones: firstly we may apply the mesh-independent
MCMC methods overviewed in \cite{CRSW08} to sample the posterior
distribution efficiently; and secondly the well-posedness gives
desirable robustness which may be used to estimate the effect
of other perturbations, such as approximating $G$ by a numerical method,
on the posterior distribution \cite{S13}. 

\end{itemize}
\end{remarks}

\subsection{Discontinuity Sets of $F$}
\label{ssec:dis}

We return to the specific setting of the level set inverse
problem where $U=C(\OL{D};\R)$.
We first note that the level set map $F:U \to L^{\infty}(D;\R)$ is not
continuous except at points $u \in U$ where no level crossings
occur in an open neighbourhood. However as a mapping $F:U \to L^{q}(D;\R)$ for $q<\infty$
the situation is much better:

\begin{proposition}\label{thm_cont}
 For $u\in C(\OL{D})$ and $1 \leq q < \infty$, the level set map $F: C(\OL{D})\gt L^q(D)$ is continuous at $u$ if and only if $m(D^0_i) = 0$ for all $i=1,\cdots, n-1$.
\end{proposition}

The proof is given in Appendix 1. For the inverse gravimetry problem
considered in the next section the space $X$ is naturally $L^{2}(D;\R)$ 
and we will be able to directly use the preceding proposition to establish
almost sure continuity of $F$ and hence $\G$.
For the groundwater flow inverse problem 
the space $X$ is naturally $L^{\infty}(D;\R)$
and we will not be able to use the proposition in this space to establish
almost sure continuity of $F$. However, we employ recent Lipschitz
continuity results for $G$ on $L^q(D;\R)$, $q < \infty$ 
to establish the almost sure continuity of $\G$ \cite{BDN03}.

\subsection{Prior Gaussian Measures}
\label{ssec:pri}

Let $D$ denote a bounded open subset of $\R^2$.
For our applications we will use the following two constructions
of Gaussian prior measures $\mu_0$ which are Gaussian $\N(0, \C_i), i=1,2$
on Hilbert function space $\HH_i, i=1,2$.
\begin{itemize}
 \item $\N(0, \C_1)$ on
 \[
 \HH_1 := \{u:D\gt \R\ |\ u\in L^2(D;\R)), \int_{D} u(x) \dd x = 0\},
 \]
 where
  \begin{equation}\label{eq:cova1}
\C_1 = \mathcal{A}^{-\alpha} \qquad  \textrm{and} \qquad \mathcal{A} := -\Delta
\end{equation}
 with domain\footnote{Here $\nu$ denotes the outward normal.}
 \[
 D(\mathcal{A}) := \{u: D\gt \R \ |\ u\in H^2(D;\R), \nabla u \cdot \nu = 0 \textrm{ on } \PAR{D} \textrm{ and } \int_D u(x) \dd x = 0\}.
 \]

 \item $\N(0, \C_2)$ on $\HH_2 := L^2(D;\R)$ with $\C_2:\HH_2 \gt \HH_2$ being the integral operator
  \begin{equation}\label{eq:cova2}
 \C_2 \phi(x) = \int_{D}c(x, y)\phi (y)\dd y \ \textrm{ with } c(x, y) = \exp\left(-\frac{|x - y|^2}{L^2}\right)
\end{equation}
   \end{itemize}
In fact, in the inverse model arising from groundwater flow studied
in \cite{ILS13geosciences,ILS13}, the Gaussian measure $\N(0, \C_1)$ was taken as the prior measure for the logarithm of the permeability. On the other
hand the Gaussian measure $\N(0,\C_2)$ is widely used to model the earth's subsurface \cite{ORL08} as
draws from this measure generate smooth random functions in which the
parameter $L$ sets the spatial correlation length.
For both of these measures it is known that, under suitable conditions
on the domain $D$, draws are almost surely in $C(\OL{D};\R)$; see
\cite{S13}, Theorems 2.16 and 2.18 for details.

Assume that $\alpha > 1$ in \eref{eq:cova1}, then the Gaussian random function with measure $\mu_0$ defined in either case above 
has the property that, for $U:=C(\OL{D};\R)$, $\mu_0(U)=1$. Since $U$ is a 
separable Banach space $\mu_0$ can be redefined as a Gaussian measure on $U$.
Furthermore it is possible to define the appropriate $\sigma$-algebra $\Sigma$
in such a way that $(U, \Sigma, \mu_0)$ is a complete probability space; 
for details see Appendix 2.  We have the following, 
which is a subset of what is proved in Proposition \ref{thm_lsm}.

\begin{proposition}\label{pro_lsm}
Consider a random function $u$ drawn from one of the Gaussian probability
measures $\mu_0$ on $U$ given above. 
Then $m(D^0_i) = 0, \,\mu_0$-almost surely, for $i=1,\cdots, n.$
\end{proposition}

This, combined with Proposition \ref{thm_cont}, is the key to
making a rigorous well-posed formulation of Bayesian level set inversion.
Together the results show that priors may be constructed for which
the problematic discontinuities in the level set map are probability
zero events. In the next section we demonstrate how the theory may be applied,
by consider two examples.

\section{Examples}
\label{sec:exa}

\subsection{Test Model 1 (Inverse Potential Problem)}
\label{ssec:tm1}

Let $D \subset \R^2$ be a bounded open set with Lipschitz boundary.
Consider the PDE
\begin{equation}\label{eq:source}
\Delta p = \kappa \quad \quad \textrm{ in } D, \quad \quad p = 0  \quad \quad  \textrm{ on } \PAR{D}. 
\end{equation}
If $\kappa \in X:=L^2(D)$ it follows that
there is a unique solution $p\in H^1_0(D)$. 
Furthermore $\Delta p\in L^2(D)$, so that the Neumann trace can 
be defined in $V:=H^{-\frac{1}{2}}(\PAR{D})$ by the following 
Green's formula:
\[
\Bigl \langle \frac{\PAR{p}}{\PAR{\nu}}, \varphi \Bigr \rangle_{\PAR{D}} = \int_{D} \Delta p \varphi \dd x + \int_{D} \nabla p \nabla \varphi \dd x
\] 
for $\varphi \in H^{1}(D)$. Here $\nu$ is the unit outward normal vector on $\PAR{D}$ and $\langle \cdot, \cdot \rangle_{\PAR{D}}$ denotes the dual pairing on the boundary. 
We can then define the bounded linear map $G:X \to V$ by 
$G(\kappa) = \frac{\PAR{p}}{\PAR{\nu}}.$

Now assume that the source term $\kappa$ has the form
$$\kappa(x)=\I_{D_1}(x)$$
for some $D_1 \subseteq D$.
The inverse potential problem is to reconstruct the support $D_1$ 
from measurements of the Neumann data of $p$ on $\PAR{D}$. 
In the case where the Neumann data is measured everywhere on the
boundary $\PAR{D}$, and where the domain $D_1$ is assumed 
to be star-shaped with respect to its center of gravity, the
inverse problem has a unique solution; see \cite{I90, I06} for details
of this theory and see \cite{HR96, BO05} for discussion of numerical methods
for this inverse problem. We will study the underdetermined case
where a finite set of bounded linear functionals $\O_j: V \to \R$ are measured, 
noisily, on $\PAR{D}$:
\begin{equation}\label{eq:source2A}
y_j=\O_j\Bigl(\frac{\PAR{p}}{\PAR{\nu}}\Bigr)+\eta_j.
\end{equation}
Concatenating we have $y=(\O \circ G)(\kappa)+\eta.$
Representing the boundary of $D_1$ as the zero level set of function
$u \in U:=C(\overline{D};\R)$ we write the inverse problem 
in the form \eref{eq:IP2}:
\begin{equation}\label{eq:source2}
y=(\O \circ G \circ F)(u)+\eta.
\end{equation}
Since multiplicity and uncertainty
of solutions are then natural, we will adopt a Bayesian approach.

Notice that the level set map $F:U \to X$ is bounded: for all $u \in U$ we have
$\|F(u)\|_X \le {\rm Vol}(D):=\int_{D} dx.$
Since $G:X \to V$ and $\O: V \to \R^J$ are bounded linear maps
it follows that $\G=\O \circ G \circ F: U \to \R^J$ is bounded:
we have constant $C^+ \in \R^+$ such that, for all $u \in U$, 
$|\G(u)| \le C^+.$ From this fact Assumptions \ref{a:1}(i) and (iii)
follow automatically. 
Since both $G:X \to V$ and $\O:V \to \R^J$ are bounded, and hence continuous, 
linear maps, the discontinuity set of $\G$ is determined by the the 
discontinuity set of $F:U \to X.$ By Proposition \ref{thm_cont} this
is precisely the set of functions for which the measure of the level
set $\{u(x)=0\}$ is zero. By Proposition \ref{pro_lsm} this occurs with
probability zero for both of the Gaussian priors specified there and
hence  Assumptions \ref{a:1}(ii) holds with these priors. Thus
Theorem \ref{t:main} applies and we have a well-posed Bayesian
inverse problem for the level set function.

\vspace{0.1in}

\subsection{Test Model 2 (Discontinuity Detection in Groundwater Flow)}
\label{ssec:tm2}

Consider the single-phase Darcy-flow model given by 

\begin{equation}\label{sdarcy}
  -\nabla\cdot (\kappa \nabla p)   = f \quad   \textrm{ in } D, \\ 
 p   =  0   \quad  \textrm{ on } \PAR{D}.
\end{equation}
Here $D$ is a bounded Lipschitz domain in $\R^2$, $\kappa$ the real-valued
isotropic permeability function and $p$ the fluid pressure. The right hand side $f$ accounts for the source of groundwater recharge. 
Let $V=H^1_0(D;\R)$, $X=L^{\infty}(D;\R)$ and $V^*$ denote the dual
space of $V$. If $f \in V^*$ and $X^+:=\{\kappa \in X: {\rm essinf}_{x \in D}
\kappa (x) \ge \kappa_{\rm min}>0\}$ then $G: X^+ \mapsto V$ defined by $G(\kappa)=p$ is Lipschitz continuous and 
\be \label{eq_wellpos}
\|G(\kappa)\|_V=\|p\|_V \leq \|f\|_{V^\ast}/ \kappa_{\min}.
\en

We consider the practically useful situation
in which the permeability function $\kappa$ is modelled as piecewise
constant on different regions $\{D_i\}_{i=1}^n$ whose union comprise $D$;
this is a natural way to characterize heterogeneous subsurface structure
in a physically meaningful way. We thus have 
\[
 \kappa(x) = \sum_{i=1}^n \kappa_i \I_{D_i}(x)
\]
where $\{D_i\}_{i=1}^n$ are subsets of $D$ such that $\cup_{i=1}^n \OL{D_i} = \OL{D}$ and $D_i \cap D_j = \varnothing$, and where the
$\{\kappa_i\}_{i=1}^n$ are positive constants. 
We let $\kappa_{\rm min}=\min_i \kappa_i$.

Unique reconstruction of the permeability in some situations
is possible if the pressure $p$ is measured everywhere \cite{alessandrini1985identification,richter81}. 
The inverse problem of interest to us is to locate the discontinuity set
of the permeability from a finite set of
measurements of the pressure $p$. Such problems have
also 
been studied in the literature. For instance, the paper \cite{TC04} considers
the problem by using multiple level set methods in the framework of 
optimization; and  in \cite{ILS14}, the authors adopt a Bayesian approach to 
reconstruct the permeability function characterized by layered or 
channelized structures whose geometry can be parameterized finite 
dimensionally.  As we consider a finite set of noisy measurements
of the pressure $p$, in $V^\ast$, and the problem is underdetermined
and uncertain, the Bayesian approach is again natural. We make the
significant extension of \cite{ILS14} to consider arbitrary interfaces,
requiring infinite dimensional parameterization: we introduce
a level set parameterization of the domains $D_i$, as in 
\eref{eq_lss} and \eref{eq_lsf}.

Let $\O: V \to \R^J$ denote the collection of $J$ linear functionals
on $V$ which are our measurements.  Because of the estimate \eref{eq_wellpos}
it is straightforward to see that $\G=\O\circ G \circ F$ is bounded
as a mapping from $U$ into $\R^J$ and hence that Assumptions \ref{a:1}(i)
and (iii) hold; it remains to establish (ii). To that end, from now on we need slightly higher regularity on $f$. In particular, we assume that, 
for some $q > 2$, $f\in W^{-1}(L^q(D))$. Here the space 
$W^{-1}(L^q(D)) := ({W^{1,q^\ast}_0(D)})^\ast \subset V^\ast$ for $q^\ast$ 
and $q$ conjugate: $1/q + 1/q^\ast = 1$.  It is shown 
in \cite{BDN03} that there 
exits $q_0 > 2$ such that the solution of \eref{sdarcy} satisfies
$$\|\nabla p\|_{L^q(D)} \leq C \|f\|_{W^{-1}(L^q(D))}$$
for some $C < \infty$ provided $2\leq q < q_0$. We assume that such a $q$
is chosen. It then 
follows that $G$ is Lipschitz continuous from $L^r$ to $V$ 
where $r := 2q/ (q - 2) \in [2, \infty)$. To
be precise, let $p_i$ be the solution to the problem \eref{sdarcy} with $\kappa_i, i = 1,2$. Then the following is proved in \cite{BDN03}:  for any $q\geq 2$, 
$$\|p_1 - p_2\|_V \leq \frac{1}{\kappa_{\min}} \|\nabla p_1\|_{L^q(D)} \|\kappa_1 - \kappa_2\|_{L^r(D)}$$ 
provided $\nabla p_1\in L^q(D)$.

Hence $G: L^r(D) \gt V$ is Lipschitz under our assumption
that $f\in W^{-1}(L^q(D))$ for some $q \in (2,_0).$ . 
By viewing $F: U\gt L^r(D)$, it follows from Proposition \eref{thm_cont} and Proposition \eref{pro_lsm} that Assumptions \eref{a:1} (ii) holds with both Gaussian priors defined in subsection \ref{ssec:pri}. As a consequence Theorem \ref{t:main} also applies in the groundwater flow model. 

\section{Numerical Experiments}
\label{sec:num}

Application of the theory developed in subsection \ref{ssec:wpb} ensures that, for the choices of Gaussian priors discussed in subsection \ref{ssec:pri}, the posterior measure on the level set is well defined and thus suitable for numerical interrogation. In this section we display numerical experiments where we characterize the posterior measure by means of sampling with MCMC. In concrete we apply the preconditioned Crank-Nicolson (pCN) MCMC method explained
in \cite{CRSW08}. We start by defining this algorithm. 
Assume that we have a prior Gaussian measure $\N(0,\C)$ on the
level set function $u$ and a posterior measure $\mu^y$ given 
by \eref{posterior}. Define
$$a(u,v)=\min\{1,\exp\bigl(\Phi(u)-\Phi(v)\bigr)\}$$
and generate $\{u^{(k)}\}_{k \ge 0}$ as follows:

\begin{algorithm}[pCN-MCMC]\label{MCMC}~~

Set $k=0$ and pick $u^{(0)} \in X$
\begin{enumerate}

\item Propose $v^{(k)}=
\sqrt{(1-\beta^2)}u^{(k)}+\beta  \xi^{(k)},
\quad \xi^{(k)} \sim \N(0,\C)$.

\item Set $u^{(k+1)}=v^{(k)}$ with probability $a(u^{(k)},v^{(k)})$, independently of $(u^{(k)},\xi^{(k)})$.

\item Set $u^{(k+1)}=u^{(k)}$ otherwise.

\item $k \to k+1$ and return to (i).

\end{enumerate}
\end{algorithm}

Then the resulting Markov chain is reversible with respect to $\mu^y$ and,
provided it is ergodic, satisfies
$$\frac{1}{K}\sum_{k=0}^K \varphi\bigl(u^{(k)}\bigr) \to \E^{\mu^y} \varphi(u)$$for a suitable class of test functions. Furthermore a central limit theorem 
determines the fluctuations around the limit, which are asymptotically
of size $K^{-\frac12}.$

\subsection{Aim of the experiments}
By means of the MCMC method described above we explore the Bayesian posterior of the level set function that we use to parameterize unknown geometry (or discontinuous model parameters) in the geometric inverse problems discussed in Section \ref{sec:exa}. The first experiment of this section concerns the inverse potential problem defined in subsection  \ref{ssec:tm1}. The second and third experiments are concerned with the estimation of geologic facies for the groundwater flow model discussed in subsection \ref{ssec:tm2}. The main objective of these experiments is to display the capabilities of the level set Bayesian framework to provide an estimate, along with a measure of its uncertainty, of unknown discontinuous model parameters in these test models. We recall that for the inverse potential problem the aim is to estimate the support $D_{1}$ of the indicator function $\kappa(x)=\I_{D_{1}}(x)$ that defines the source term of the PDE (\ref{eq:source}) given data/observations from the solution of this PDE. Similarly, given data/observations from the solution of the Darcy flow model (\ref{sdarcy}), we wish to estimate the interface between geologic facies $\{D_{i}\}_{i=1}^{n}$  corresponding to regions of different structural geology and which leads to a discontinuous permeability $\kappa(x)=\sum_{i=1}^{n}\kappa_i\I_{D_{i}}(x)$ in the flow model (\ref{sdarcy}). In both test models, we introduce the level set function merely as an artifact to parameterize the unknown geometry (i.e. $\ D_i = \{x\in D\ |\ c_{i-1} \leq u(x) < c_i\}$), or equivalently, the resulting discontinuous field $\kappa(x)$. The Bayesian framework applied to this level/
set parameterization then provides us with a posterior measure $\mu^y$ on the level set function $u$. The push-forward of $\mu^y$ under the level set map $F$ (\ref{eq_lsf}) results in a distribution on the discontinuous field of interest $\kappa$. 
This push-forward of the level set posterior $F^\ast \mu^y := \mu^y \circ F^{-1}$ comprises the statistical solution of the inverse problem which may, in turn, be used for practical applications.

A secondary aim of the experiments is to explore the role of the choice
of prior on the posterior. Because the prior is placed on the
level set function, and not on the model paramerers of direct interest,
this is a non-trivial question. To be concrete, the posterior depends on the Gaussian prior that we put on the level set. While the prior may incorporate our a priori knowledge concerning the regularity and the spatial correlation of the unknown geometry (or alternatively, the regions of discontinuities in the fields of interest) it is clear that such selection of the prior on the level set may have a strong effect on the resulting posterior $\mu^y$ and the corresponding push-forward $F^\ast \mu^y$. One of the key aspects of the subsequent numerical study is to understand the role of the selection of the prior on the level set
functions in terms of the quality and efficiency of the solution to the 
Bayesian inverse problem as expressed via the push-forward of the
posterior $F^\ast \mu^y$.

\subsection{Implementational aspects}

For the numerical examples of this section we consider synthetic experiments. The PDEs that define the forward models of subsection \ref{sec:exa} (i.e. expressions (\ref{eq:source})  and (\ref{sdarcy})) are solved numerically, on the unit-square, with cell-centered finite differences \cite{Arbogast}. In order to avoid inverse crimes \cite{KS05}, for the generation of synthetic data we use a much finer grid (size specified below) than the one of size $80\times 80$ used for the inversion via the MCMC method displayed in Algorithm \ref{MCMC}. 

The Algorithm \ref{MCMC} requires, in step (i), sampling of the prior. This is accomplished by parameterizing the level set function in terms of the Karhunen-Loeve (KL) expansion associated to the prior covariance operator $\C$ (See Appendix 2, equation (\ref{eq_rf})). Upon discretization, the number of eigenvectors of $\C$ equals the dimensions of the discretized physical domain of the model problems (i.e. $N=6400$ in expression (\ref{eq_rf2})). Once the eigendecomposition of $\C$ has been conducted, then sampling from the prior can be done simply by sampling an i.i.d set of random variables $\{\xi_{k}\}$ with $\xi_1 \sim \N(0,1)$ and using it in (\ref{eq_rf2}). For the present experiments we consider all these KL modes without any truncation. However, during the burn-in period 
(which here is taken to comprise $10^4$ iterations) of the MCMC method, we 
find it advantageous to freeze the higher KL modes and conduct the sampling only for the lower modes. After the aforementioned burn-in, the sampling is then carried out on the full set of KL modes. This mechanism enables the MCMC method to quickly reach an ``optimal'' state where the samples of the level set function provides a field $\kappa(x)$ that is close to the truth. However, once this optimal state has been reached, it is essential to conduct the sampling on the full spectrum of KL modes to ensure that the MCMC chain mixes properly. More precisely, if only the lowest modes are retained for the full chain, the MCMC may collapse into the optimal state but without mixing. Thus, while the lowest KL modes determine the main geometric structure of the underlying discontinuous field, the highest modes are essential for the proper mixing and thus the proper and efficient characterization of the posterior.


\subsection{Inverse Potential Problem}

In this experiment we generate synthetic data by solving (\ref{eq:source}), on a fine grid of size $240\times 240$ with the ``true'' indicator function $\kappa^{\dagger}=\I_{D_{1}^{\dagger}}$ displayed in Figure \Fref{Fig1} (top). The observation operator $\O=(\O_{1},\dots,\O_{64})$ is defined in terms of 64 mollified Dirac deltas $\{\O_{j}\}_{j=1}^{64}$ centered  at the measurement locations display as white squares along the boundary of the domain in \Fref{Fig1} (top).  Each coordinate of the data is computed by means of (\ref{eq:source2}) with $p$ from the solution of the PDE with the aforementioned true source term and by adding Gaussian noise $\eta_{j}$ with standard deviation of $10\%$ of the size of the noise-free measurements (i.e. of $\O_{j}(\frac{\partial p}{\partial \nu})$). We reiterate that, in order to avoid inverse crimes \cite{KS05}, we use a coarser grid of size $80\times 80$ for the inversion via the MCMC method (Algorithm \ref{MCMC}). The parameterization of $D_{1}$ in terms of the level set function is given by $\ D_1 = \{x\in D\ |  u(x) < 0\}$ (i.e. by simply choosing $c_{0}=-\infty$ and $c_{1}=0$ in (\ref{eq_lss})). 

For this example we consider a prior covariance $\C$ of the form presented in (\ref{eq:cova2}) for some choices of $L$ in the correlation function. We construct $\C$ directly from this correlation function and then we conduct the eigendecomposition needed for the KL expansion and thus for sampling the prior. In \Fref{Fig2} we display samples from the prior $\N(0,\C)$ on the level set function $u$ (first, third and fifth rows) and the corresponding indicator function $\kappa=\I_{D_{1}}$ (second, fourth and sixth rows) for (from left to right) $L=0.1, 0.15, 0.2, 0.3, 0.4$. Different values of $L$ in (\ref{eq:cova2}) clearly result in substantial differences in the spatial correlation of the zero level set associated to the samples of the level set function. The spatial correlation of the 
zero level set funtion, under the prior, has significant effect on $\I_{D_{1}}$ which we use as the right-hand side (RHS) in problem (\ref{eq:source}) and whose solution, via expression  (\ref{eq:source2}), determines the likelihood (\ref{Phi}). It then comes as no surprise that the posterior measure on the level set is also strong;y dependent on the choice of the prior via the parameter $L$. 
We explore this effect in the following paragraphs.

In \Fref{Fig3} we present the numerical results from different MCMC chains computed with different priors corresponding to the aforementioned choices of $L$. The MCMC mean of the level set function is displayed in the top row of \Fref{Fig3} for the choices (from left to right) $L=0.1, 0.15, 0.2, 0.3, 0.4$. We reiterate that although the MCMC method provides the characterization of the posterior of the level set function, our primary aim is to identify the field $\kappa(x)=\I_{D_{1}}(x)$ that determines the RHS of (\ref{eq:source}) by means of conditioning the prior $\N(0,\C)$ to noisy data from (\ref{eq:source2}). A straightforward estimate of such field can be obtained by mapping, via the level set map (\ref{eq_lsf}), the posterior mean level set function denoted by $\overline{u}$ into the corresponding field $F(\overline{u}(x))=\overline{\kappa}(x)=\I_{\overline{D}_{1}}(x)$ where $\overline{D}_1 = \{x\in D\ |  \overline{u}(x) < 0\}$. We display $\overline{\kappa}(x)=\I_{\overline{D}_{1}}(x)$ in the top-middle row of \Fref{Fig3} along with the plot of the true field $\kappa^{\dagger}=\I_{D_{1}^{\dagger}}$ (right column) for comparison. 

As mentioned earlier, we are additionally interested in the push-forward of the posterior measure of the level set function $u$ under the level set map (i.e. $(F^\ast \mu^y)(du)$). We characterize $F^\ast \mu^y$ by mapping under $F$ our MCMC samples from $\mu^{y}$. In \Fref{Fig3} we present the mean (bottom-middle) and the variance (bottom) of $F^\ast \mu^y$. \Fref{Fig4} shows some posterior (MCMC) samples $u$ of the level set function (first, third and fifth rows) and the corresponding level set map $F(u)=\I_{D_{1}}$ with $\ D_1 = \{x\in D\ |  u(x) < 0\}$ associated to these posterior samples (second, fourth and sixth rows).

The push-forward of the posterior measure under the level set map (i.e. $F^\ast \mu^y$) thus provides a probabilistic description of the inverse problem of identifying the true $\kappa^{\dagger}=\I_{D_{1}^{\dagger}}$. We can see from \Fref{Fig3} that, for some choices of $L$, the mean of $F^\ast \mu^y$ provides reasonable estimates of the truth. However, the main advantage of the Bayesian approach proposed here is that a measure of the uncertainty of such estimate is also obtained from $F^\ast \mu^y$. The variance  (\Fref{Fig3} bottom), for example, is a measure of the uncertainty in the location of the interface between the two regions $D$ and $D\setminus D_{1}$. 

The results displayed in \Fref{Fig3} show the strong effect that the selection of the prior has on the posterior measure $\mu^{y}$ and the corresponding 
pushforward measure $F^\ast \mu^y$. In particular, there seems to be a critical value $L=0.2$ above of which the corresponding posterior mean on $F^\ast\mu^{y}$ provides a reasonable identification of the true $\I_{D_{1}^{\dagger}}$ with relatively small variance. This critical value seems to be related to the size and the shape of the inclusions that determines the true region $D_{1}^{\dagger}$ (\Fref{Fig1} (top)). It is intuitive that posterior samples that result from very small spatial correlation cannot easily characterize these inclusions accurately
unless the data is overwhelmingly informative. The lack of a proper characterization of the geometry from priors associated with small $L$ is also reflected with larger variance around the interface. It is then clear that the capability of the proposed level set Bayesian framework to properly identify a shape $D_{1}^{\dagger}$ (or alternatively its indicator function $\I_{D_{1}^{\dagger}}$) depends on properly incorporating, via the prior measure,  a priori information on the regularity and spatial correlation of the unknown geometry of $D_{1}^{\dagger}$.

Since the selection of the prior has such a clear effect on the posterior, it comes as no surprise that it also affects the efficiency of the MCMC method as we now discuss. In the bottom-right panel of \Fref{Fig1} we show the autocorrelation function (ACF) of the first KL mode of the level set function from different MCMC chains with different priors corresponding to our different choices of correlation length $L$ in (\ref{eq:cova2}). The tunable parameters in the pCN-MCMC method are fixed for these experiments. We recall from \Fref{Fig3} that larger values of $L$ result in a mean level set whose corresponding indicator function better captures the spatial structures form the truth and with smaller variance around the interface. However, the larger the value of $L$ the slower the decay of the ACF. From these ACF plots, we note that even for the apparent optimal value of $L=0.3$, our MCMC method produces samples that are highly correlated and thus very long chains may be needed in order to produce a reasonable number of uncorrelated samples needed for statistical analysis. For this particular choice of $L=0.3$ we have conducted 50 multiple MCMC chains of length $10^6$ (after burn-in period) initialized from random samples from the prior. In \Fref{Fig1} (bottom-left) we show the Potential Scale Reduction Factor (PSRF)  \cite{Gelman}
computed from MCMC samples of the level set function (red-solid line) and the corresponding samples under $F$ (i.e. the $\I_{D_{1}}$'s) (blue-dotted line). We observe that the PSRF goes below 1.1 after  (often taken as an approximate 
indication of convergence  \cite{Gelman}); thus the
Gelman-Rubin diagnostic \cite{Gelman} based on the PSRF is passed for this selection of $L$. The generation of multiple independent MCMC chains that are statistically consistent opens up the possibility of using high-performance computing to enhance our capabilities of properly exploring the posterior. While we use a relatively small number of chains as a proof-of-concept, the MCMC chains are fully independent and so the computational cost of running multiple chains scales with the number of available processors.

The $5\times 10^7$ samples that we obtained from the $50$ MCMC chains are combined to provide a full characterization of the posterior $\mu^{y}$ on the level set and the corresponding push-forward $F^\ast \mu^y$ (i.e. The $\I_{D_{1}}$'s computed from $D_{1}$ with posterior samples $u$). We reemphasize that our aim is the statistical identification of $\I_{D_{1}}^{\dagger}$. Therefore, in order to obtain a quantity from the true $\I_{D_{1}^{\dagger}}$ against to which compare the computed push-forward of the level set posterior, we consider the Discrete Cosine Transform (DCT) of the true field $\I_{D}$. Other representations/expansions of the true field could be considered for the sake of assessing the uncertainty of our estimates with respect to the truth. In \Fref{Fig5}  we show the prior and posterior densities of the first DCT coefficients of $\I_{D_{1}}$ where $\ D_1 = \{x\in D\ |  u(x) < 0\}$ with $u$ from our MCMC samples (the vertical dotted line corresponds to the DCT coefficient of the true $\I_{D_{1}^{\dagger}}$). We can observe how the push forward posterior are concentrated around the true values. It is then clear how the data provide a strong conditioning on the first DCT coefficients of the discontinuous field that we aim at obtaining with our Bayesian level set approach.

\begin{figure}[htbp]
\begin{center}
\includegraphics[scale=0.4]{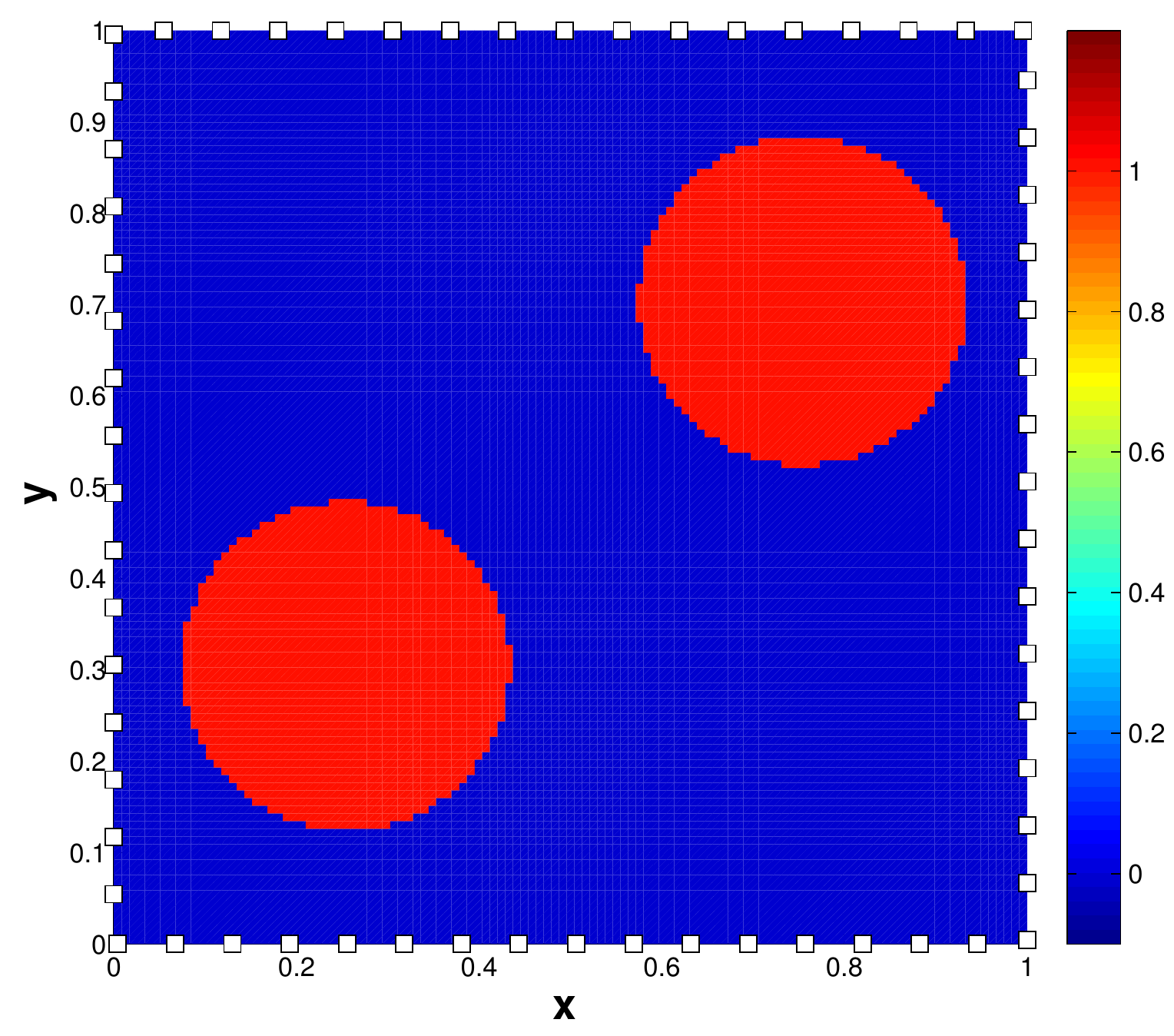}\\
\includegraphics[scale=0.4]{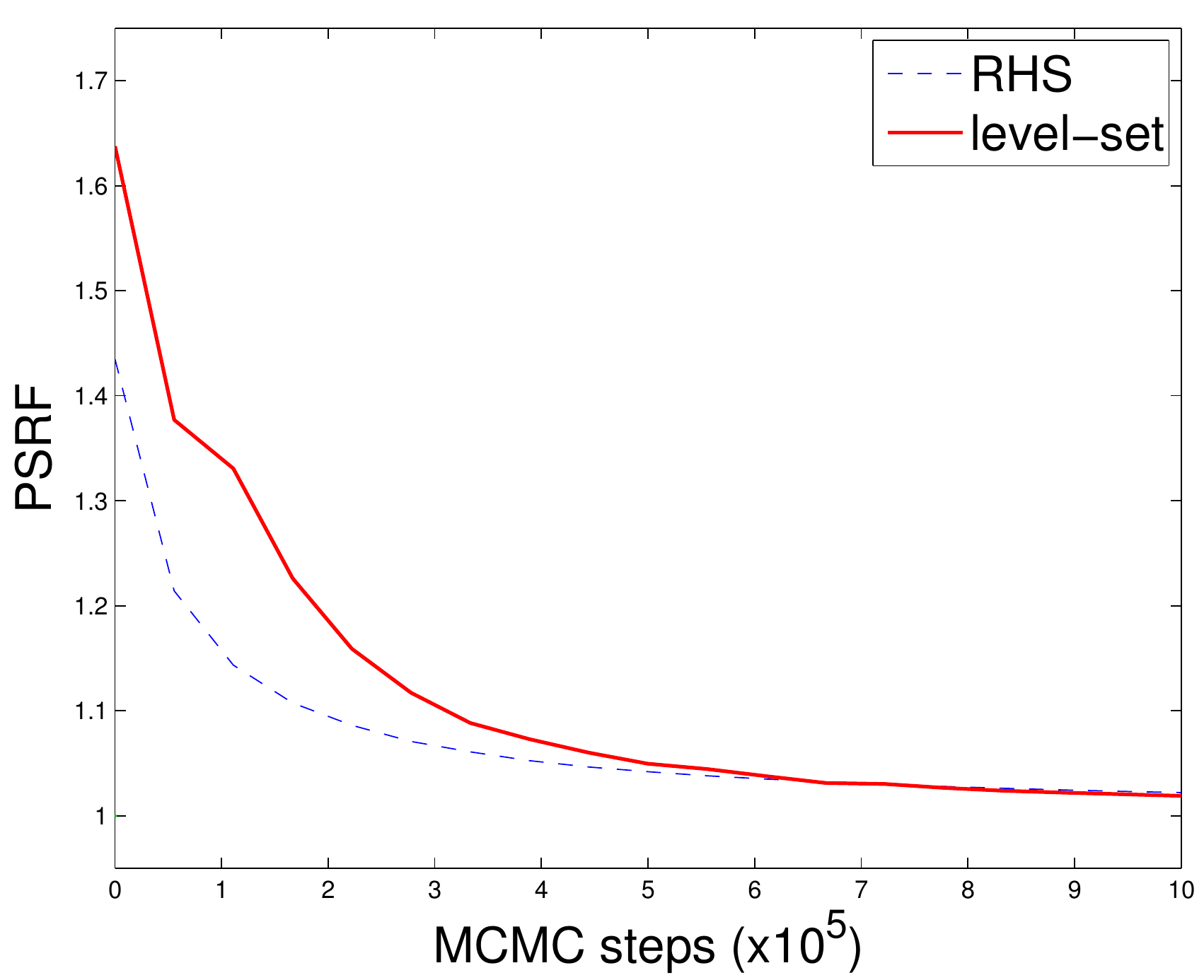}
\includegraphics[scale=0.4]{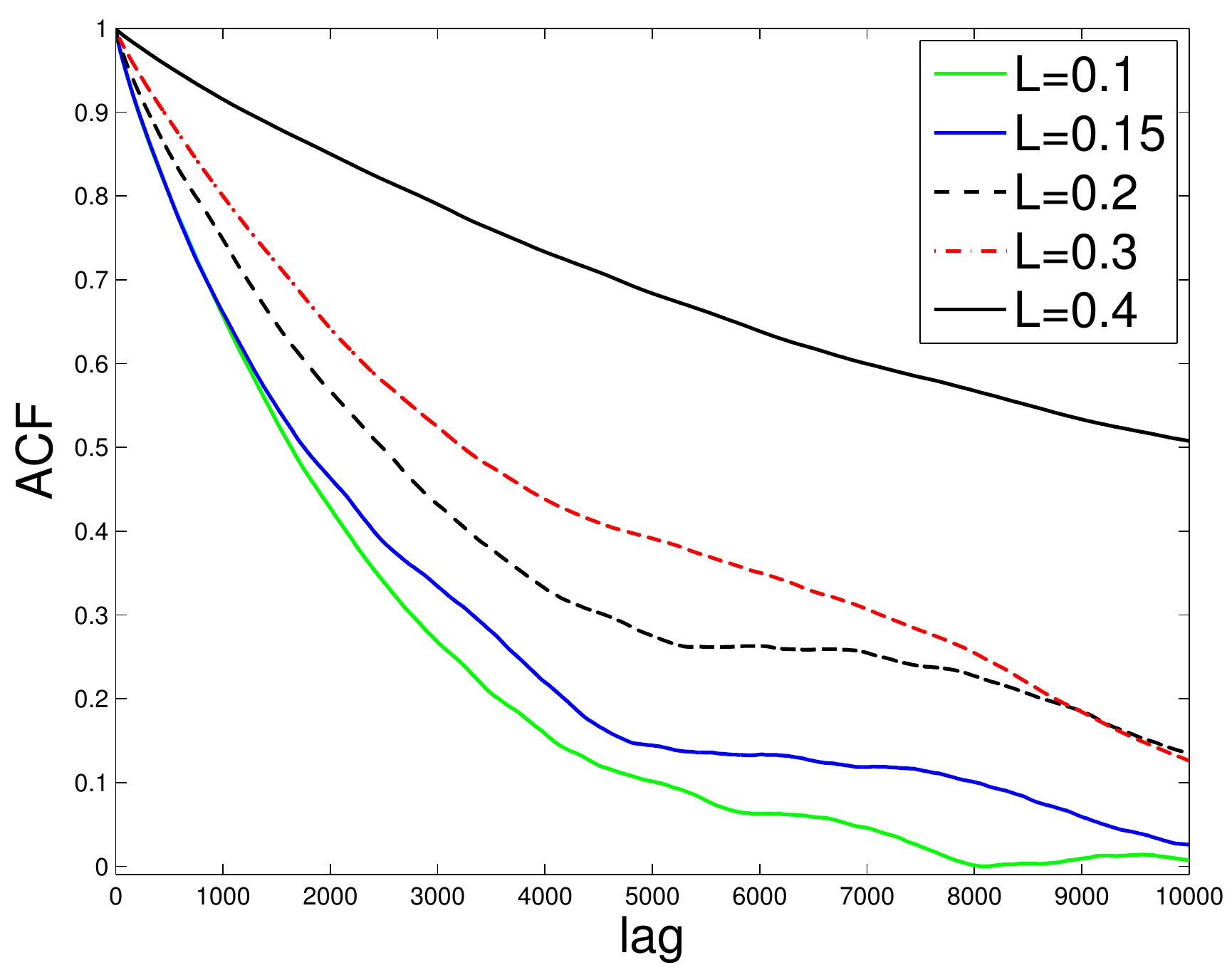}
 \caption{Inverse Potential. Top: True source term $\kappa^{\dagger}=\I_{D_{1}^{\dagger}}$ of eq. (\ref{eq:source}). Bottom-left: PSRF from multiple chains with $L=0.3$ in (\ref{eq:cova2}).  Bottom-right: ACF of first KL mode of the level set function from single-chain MCMC with different choices of $L$.}   \label{Fig1}
\end{center}
\end{figure}

\begin{figure}[htbp]
\begin{center}
  \includegraphics[scale=0.85]{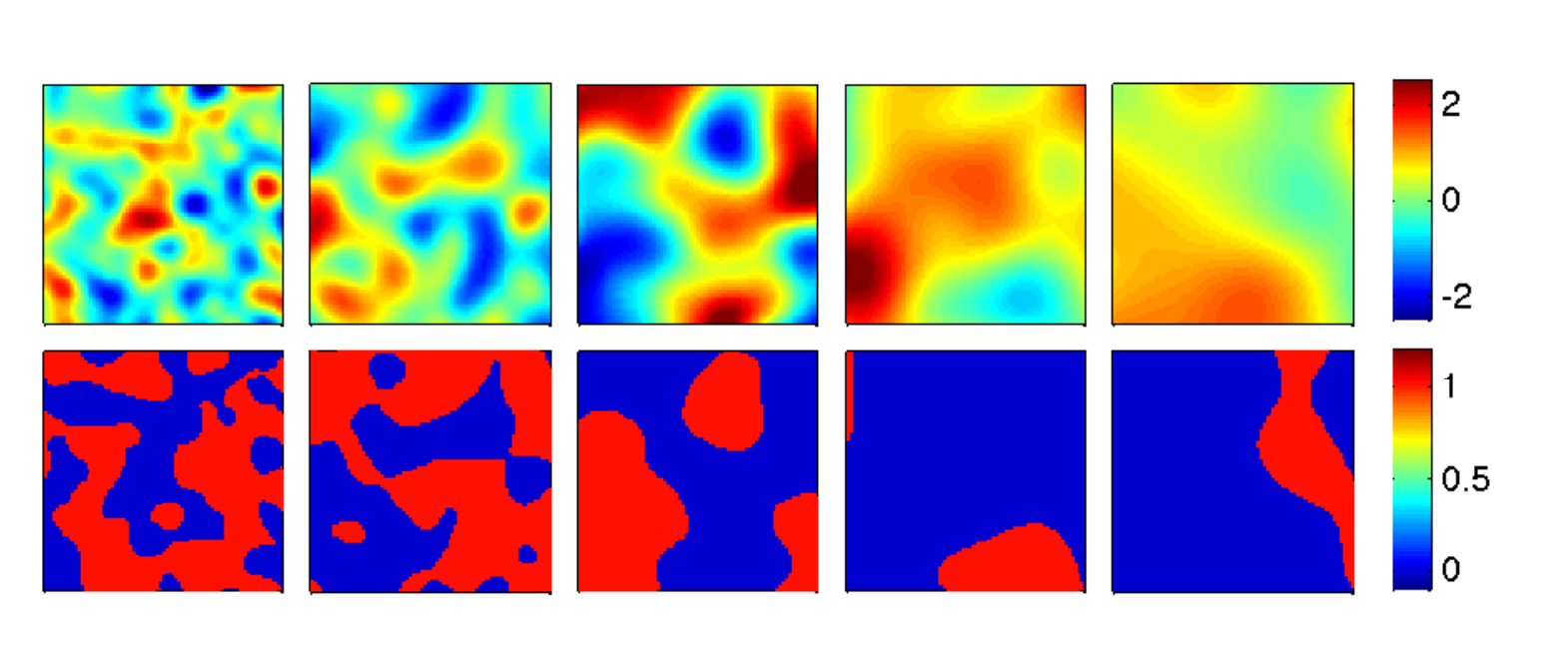}
\vskip-15pt
     \includegraphics[scale=0.85]{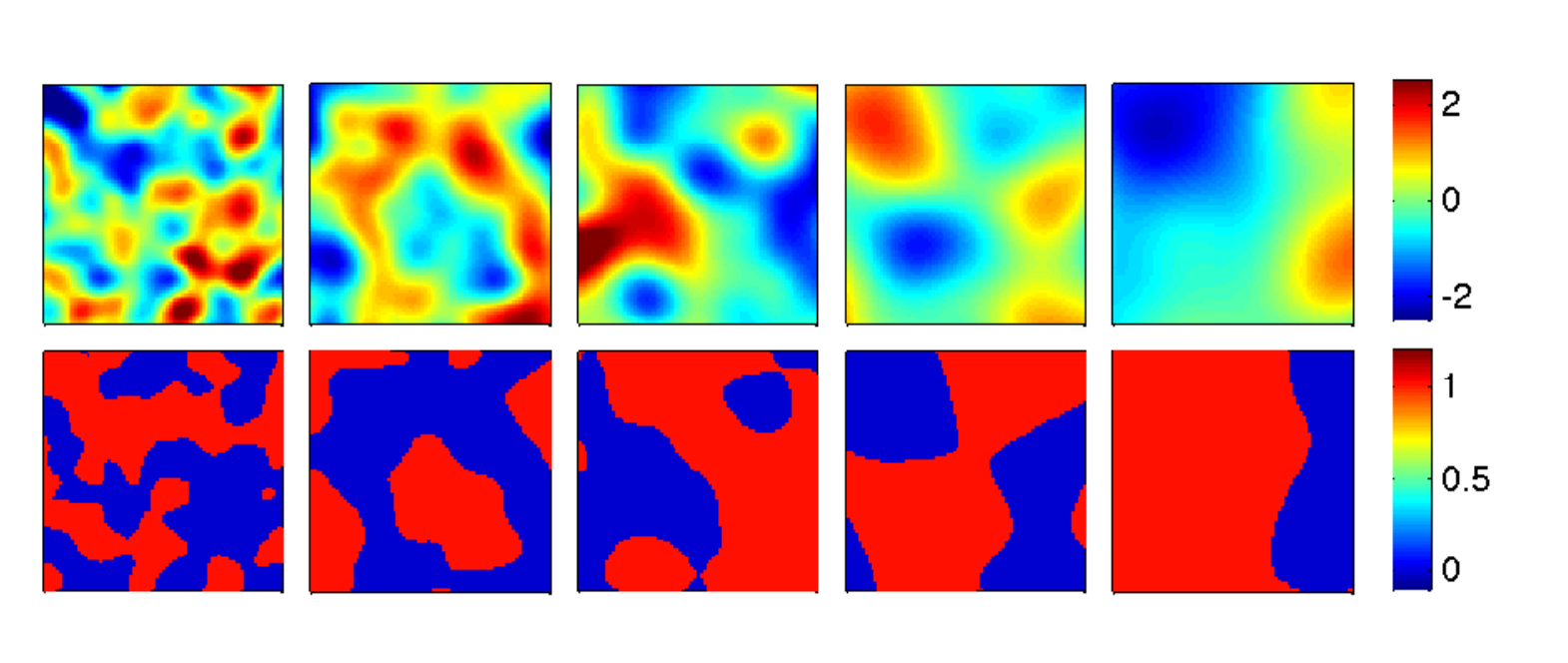}\\
\vskip-15pt
     \includegraphics[scale=0.85]{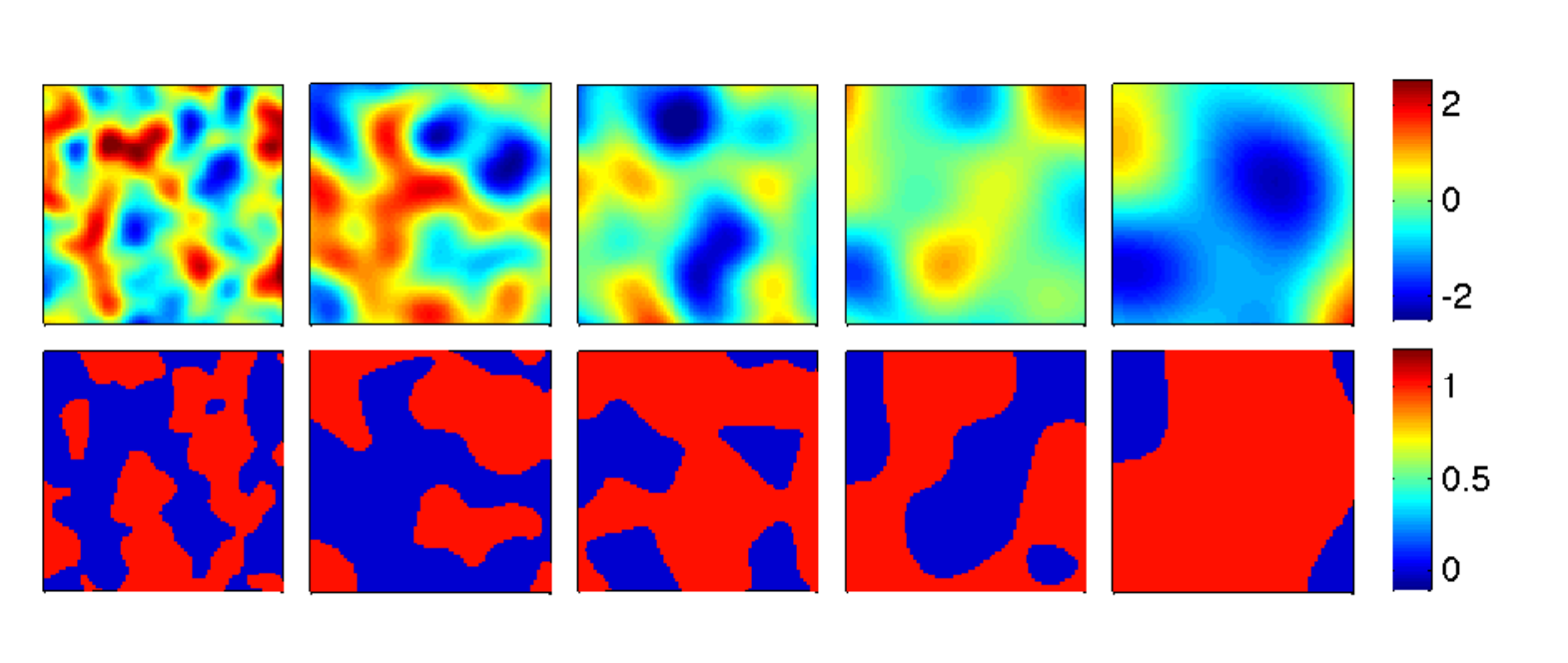}
 \caption{Inverse Potential. Samples from the prior on the level set function $u$ (first, third and fifth rows) for (from left to right) $L=0.1, 0.15, 0.2, 0.3, 0.4$. Corresponding $\I_{D_{1}}$ with $\ D_1 = \{x\in D\ |  u(x) < 0\}$ (second, fourth and sixth rows) associated to each of these samples from the level set function.}   \label{Fig2}
\end{center}
\end{figure}

\newpage

\begin{figure}[htbp]

\includegraphics[scale=1.1]{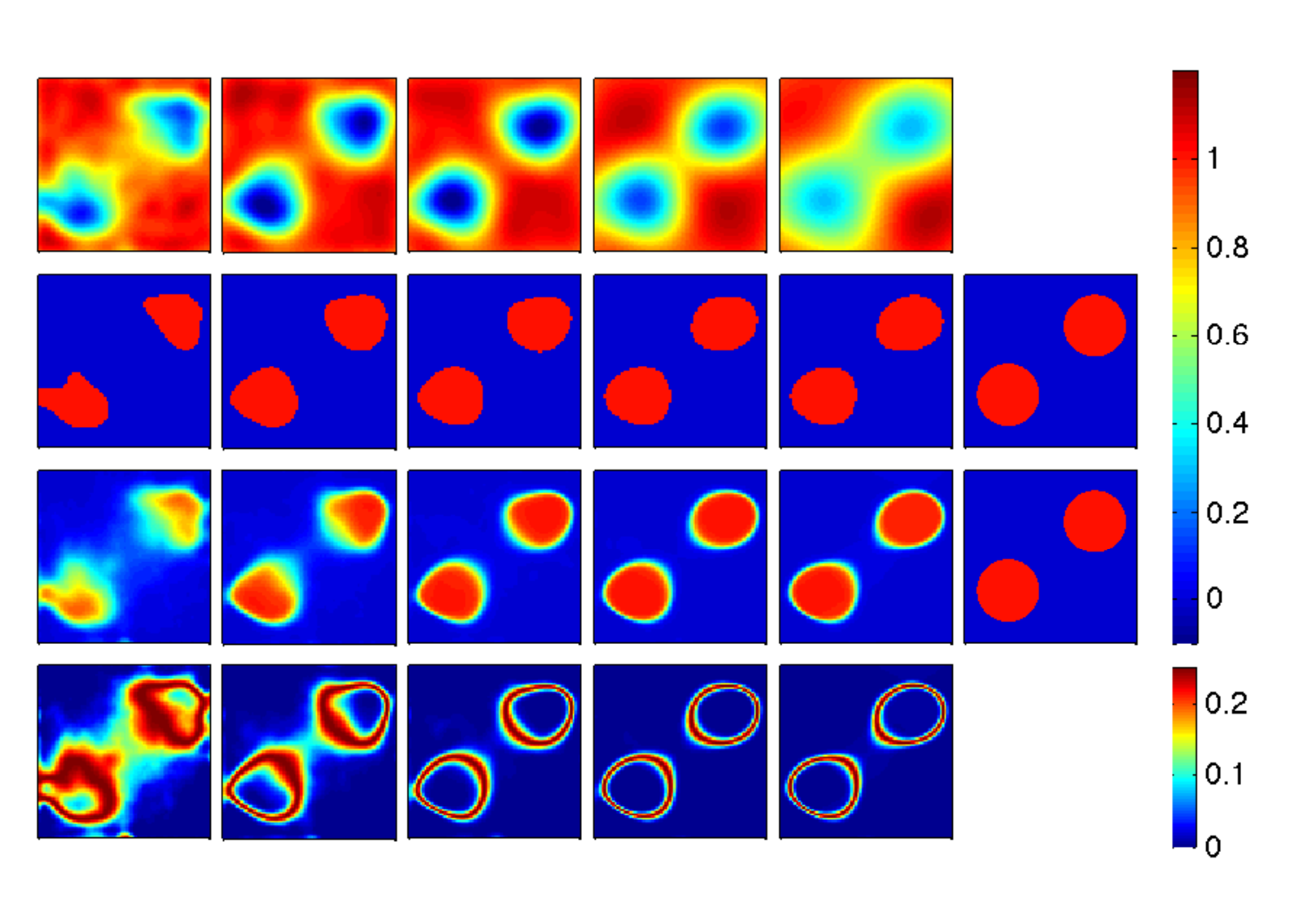}
 \caption{Inverse Potential. MCMC results for (from left to right) $L=0.1, 0.15, 0.2, 0.3, 0.4$  in the eq. (\ref{eq:cova2}). Top: Posterior mean level set function $\overline{u}$ (computed via MCMC). Top-middle: Plot of $\I_{\overline{D}_{1}}$ with $\overline{D}_1 = \{x\in D\ |  \overline{u}(x) < 0\}$ (the truth $\I_{D_{1}^{\dagger}}$ is presented in the right column). Bottom-middle: Mean of $\I_{D_{1}}$ where $\ D_1 = \{x\in D\ |  u(x) < 0\}$ and $u$'s are the posterior MCMC samples (the truth is presented in the right column). Bottom: Variance of $\I_{D_{1}}$ where $\ D_1 = \{x\in D\ |  u(x) < 0\}$ and $u$'s are the posterior MCMC samples }
   \label{Fig3}

\end{figure}

\begin{figure}[htbp]
\begin{center}
\includegraphics[scale=0.85]{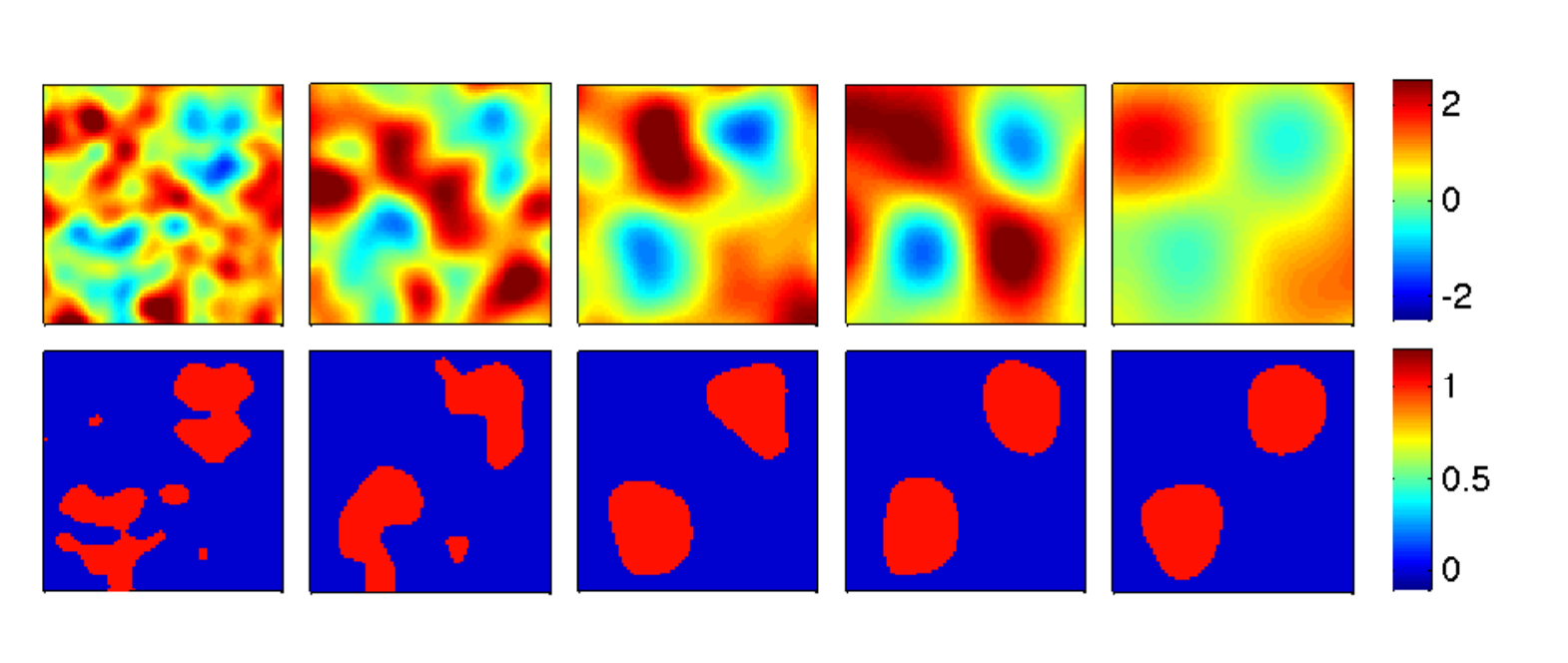}
\vskip-15pt
\includegraphics[scale=0.85]{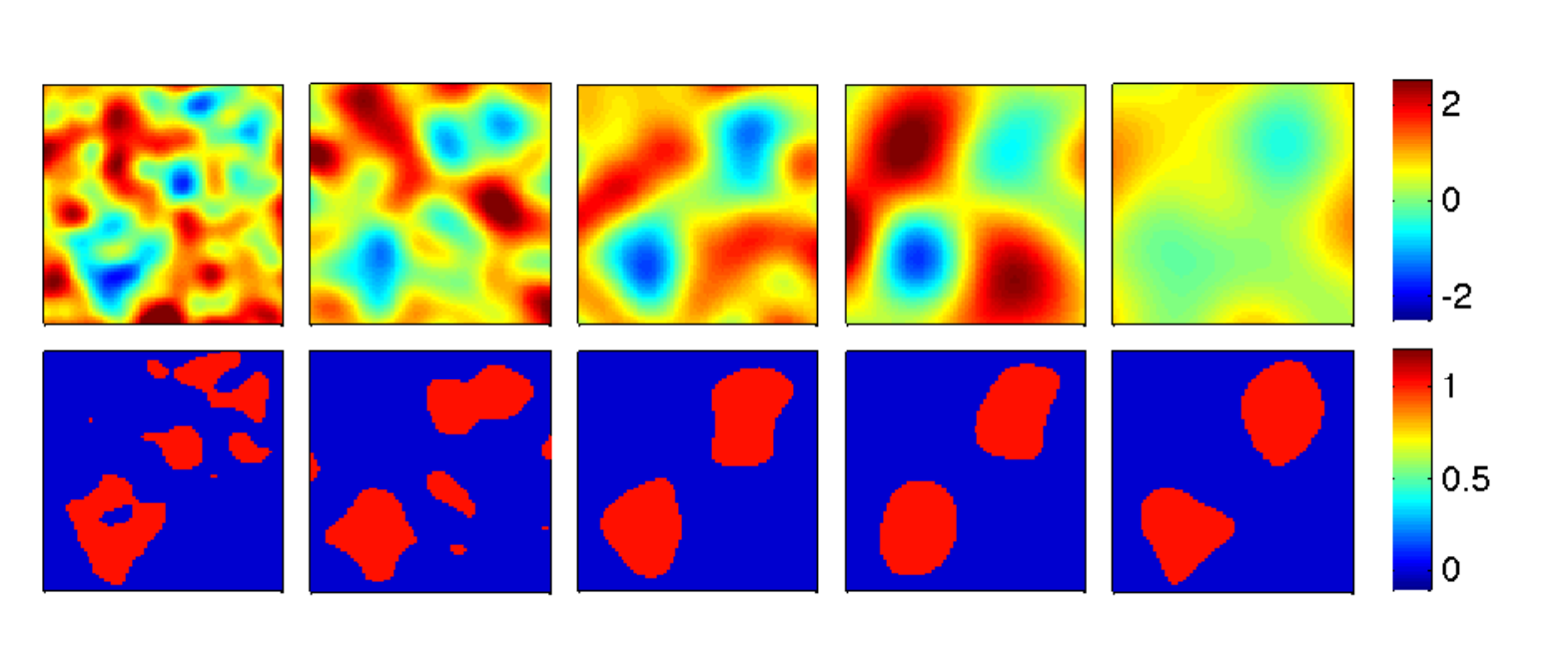}\\
\vskip-15pt
\includegraphics[scale=0.85]{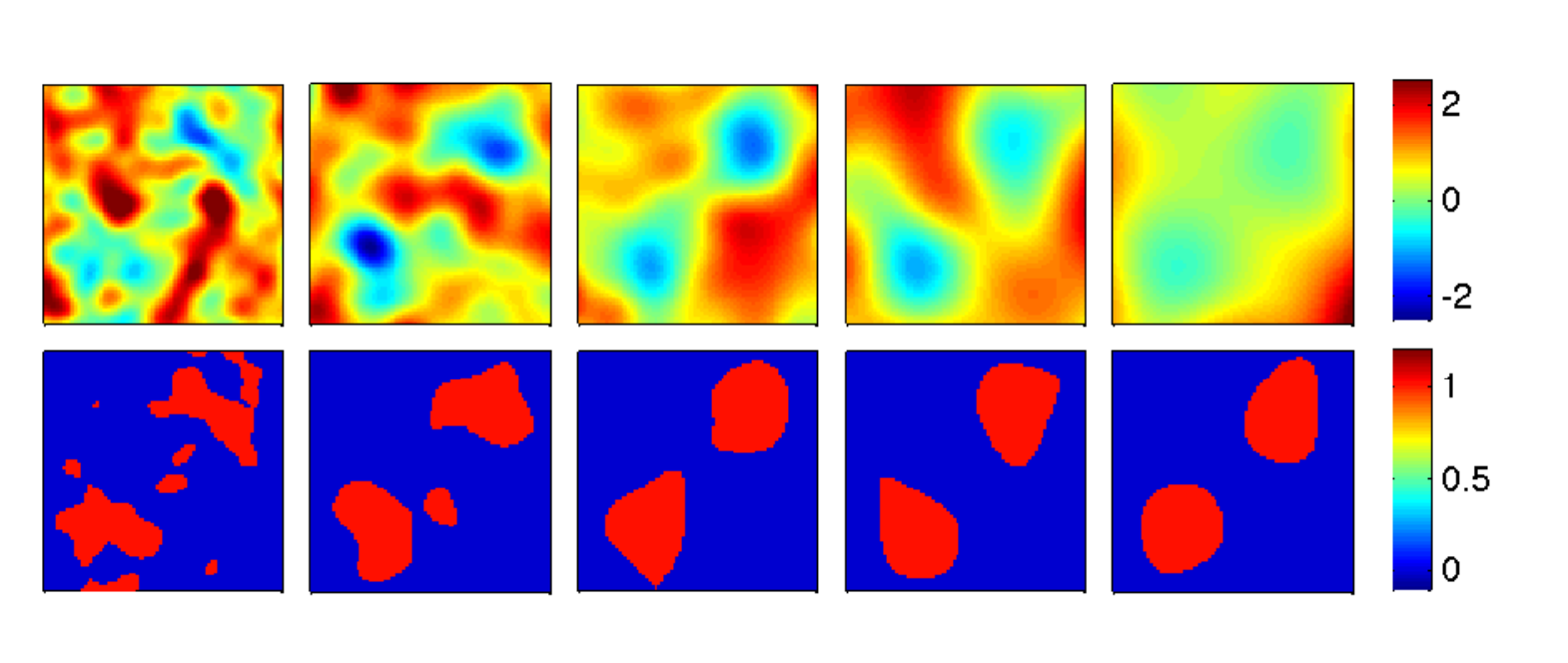}

 \caption{Inverse Potential. Samples from the posterior on the level set $u$ (first, third and fifth rows) for (from left to right) $L=0.1, 0.15, 0.2, 0.3, 0.4$. Corresponding $\I_{D_{1}}$ where $\ D_1 = \{x\in D\ |  u(x) < 0\}$ (second, fourth and sixth rows) associated to each of these samples from the level set function.}   \label{Fig4}
\end{center}
\end{figure}

\begin{figure}[htbp]
\begin{center}
\includegraphics[scale=0.25]{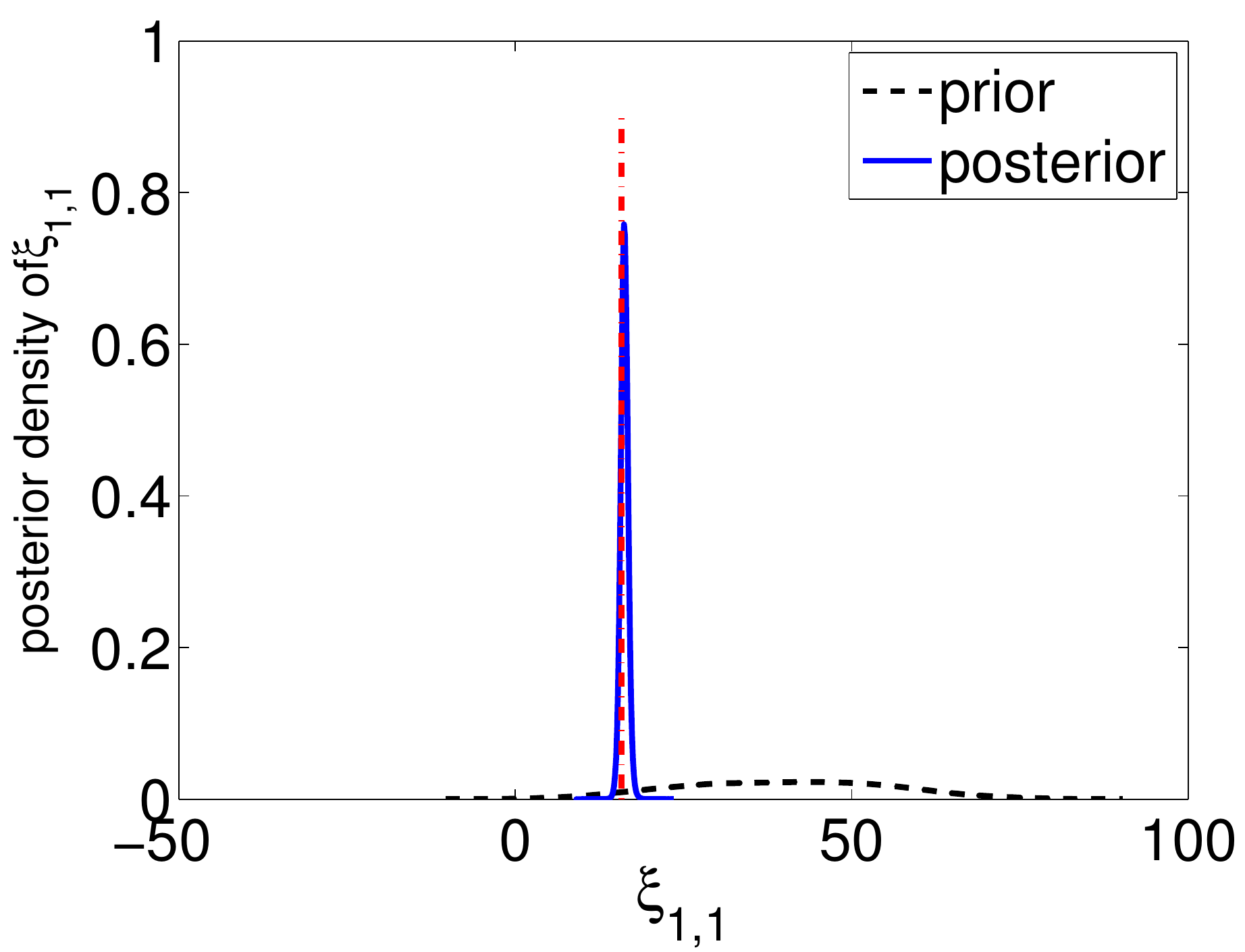}
\includegraphics[scale=0.25]{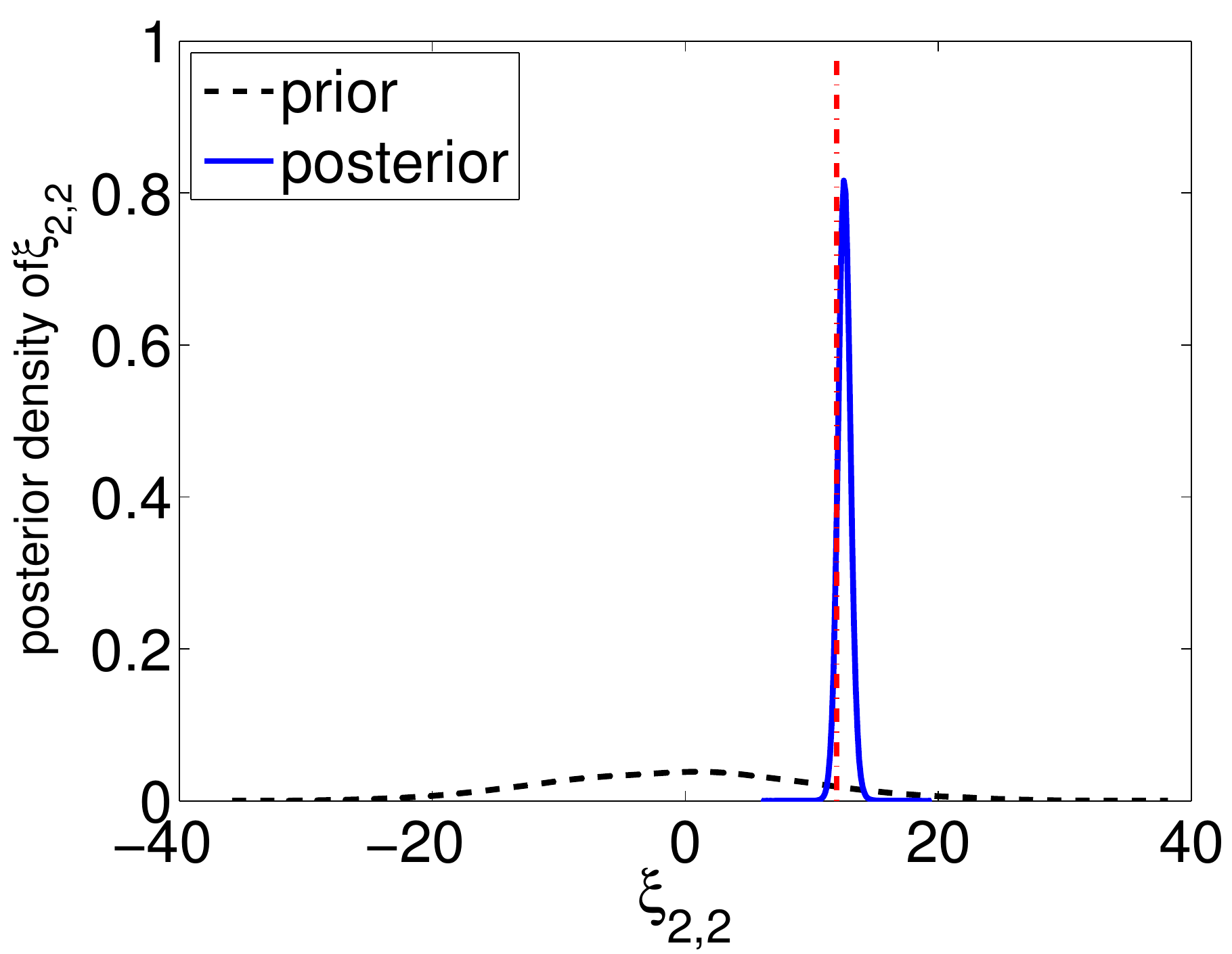}
\includegraphics[scale=0.25]{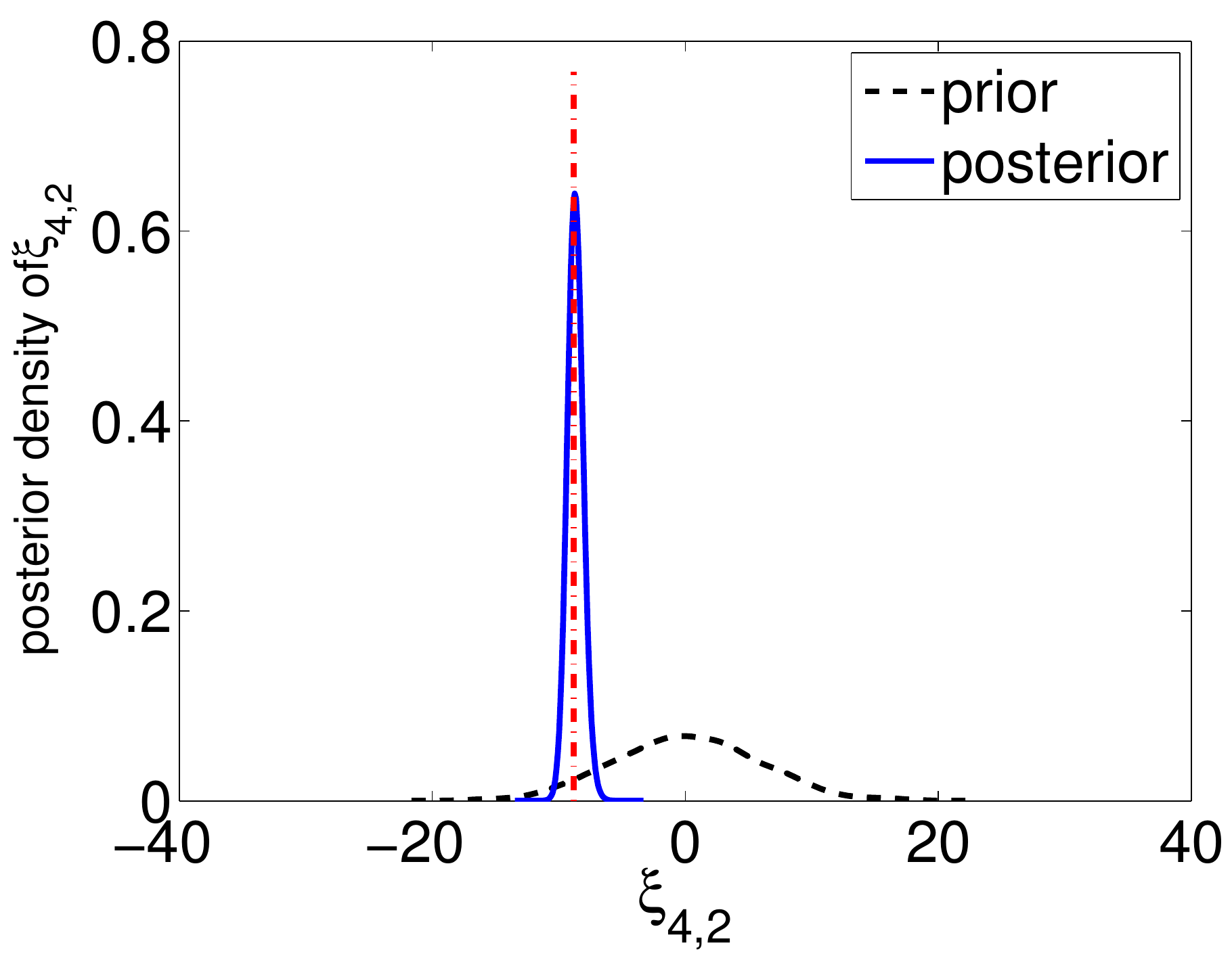}\\
\includegraphics[scale=0.25]{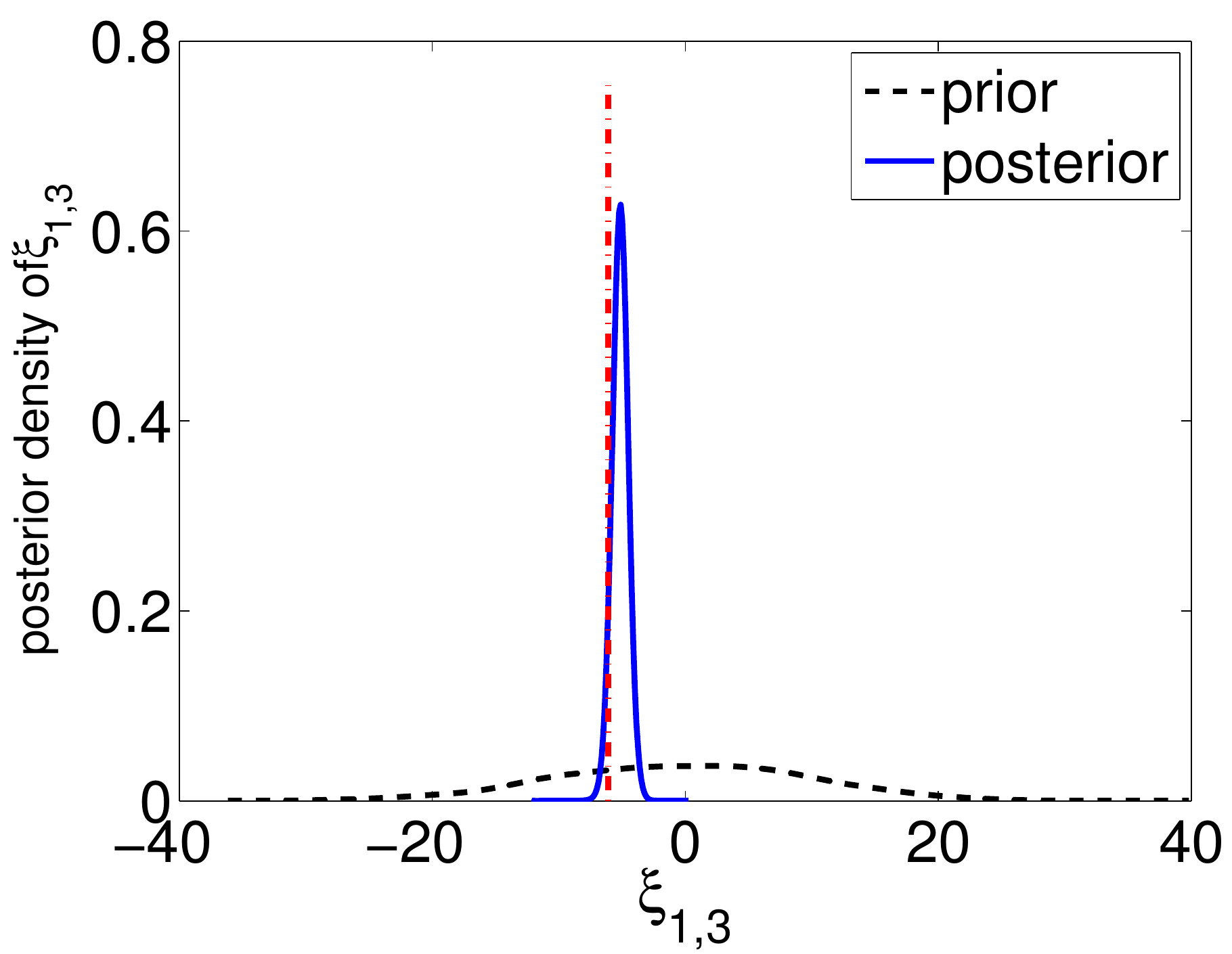}
\includegraphics[scale=0.25]{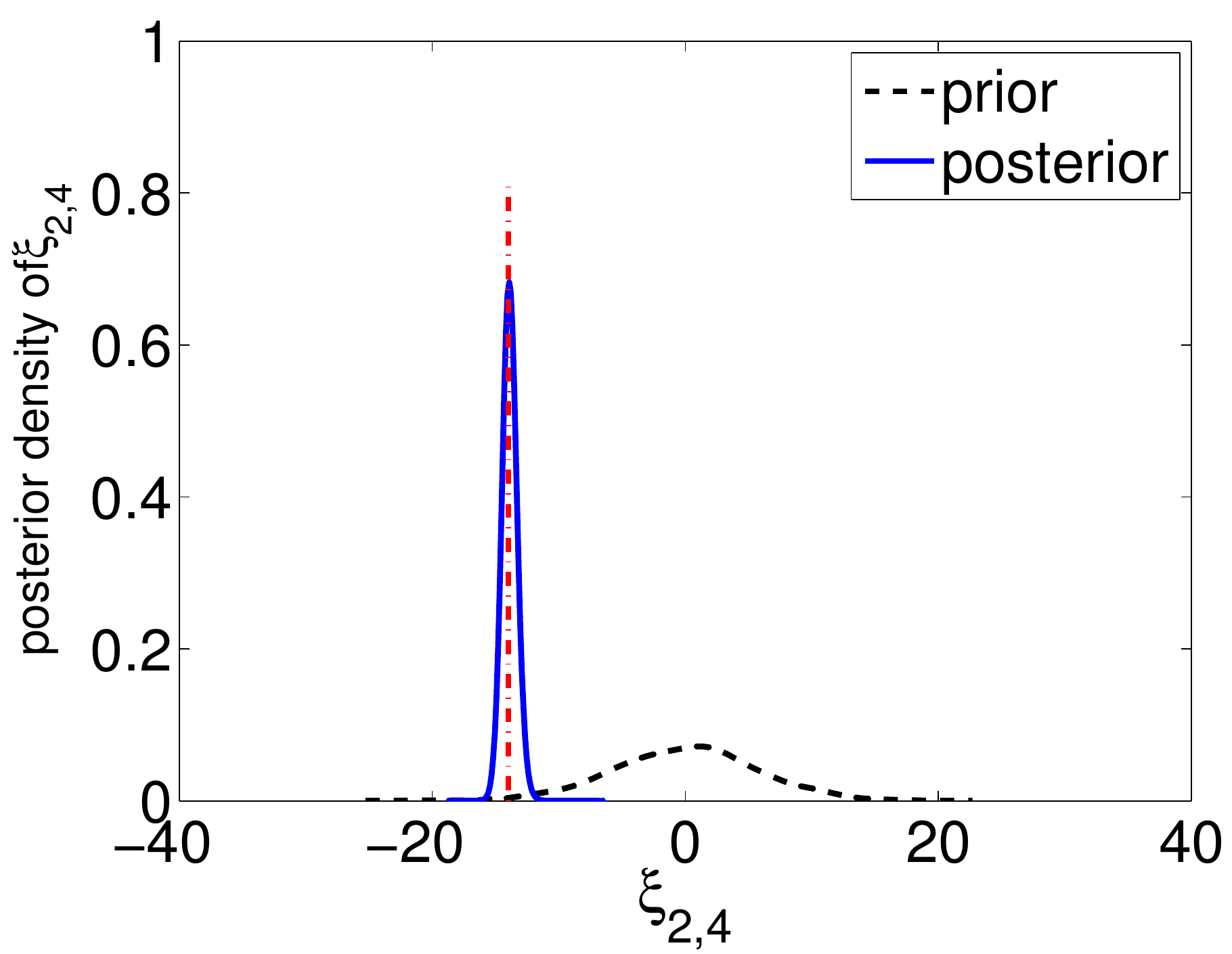}
\includegraphics[scale=0.25]{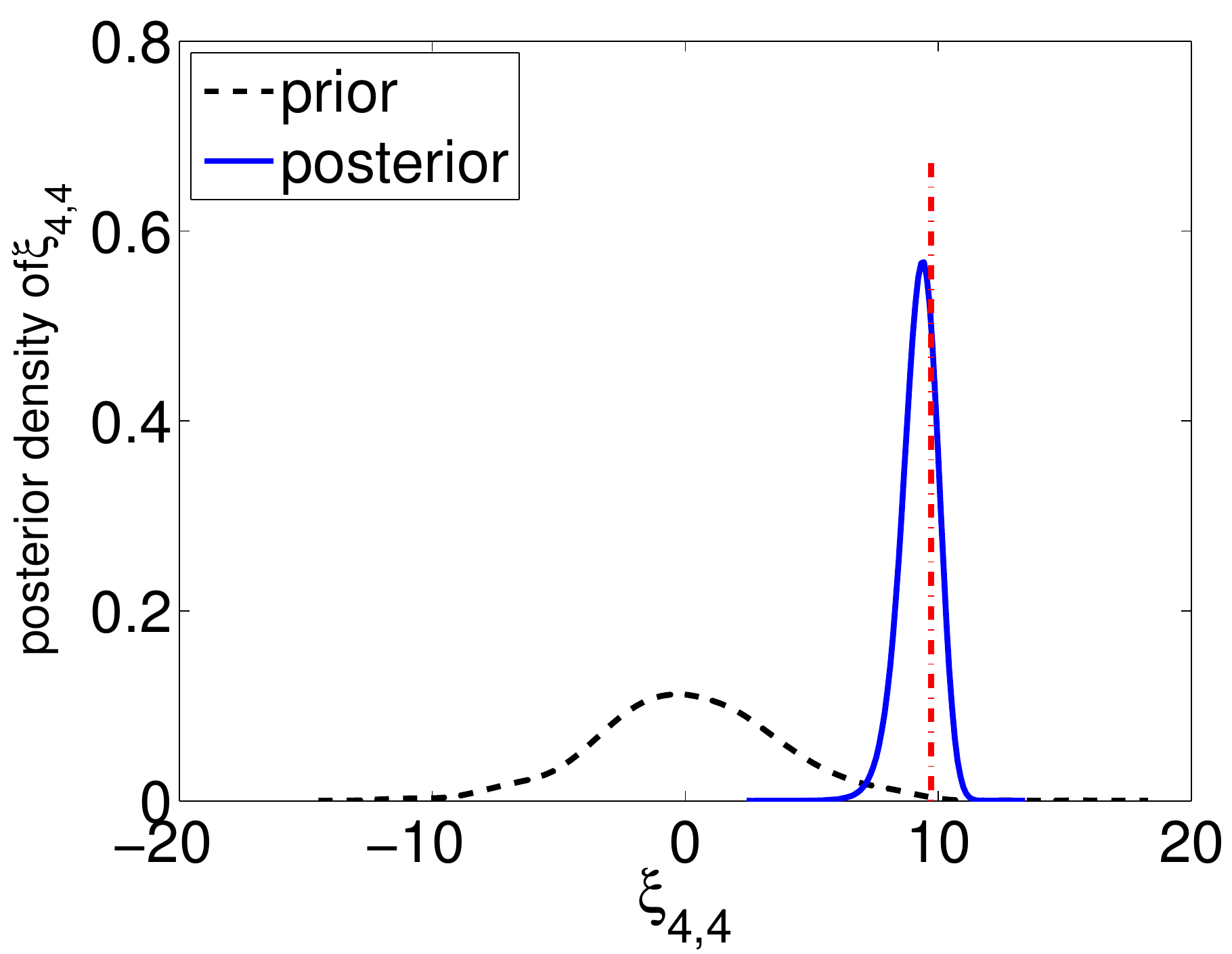}
 \caption{Inverse Potential. Densities of the prior and posterior of various DCT coefficients of the $\I_{D_{1}}$ where $\ D_1 = \{x\in D\ |  u(x) < 0\}$ obtained from MCMC samples on the level set $u$ for $L=0.3$ (vertical dotted line indicates the truth).}   \label{Fig5}
\end{center}

\end{figure}

\clearpage
\subsection{Structural Geology: Channel Model}

In this section we consider the inverse problem discussed in subsection \ref{ssec:tm2}. We consider the Darcy model (\ref{sdarcy}) but with a more realistic set of boundary conditions that consist of a mixed Neumman and Dirichlet conditions. For the concrete set of boundary conditions as well as the right-hand-side we use for the present example we refer the reader to \cite[Section 4]{EnKF_US}. This flow model, initially used in the seminal paper of \cite{Carrera}, has been used as a benchmark for inverse problems in \cite{EnKF_US,Hanke,IglesiasDawson2}. While the mathematical analysis of subsection is \ref{ssec:tm2} conducted on a model with Dirichlet boundary conditions, in order to streamline the presentation, the corresponding extension the case of mixed  boundary conditions can be carried out with similar techniques.

We recall that the aim is to estimate the interface between regions $D_{i}$ of different structural geology which result in a discontinuous permeability $\kappa$ of the form (\ref{eq:kd}). In order to generate synthetic data, we consider a true  $\kappa^{\dagger}(x)=\sum_{i=1}^{3}\kappa_{i}\I_{D_{i}^{\dagger}}$ with prescribed (known) values of $\kappa_{1}=7$, $\kappa_{2}=50$ and $\kappa_{3}=500$. This permeability, whose plot is displayed in \Fref{Fig6} (top), is used in (\ref{sdarcy}) to generate synthetic data collected from interior measurement locations (white squares in \Fref{Fig6}). The estimation of $\kappa$ is conducted given observations of the solution of the Darcy model (\ref{sdarcy}). To be
concrete, the observation operator $\O=(\O_{1},\dots,\O_{25})$ is defined in terms of 25 mollified Dirac deltas $\{\O_{j}\}_{j=1}^{25}$ centered  at the aforementioned measurement locations and acting on the solution $p$ of the Darcy flow model. For the generation of synthetic data we use a grid of $160\times 160$ which, in order to avoid inverse crimes \cite{KS05}, is finer than the one used for the inversion ($80\times 80$). As before, observations are corrupted with Gaussian noise proportional to the size of the noise-free observations ($\O_{j}(p)$ in this case).

For the estimation of $\kappa$ with the proposed Bayesian framework we assume that knowledge of three nested regions is available with the permeability values $\{\kappa_{i}\}_{i=1}^{3}$ that we use to define the true $\kappa^{\dagger}$. Again, we are interested in the realistic case where the rock types of the formation are known from geologic data but the location of the interface between these rocks is uncertain. In other words, the unknowns are the geologic facies $D_{i}$ that we parameterize in terms of a level set function, i.e. $ D_i = \{x\in D\ |  c_{i-1}\leq u(x) < c_{i}\}$ with $c_{0}=-\infty$, $c_{1}=0$, $c_{2}=1$, $c_{3}=\infty$. Similar to the previous example, we use a prior of the form (\ref{eq:cova2}) for the level set function. In  \Fref{Fig7} we display samples from the prior on the level set function (first, third and fifth rows) and the corresponding permeability mapping under the level set map (\ref{eq_lsf}) $F(u)(x)=\kappa(x)=\sum_{i=1}^{3}\kappa_{i}\I_{D_{i}}$ (second, fourth and sixth rows) for (from left to right) $L=0.2, 0.3, 0.35, 0.4, 0.5$. As before, we note that the spatial correlation of the covariance function has a significant effect on the spatial correlation of the interface between the regions that define the interface between the geologic facies (alternatively, the discontinuities of $\kappa$). Longer values of $L$ provide $\kappa$'s that seem more visually consistent with the truth. The results from \Fref{Fig8} show MCMC results from experiments with different priors corresponding to the aforementioned choices of $L$. The posterior mean level set function $\overline{u}$ is displayed in the top row of \Fref{Fig8}. The corresponding mapping under the level set function $\overline{\kappa}\equiv \sum_{i=1}^{3}\kappa_{i}\I_{\overline{D}_{i}}$ (with  $ \overline{D}_i = \{x\in D\ |  c_{i-1}\leq \overline{u}(x) < c_{i}\}$) is shown in the top-middle.

Similar to our discussion of the preceding subsection, for the present example we are interested in the push-forward of the posterior $\mu^{y}$ under the level set map $F$. More precisely, $F^\ast \mu^y$ provides a probability description of the solution to the inverse problem of finding the permeability given observations from the Darcy flow model. In \Fref{Fig8} we present the mean (bottom-middle) and the variance (bottom) of $F(\mu^{y})$ characterized by posterior samples on the level set function mapped under $F$. In other words, these are the mean and variance from the $\kappa$'s obtained from the MCMC samples of the level /
set function. As in the previous example, there is a critical value of $L=0.3$ below of which the posterior estimates cannot accurately identify the main spatial features of $\kappa^{\dagger}$. \Fref{Fig9} shows posterior samples of the level set function (first, third and fifth rows) and the corresponding $\kappa$ (second, fourth and sixth rows). The posterior samples, for values of $L$ above the critical value $L=0.3$, capture the main spatial features from the truth. There is, however, substantial uncertainty in the location of the interfaces. Our results offer evidence that this uncertainty can be properly captured with our level set Bayesian framework. Statistical measures of $F^\ast \mu^y$ (i.e. the posterior permeability measure on $\kappa$) is essential in practice. The proper quantification of the uncertainty in the unknown geologic facies is vital for the proper assessment of the environmental impact in applications such as CO$_2$ capture and storage, nuclear waste disposal and enhanced oil recovery. 

In \Fref{Fig6} (bottom-right) we show the ACF of the first KL mode of level set function from different MCMC chains corresponding to different priors defined by the choices of $L$ indicated previously. In contrast to the previous example, here we cannot appreciate substantial differences in the efficiency of the chain with respect to the selected values of $L$. However, we note that ACF exhibits a slow decay and thus long chains and/or multiple chains are need to properly explore the posterior. For the choice of $L=0.4$ we consider $50$ multiple MCMC chains. Our MCMC chains pass the Gelman-Rubin test \cite{Gelman}
as we can note from \Fref{Fig6} (bottom-left) where we show the PSRF computed from MCMC samples of the level set function $u$ (red-solid line) and the corresponding mapping, under the level set map, into the permeabilities $\kappa$ (blue-dotted line). As indicated earlier, we may potentially increase the number of multiple chains and thus the number of uncorrelated samples form the posterior. 

Finally, in \Fref{Fig10}  we show the prior and posterior densities of the first DCT coefficients on the $\kappa$ obtained from the MCMC samples of the level set function (the vertical dotted line corresponds to the DCT coefficient of the truth $\kappa^{\dagger}$). For some of these modes we clearly see that the posterior is concentrated around the truth. However, for the mode $\xi_{4,4}$ we note that the posterior is quite close to the prior indicating that the data have not informed this mode in any significant way.

\begin{figure}[htbp]
\begin{center}
\includegraphics[scale=0.35]{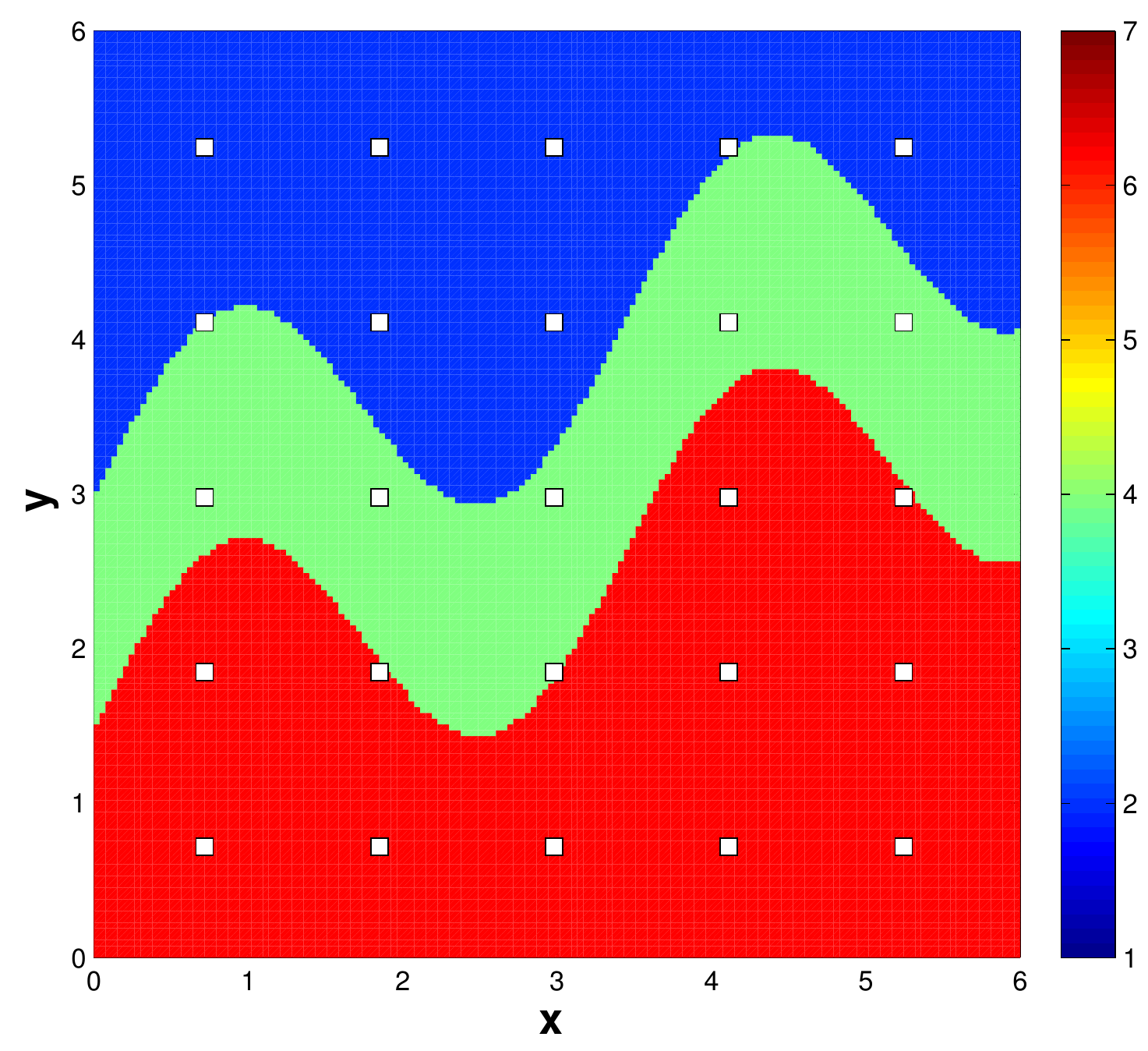}\\
\includegraphics[scale=0.3]{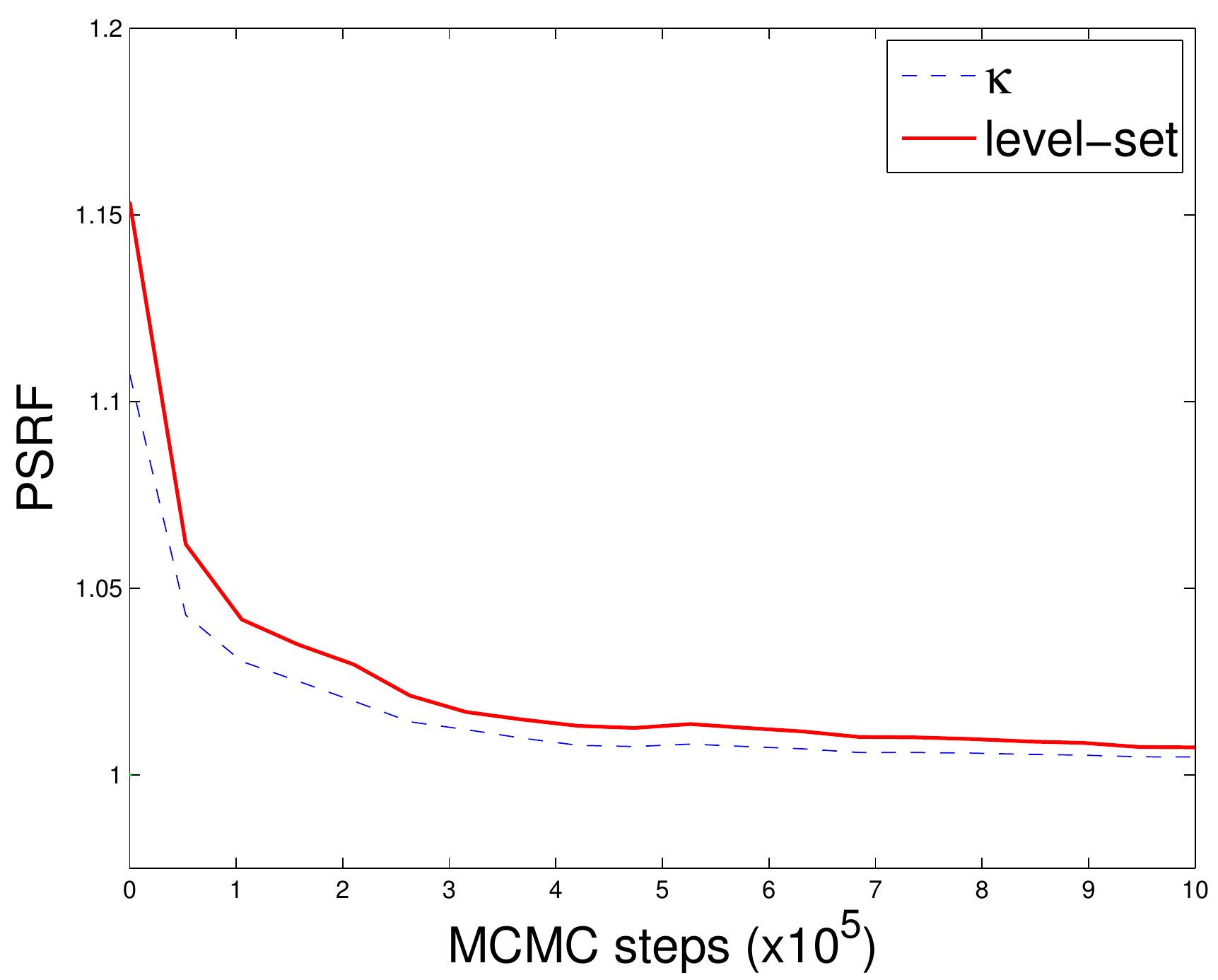}
\includegraphics[scale=0.3]{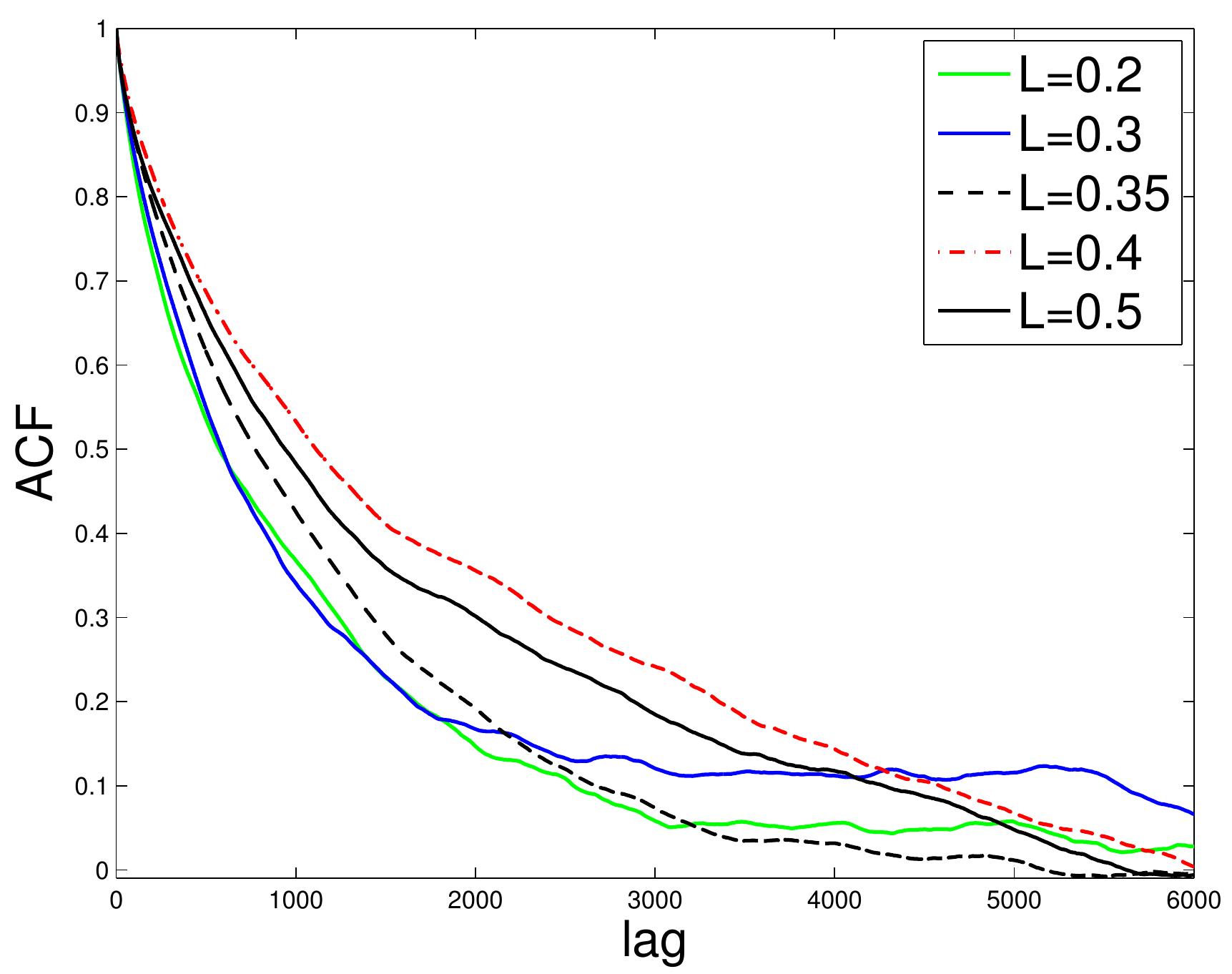}
 \caption{Identification of structural geology (channel model). Top: True $\kappa$ in eq. (\ref{sdarcy}). Bottom-left: PSRF from multiple chains with $L=0.4$ in (\ref{eq:cova1}).  Bottom-right: ACF of first KL mode of the level 
set function from single-chain MCMC with different choices of $L$.}   \label{Fig6}

\end{center}
\end{figure}

\begin{figure}[htbp]
\begin{center}
  \includegraphics[scale=0.85]{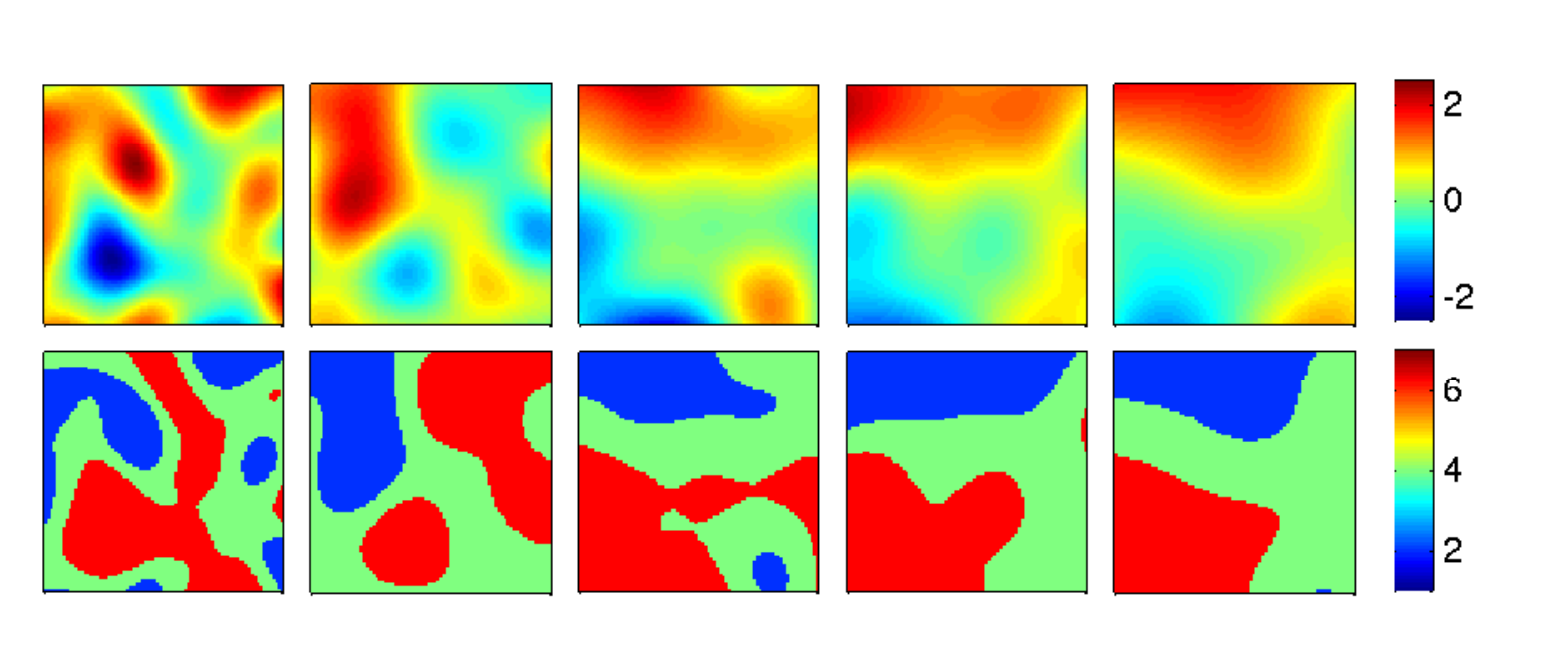}
\vskip-15pt
     \includegraphics[scale=0.85]{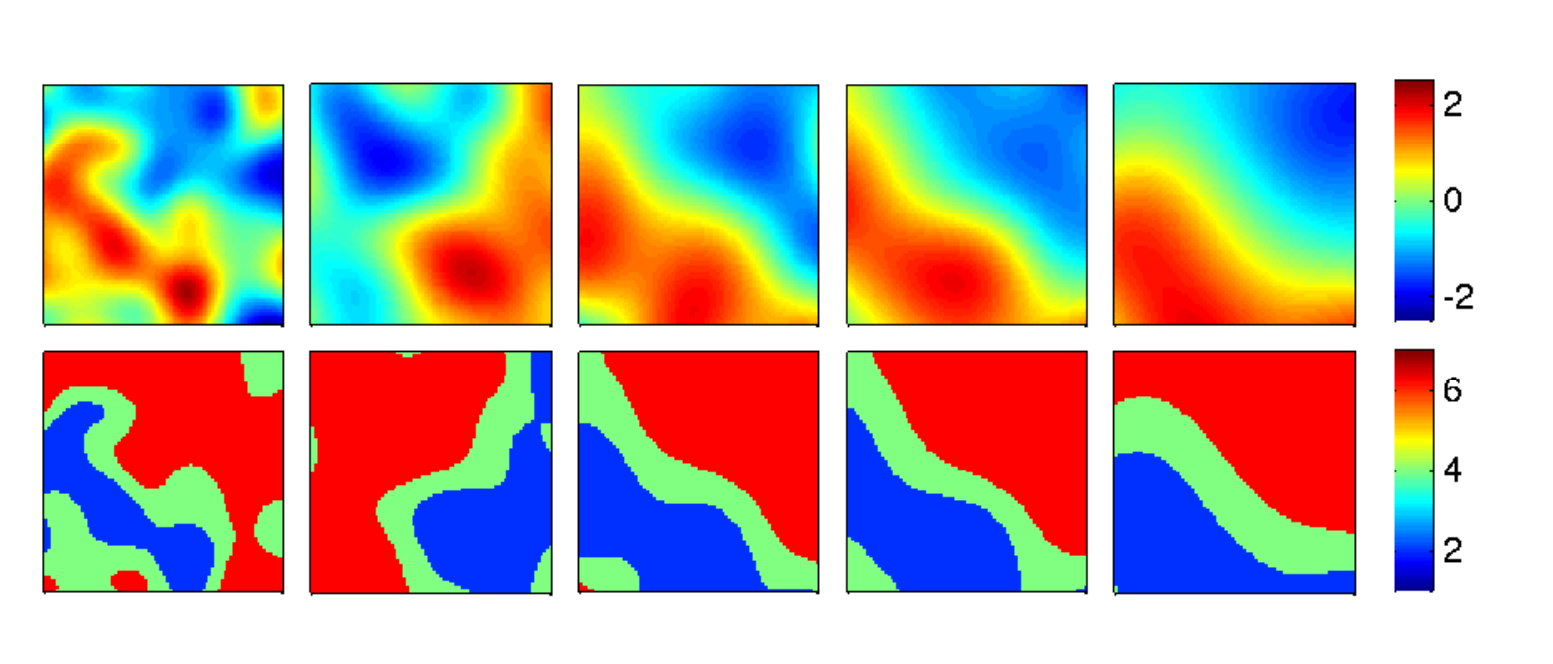}\\
\vskip-15pt
     \includegraphics[scale=0.85]{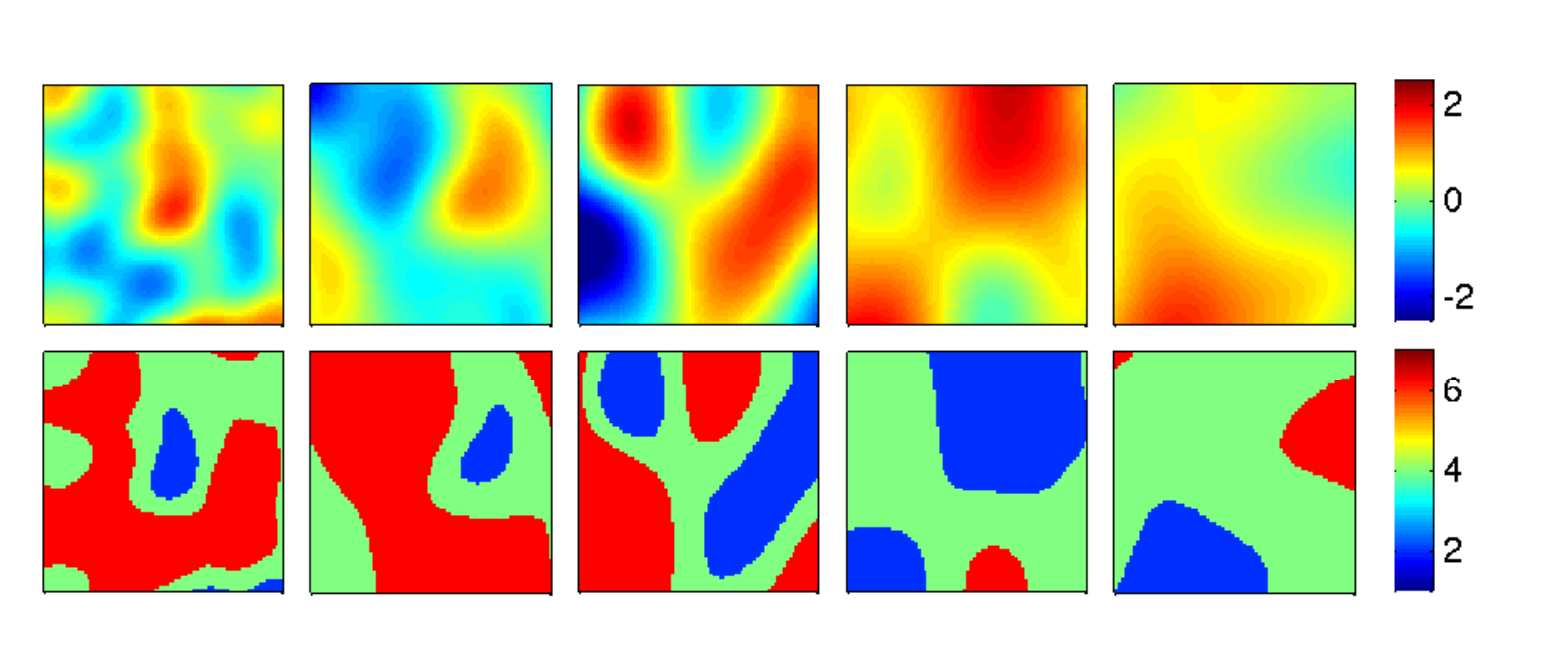}
 \caption{Identification of structural geology (channel model). Samples from the prior on the level set (first, third and fifth rows) for (from left to right) $L=0.2, 0.3, 0.35, 0.4, 0.5$. Pushforward onto $\kappa$ (second, fourth and sixth rows) associated to each of these samples from the level set function.}     \label{Fig7}
\end{center}
\end{figure}

\begin{figure}[htbp]
\begin{center}
\includegraphics[scale=1.1]{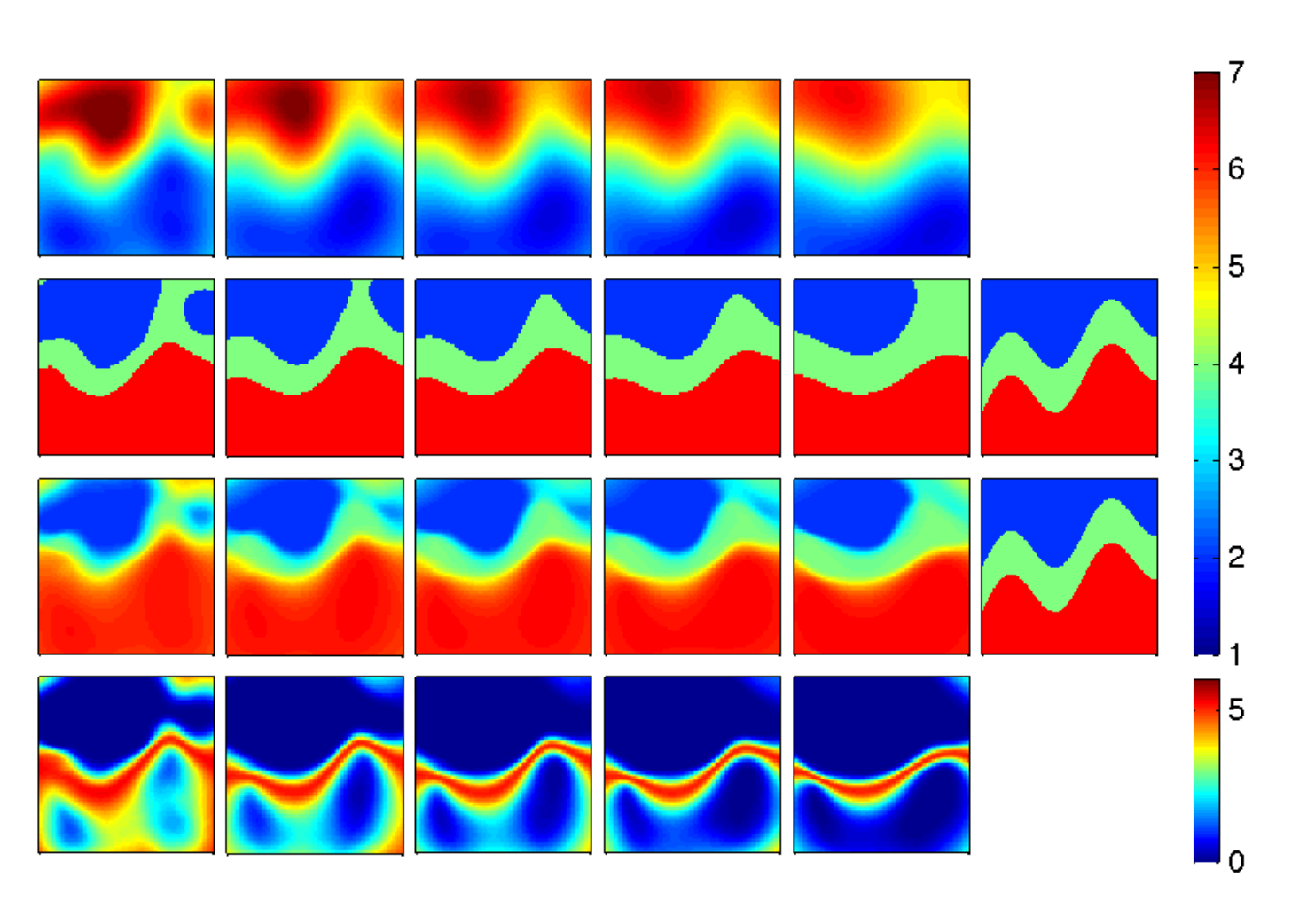}
 \caption{Identification of structural geology (channel model). MCMC results for (from left to right) $L=0.2, 0.3, 0.35, 0.4, 0.5$  in the eq. (\ref{eq:cova1}). Top: MCMC mean of the level set function. Top-middle: $\kappa$ associated to the mean of the level set function (true $\kappa$ is displayed in the last column). Bottom-middle: Mean of the $\kappa$. Bottom: Variance of $\kappa$}   \label{Fig8}
\end{center}
\end{figure}

\begin{figure}[htbp]
\begin{center}
\includegraphics[scale=0.85]{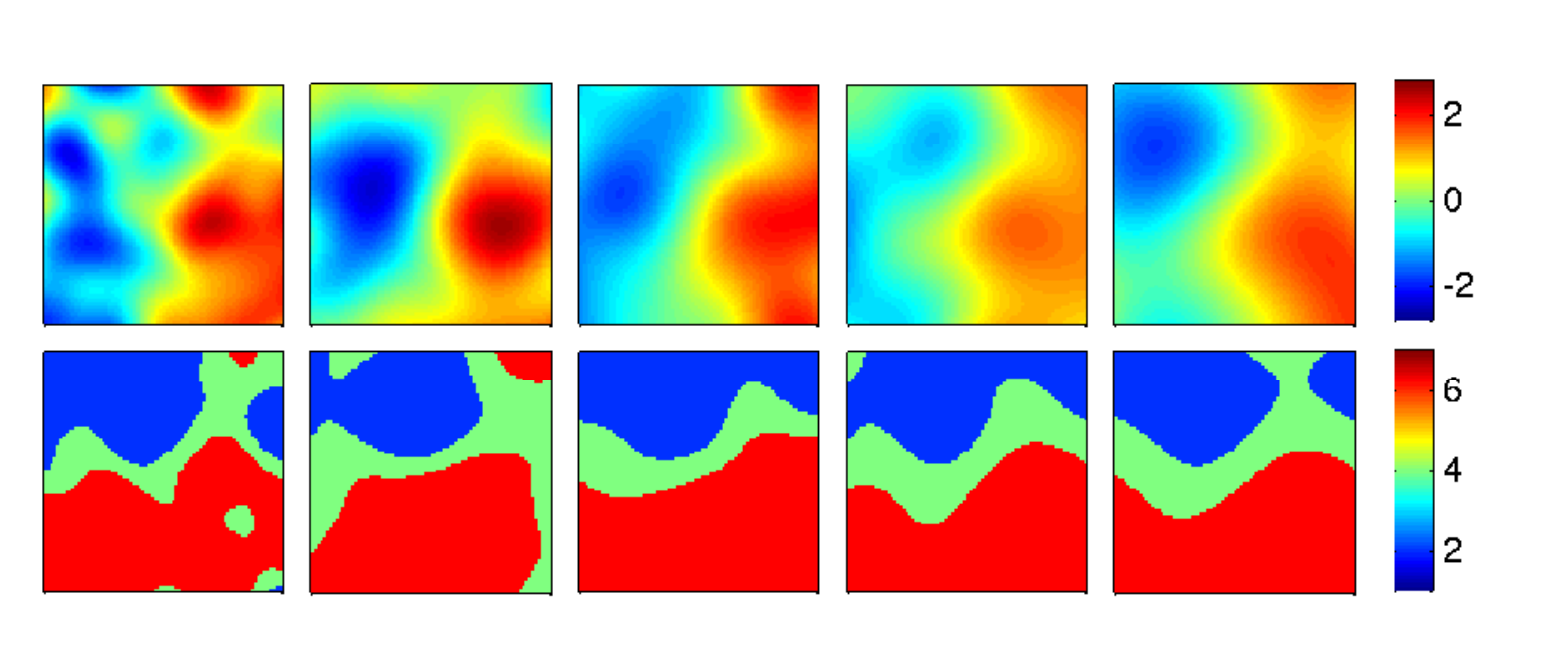}\\
\vskip-15pt
\includegraphics[scale=0.85]{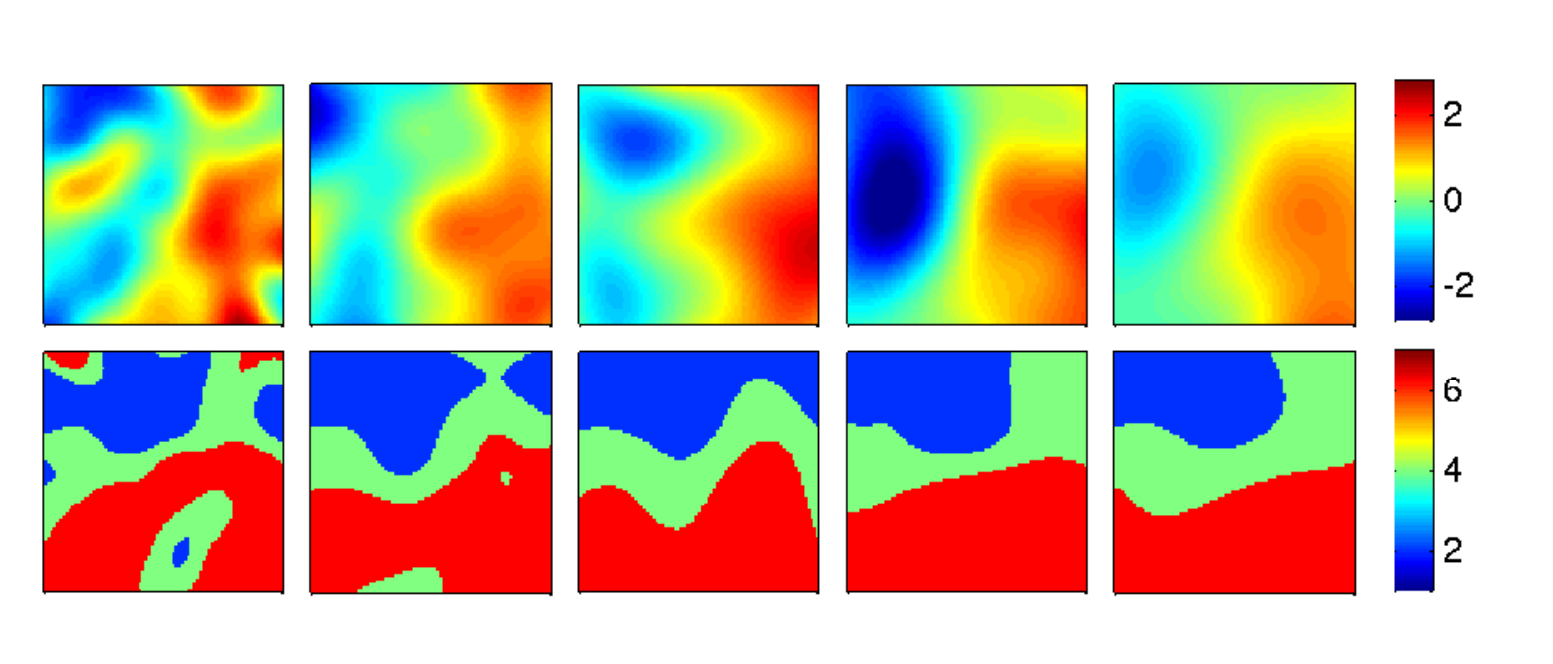}\\
\vskip-15pt
\includegraphics[scale=0.85]{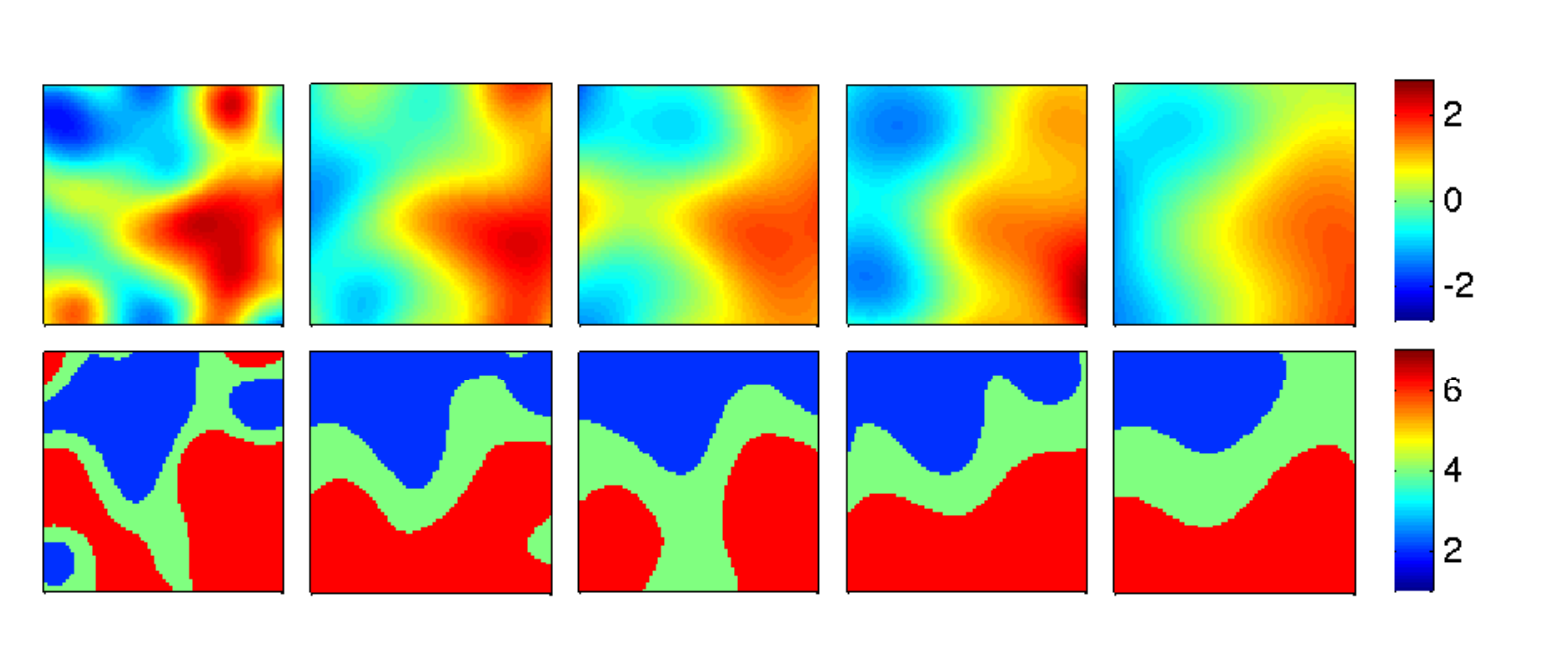}

 \caption{Identification of structural geology (channel model). Samples from the posterior on the level set (first, third and fifth rows) for (from left to right) $L=0.2, 0.15, 0.2, 0.3, 0.4$. $\log(\kappa)$ (second, fourth and sixth rows) associated to each of these samples from the level set function.}   \label{Fig9}
\end{center}
\end{figure}

\begin{figure}[htbp]
\begin{center}
\includegraphics[scale=0.25]{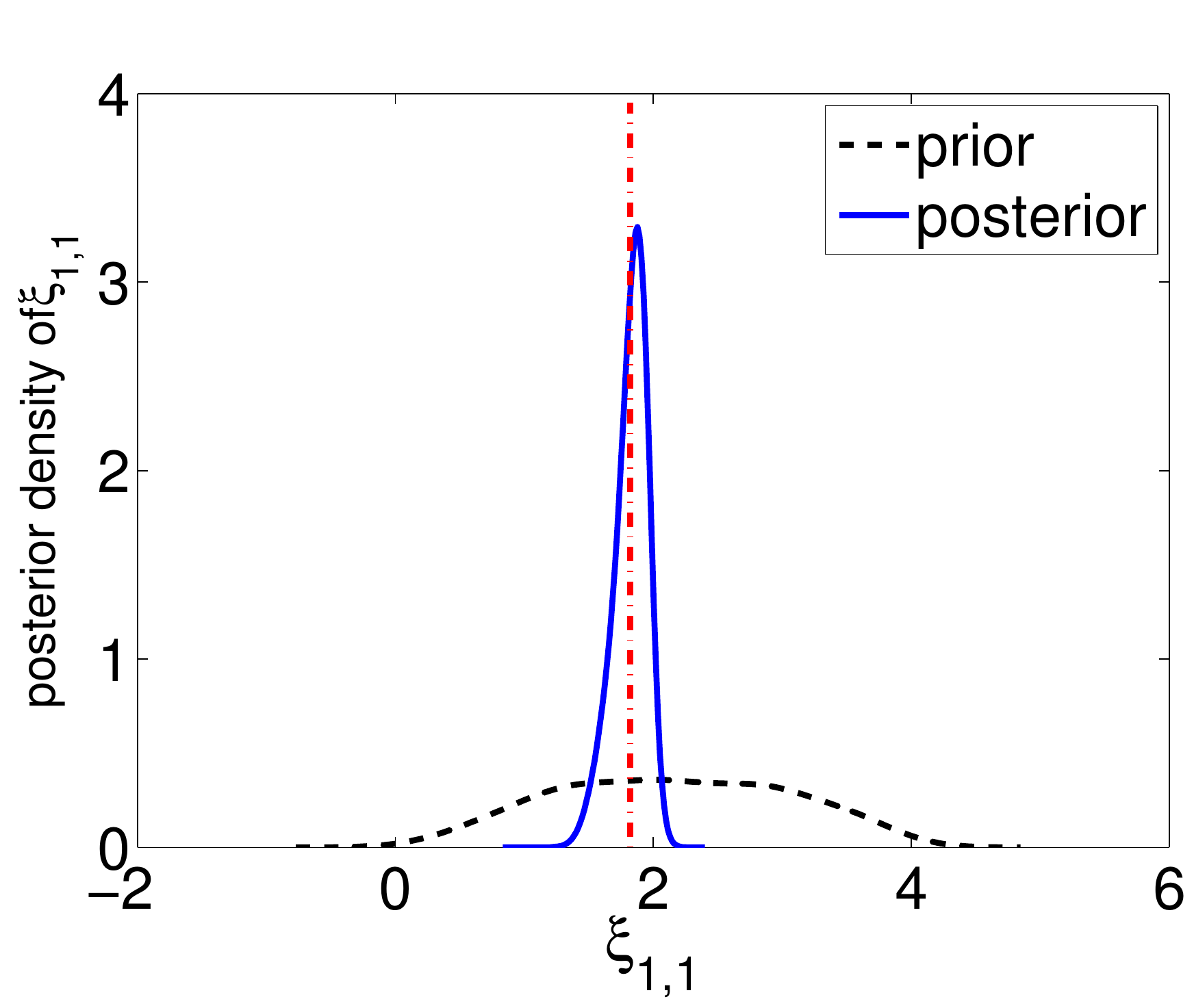}
\includegraphics[scale=0.25]{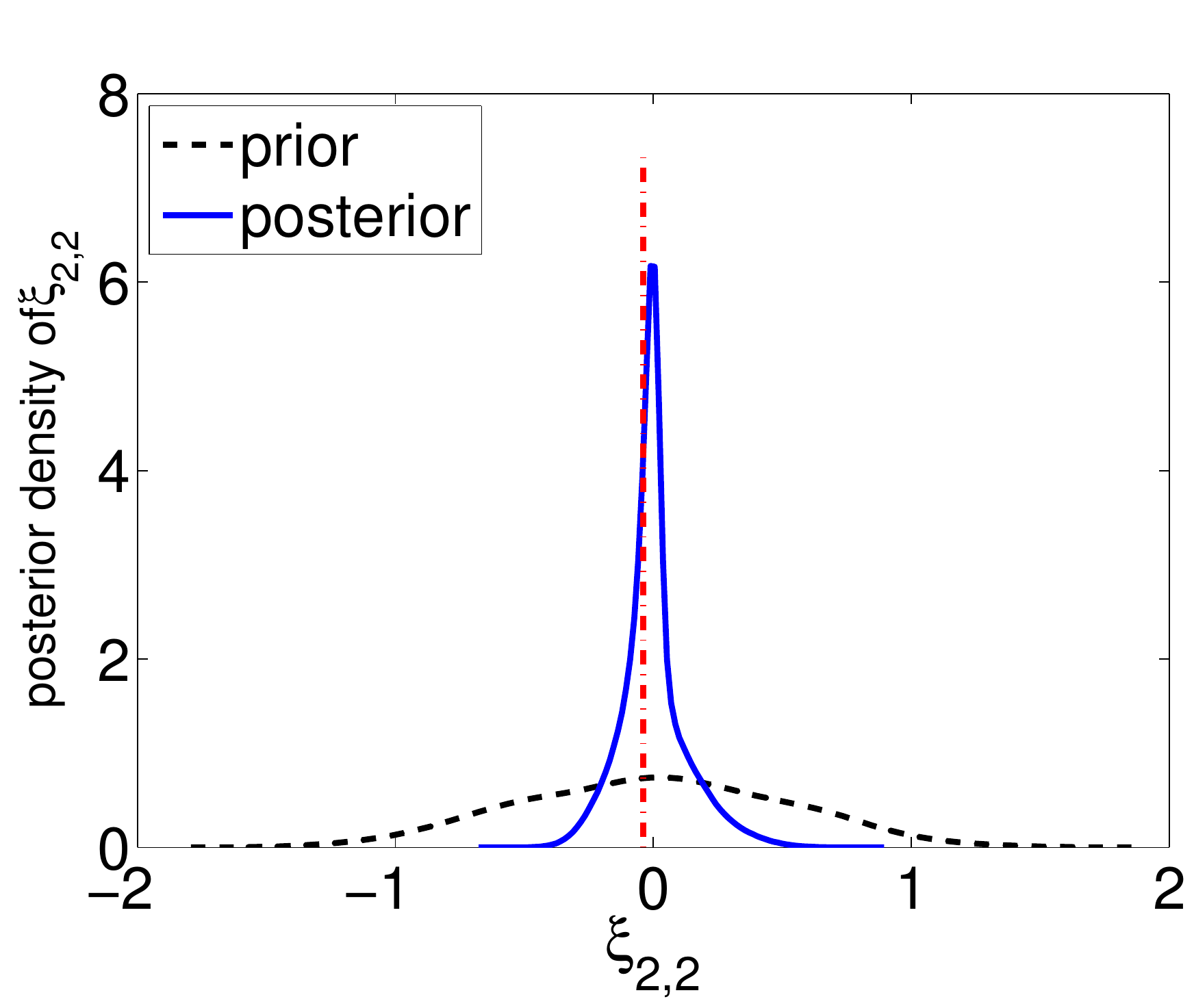}
\includegraphics[scale=0.25]{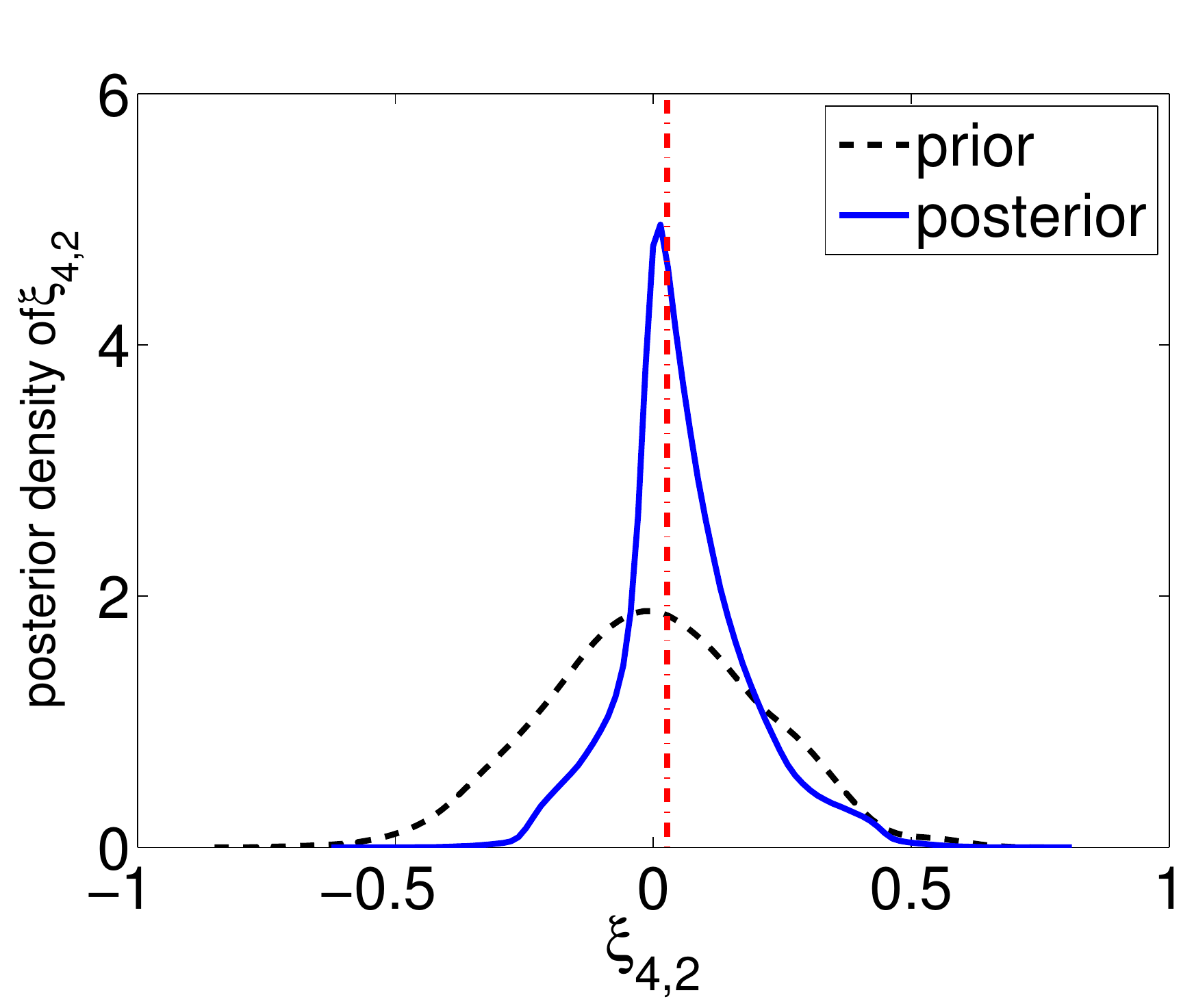}\\
\includegraphics[scale=0.25]{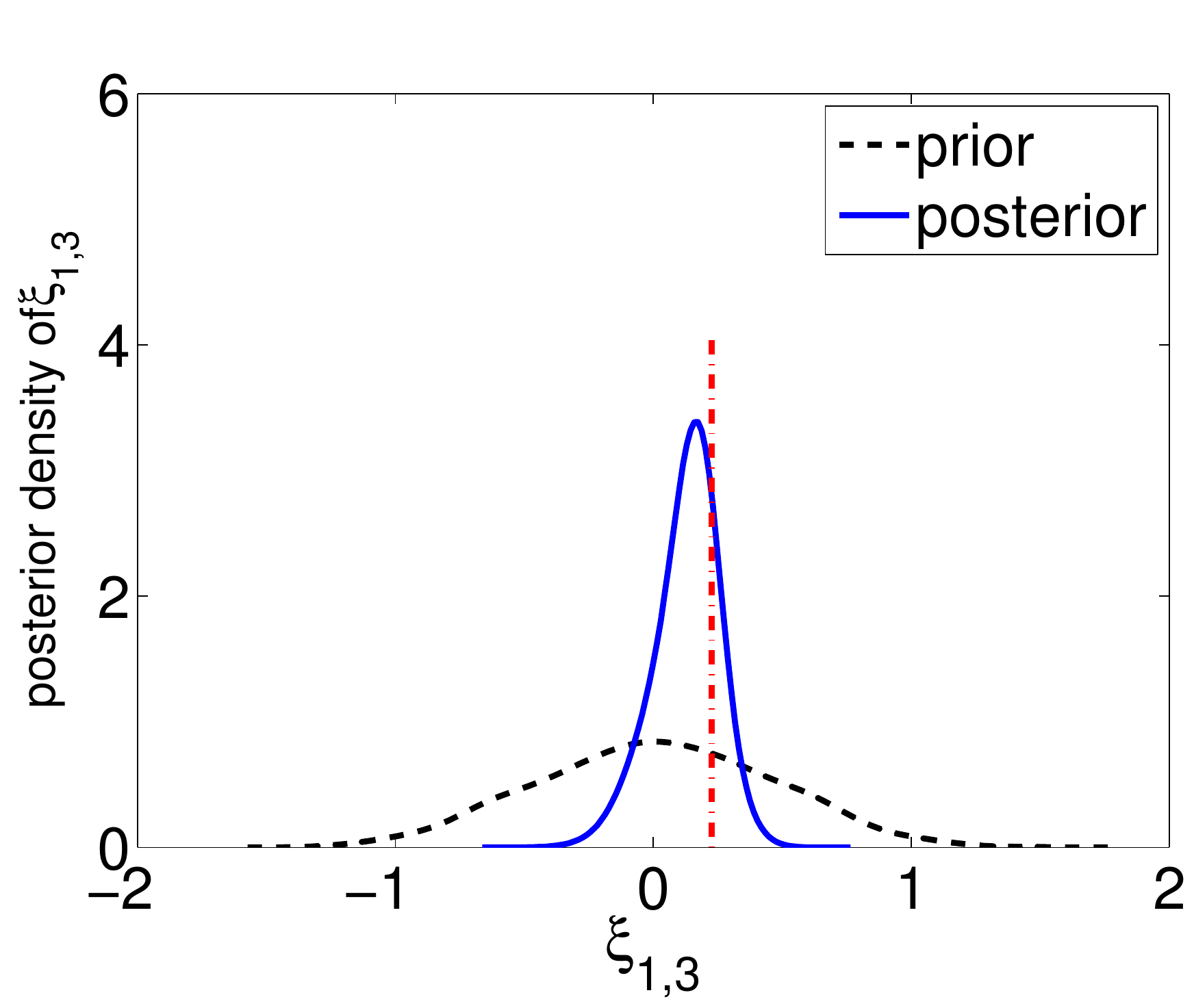}
\includegraphics[scale=0.25]{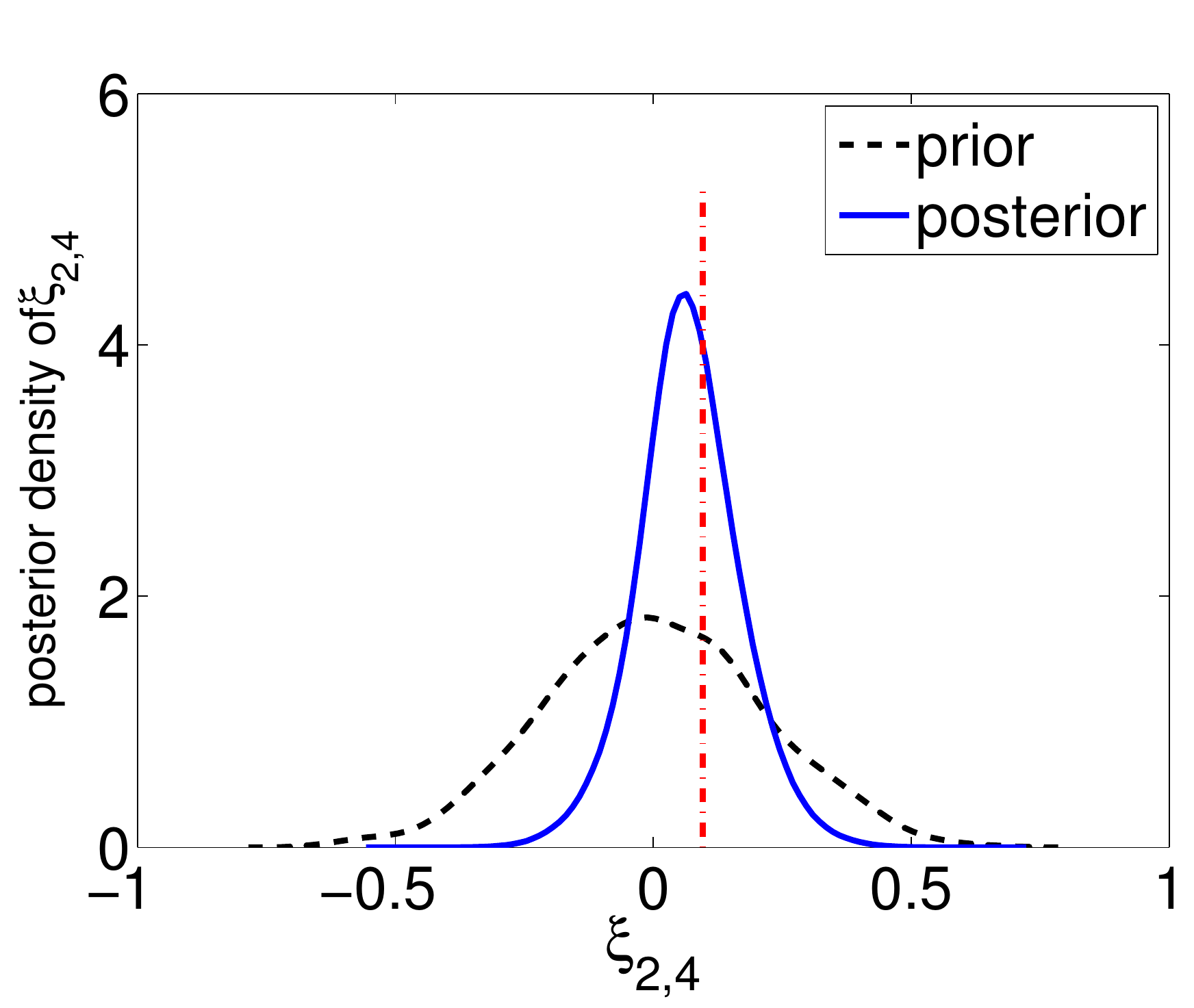}
\includegraphics[scale=0.25]{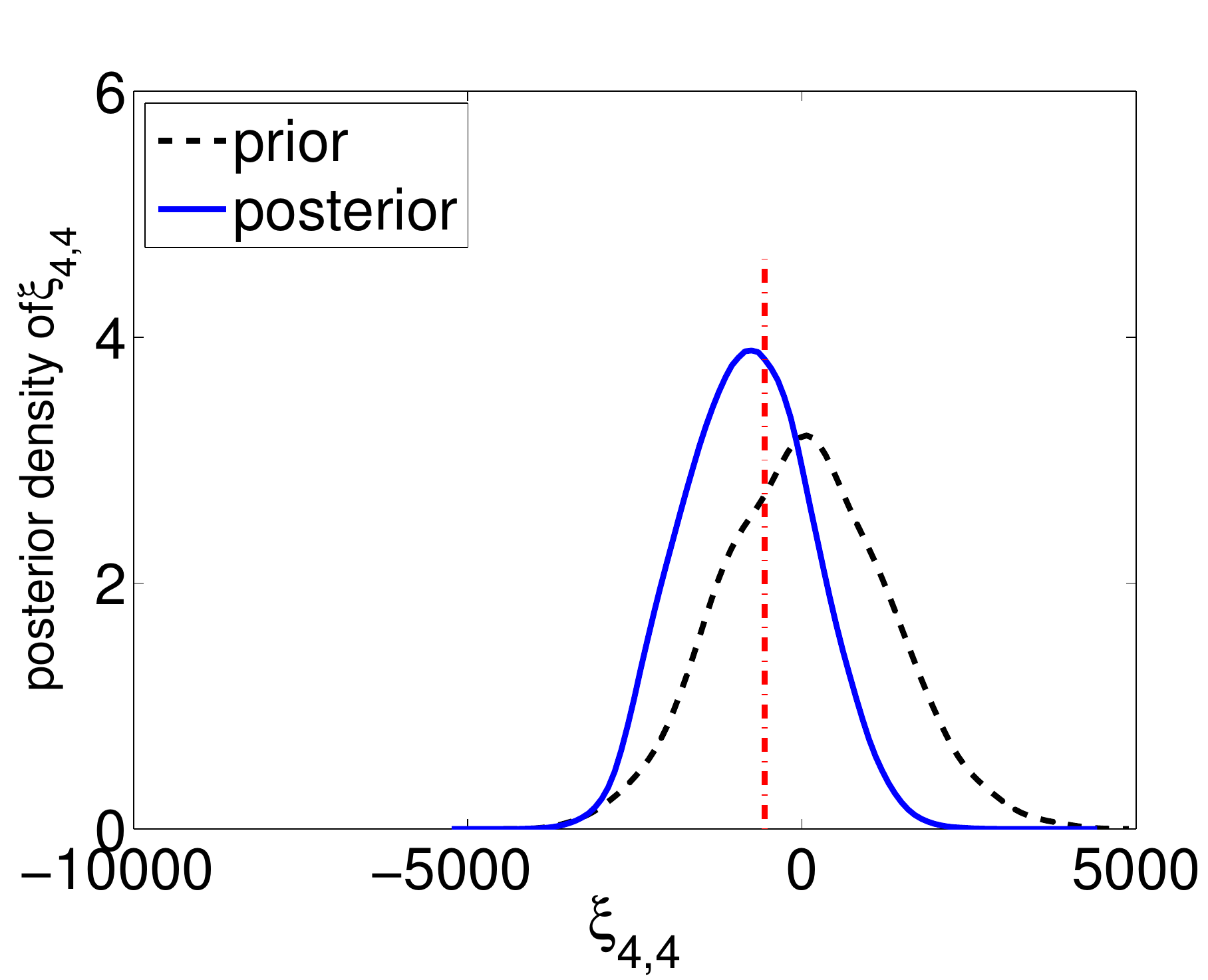}
 \caption{Identification of structural geology (channel model). Densities of the prior and posterior of some DCT coefficients of the  $\kappa$'s obtained from MCMC samples on the level set for $L=0.4$ (vertical dotted line indicates the truth).}   \label{Fig10}
\end{center}
\end{figure}

\subsection{Structural Geology: Layer Model}

In this experiment we consider the groundwater model (\ref{sdarcy}) with the same domain and measurement configurations from the preceding subsection. However, for this case we define the true permeability $\kappa^{\dagger}$ displayed in \Fref{Fig11} (top). The permeability values are as before. The generation of synthetic data is conducted as described in the preceding subsection. For this example we consider the Gaussian prior on the level set defined by (\ref{eq:cova1}). Since for this case the operator $\C$ is diagonalisable by cosine functions, we use the fast Fourier transform to sample from the corresponding Gaussian measure $\N(0,\C)$ required by the pCN-MCMC algorithm. 

The tunable parameter $\alpha$ in the covariance operator (\ref{eq:cova1}) determines the regularity of the corresponding samples of the Gaussian prior (see for example \cite{S10}). Indeed, in \Fref{Fig12} we show samples from the prior on the level set function (first, third and fifth rows) and the corresponding $\kappa$ (second, fourth and sixth rows) for (from left to right) $\alpha=1.5, 2.0, 2.5, 3.0, 3.5$. Indeed, changes in $\alpha$ have a dramatic effect on the regularity of the interface between the different regions. We therefore expect strong effect on the resulting posterior on the level set and thus on the permeability.

 In \Fref{Fig13} we display numerical results from MCMC chains with different priors corresponding to  (from left to right) $\alpha=1.5, 2.0, 2.5, 3.0, 3.5$. In \Fref{Fig13} we present the MCMC mean of the level set function.The corresponding $\kappa$ is shown in the top-middle of \Fref{Fig13}. In this figure we additionally display the mean (bottom-middle)  and the variance (bottom) of the $\kappa$'s obtained from the MCMC samples of the level set function. Above a critical value $\alpha=2.5$ we obtain a reasonable identification of the layer permeability with a small uncertainty (quantified by the variance).  \Fref{Fig14} shows posterior (MCMC) samples of the level set function (first, third and fifth rows) and the corresponding $\kappa$ (second, fourth and sixth rows) for the aforementioned choices of $\alpha$.

\Fref{Fig11} (bottom-right) shows the ACF of the first KL mode of level set function from MCMC experiments with different priors with $\alpha$'s as before. The efficiency of the MCMC chain does not seem to vary significantly for the values above the critical value of $\alpha$. However, as in the previous examples a slow decay in the ACF is obtained. An experiment using $50$ multiple MCMC chains initialized randomly from the prior reveals that the Gelman-Rubin diagnostic 
test \cite{Gelman} is passed for $\alpha=2.5$ as we can observe from  \Fref{Fig11} (bottom-left) where we the display PSRF from MCMC samples of the level set function (red-solid line) and the corresponding mapping into the $\kappa$ (blue-dotted line). In \Fref{Fig15}  we show the prior and posterior densities of the DCT coefficients on the $\kappa$ obtained from the MCMC samples of the level set function (the vertical dotted line corresponds to the truth DCT coefficient). We see clearly that the DCT coefficients are substantially informed by the data although the spread around the truth confirms the variability in the location of the interface between the layers that we can ascertain from the posterior samples (see \Fref{Fig14}).

\begin{figure}[htbp]
\begin{center}
\includegraphics[scale=0.35]{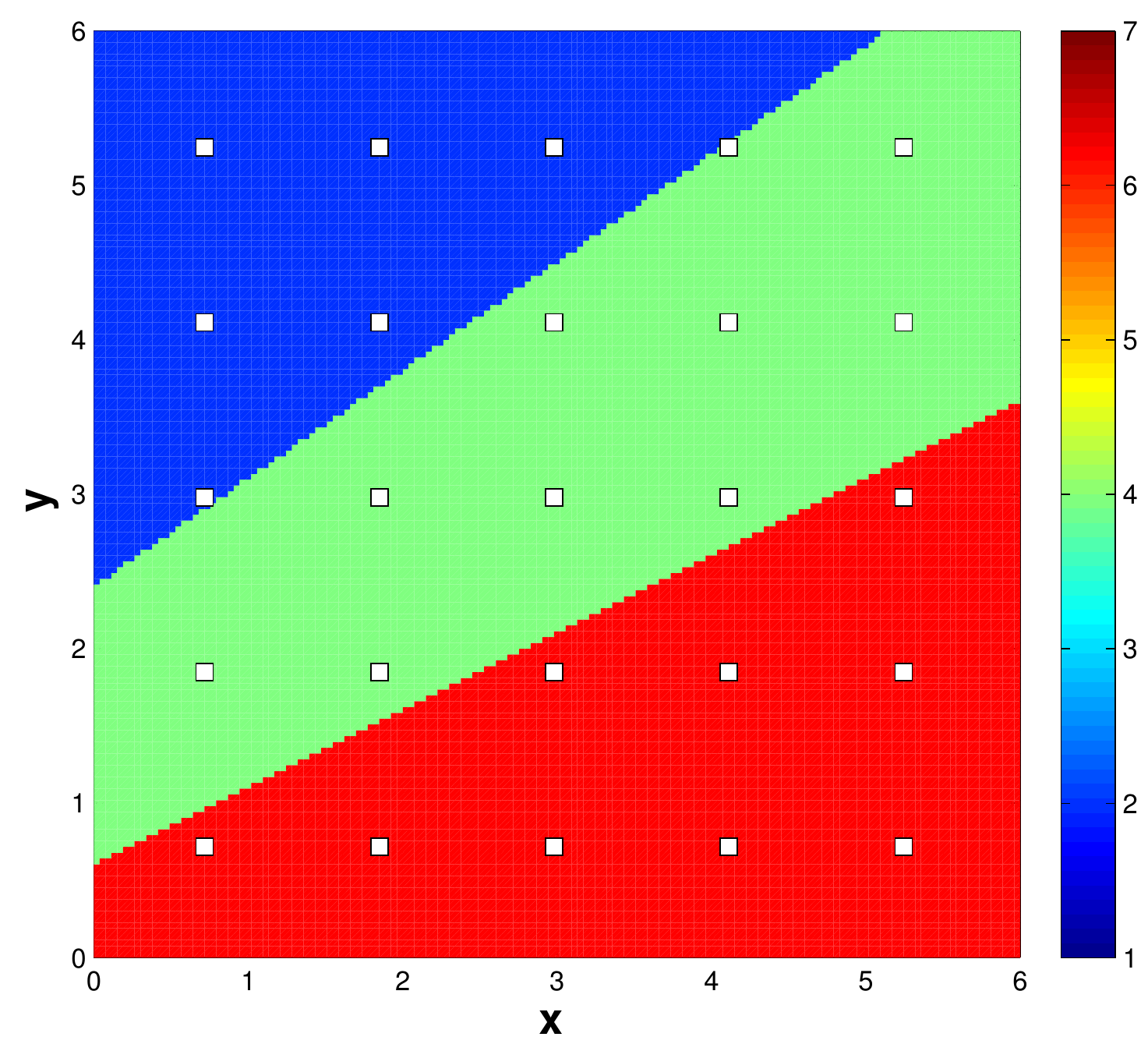}\\
\includegraphics[scale=0.3]{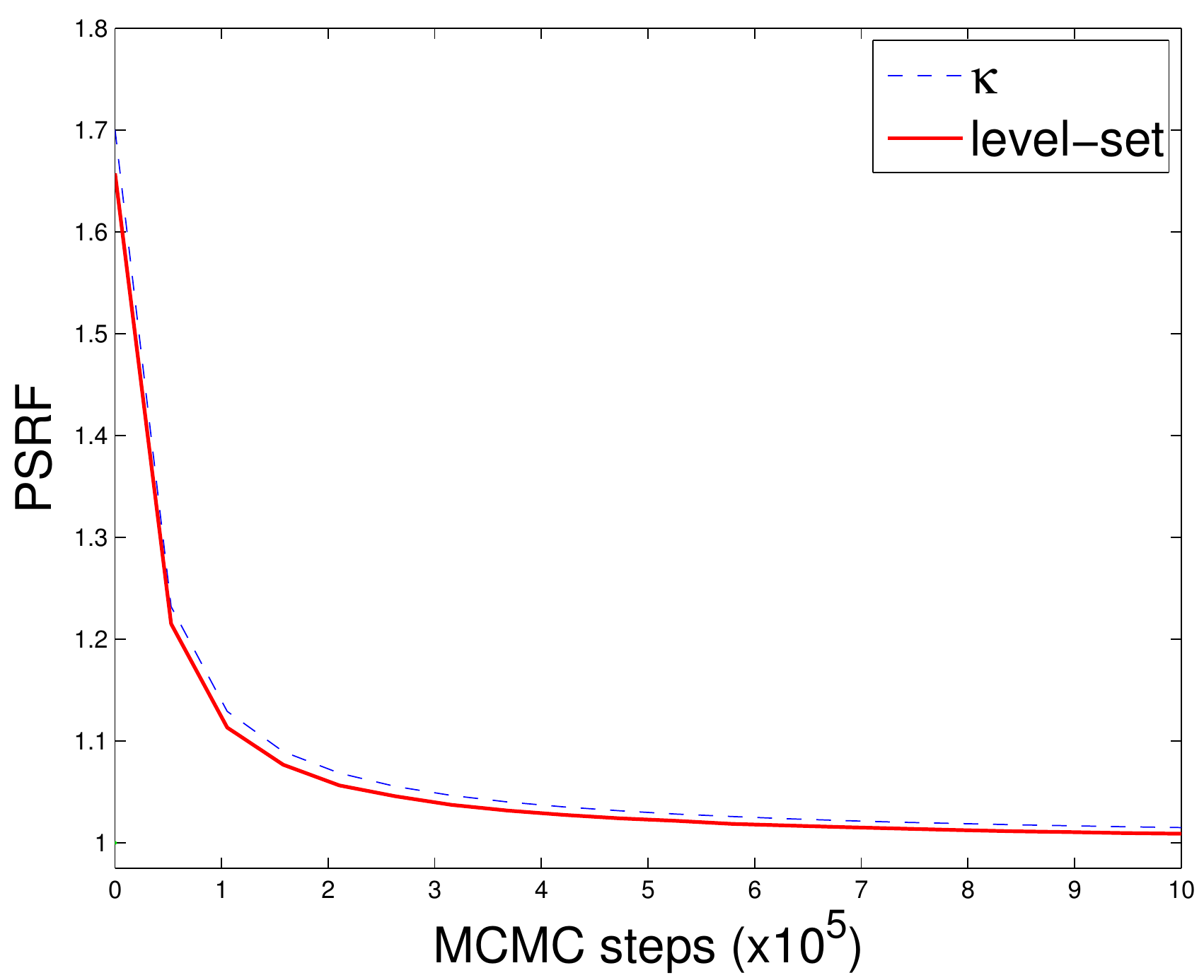}
\includegraphics[scale=0.3]{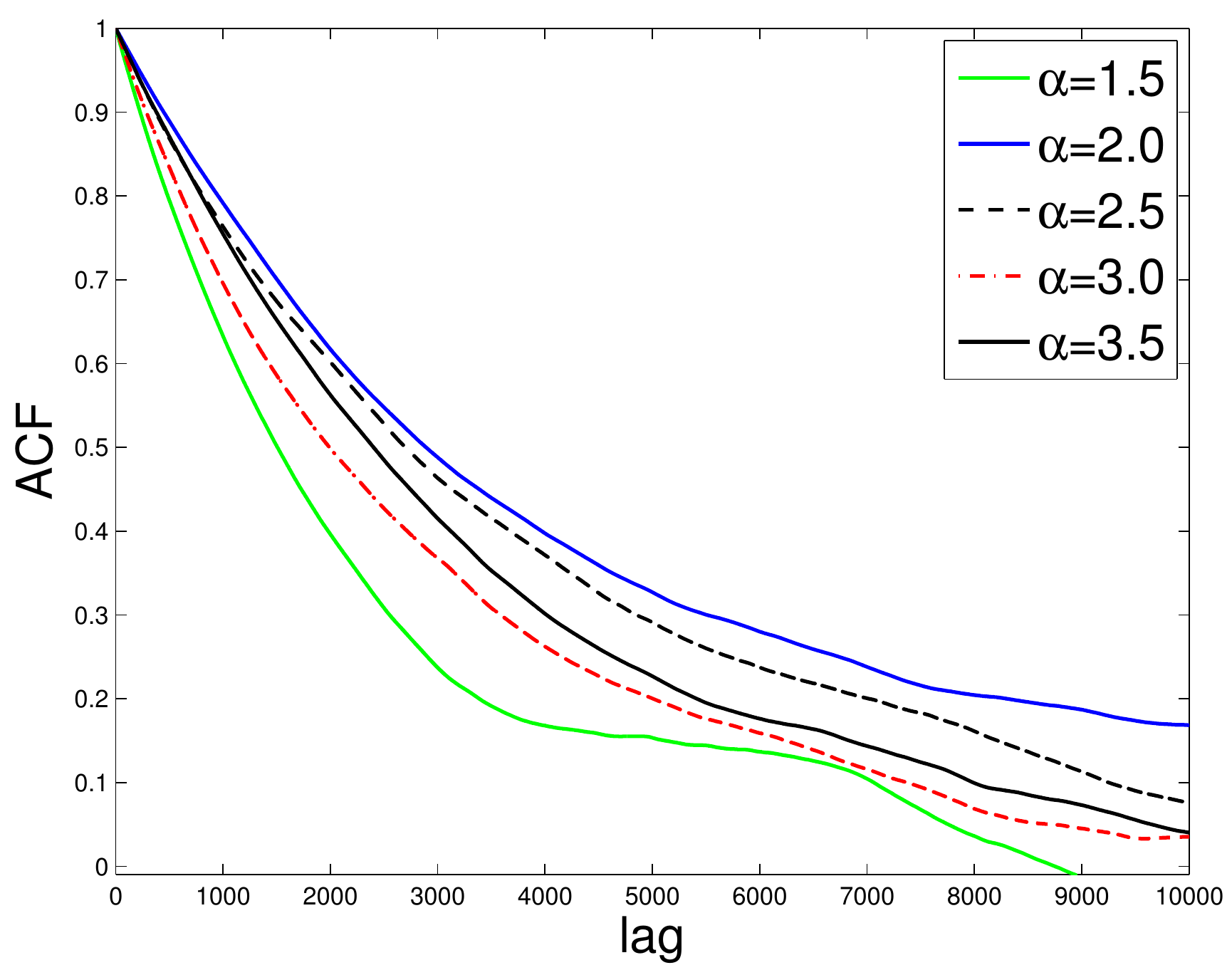}
 \caption{Identification of structural geology (layer model). Top: True $\kappa$ in eq. (\ref{sdarcy}). Bottom-left: PSRF from multiple chains with $\alpha=2.5$ in (\ref{eq:cova2}).  Bottom-right: ACF of first KL mode of the level set function from single-chain  single-chain MCMC with different choices of $\alpha$.}   \label{Fig11}
\end{center}
\end{figure}

\begin{figure}[htbp]
\begin{center}
\includegraphics[scale=0.85]{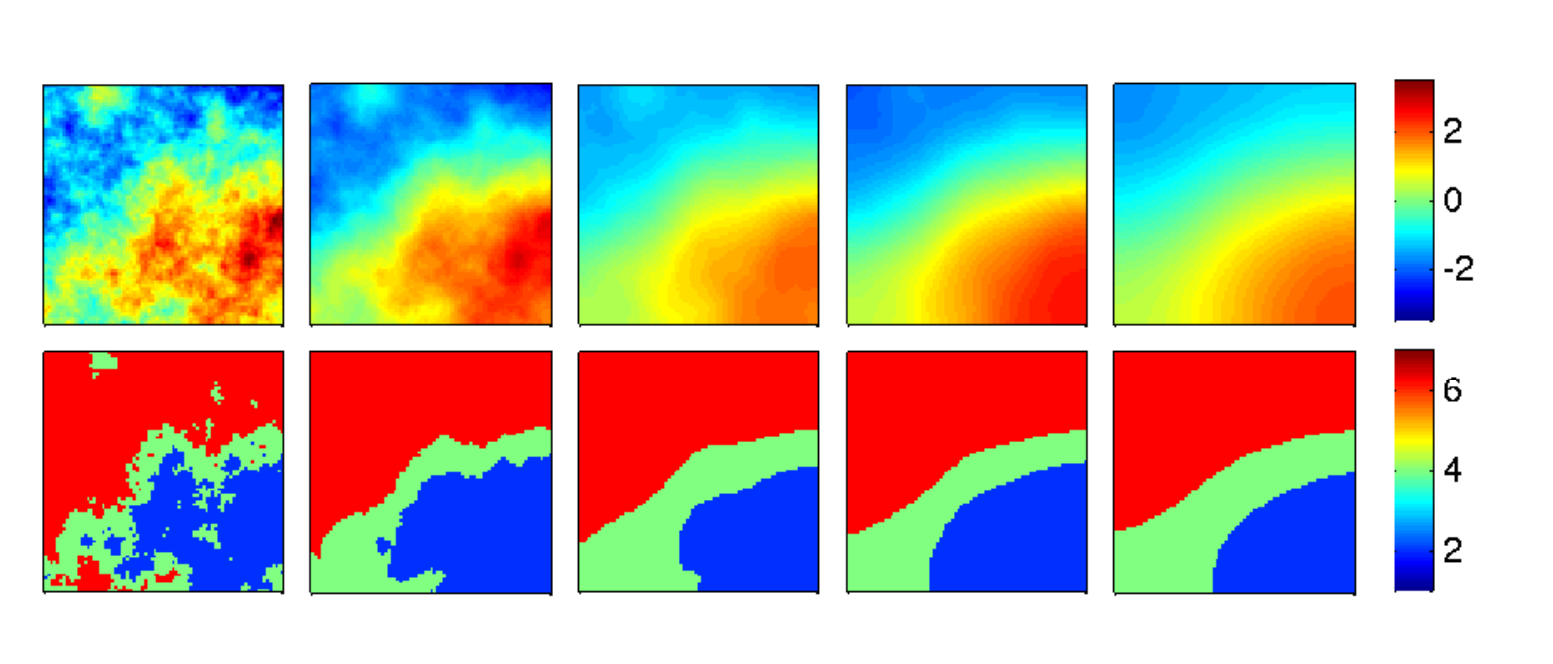}
\vskip-15pt
\includegraphics[scale=0.85]{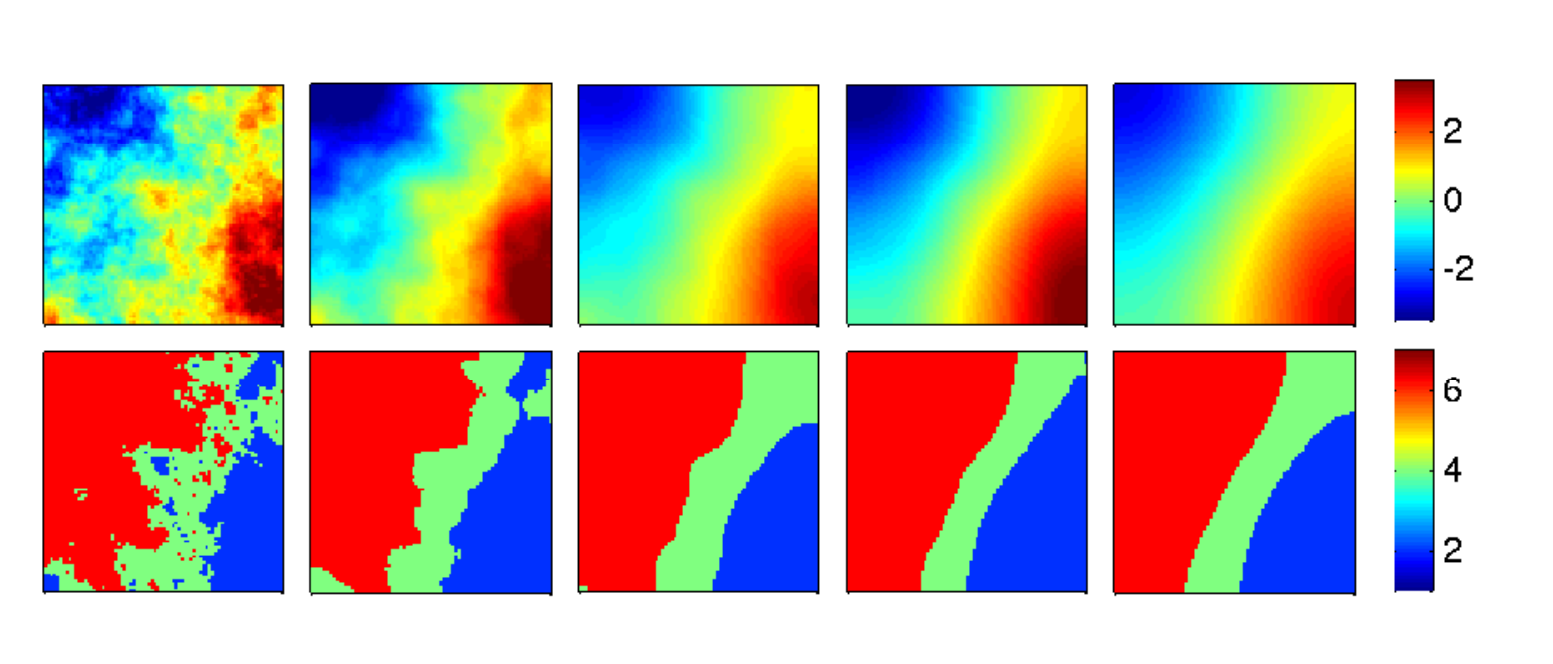}
\vskip-15pt
\includegraphics[scale=0.85]{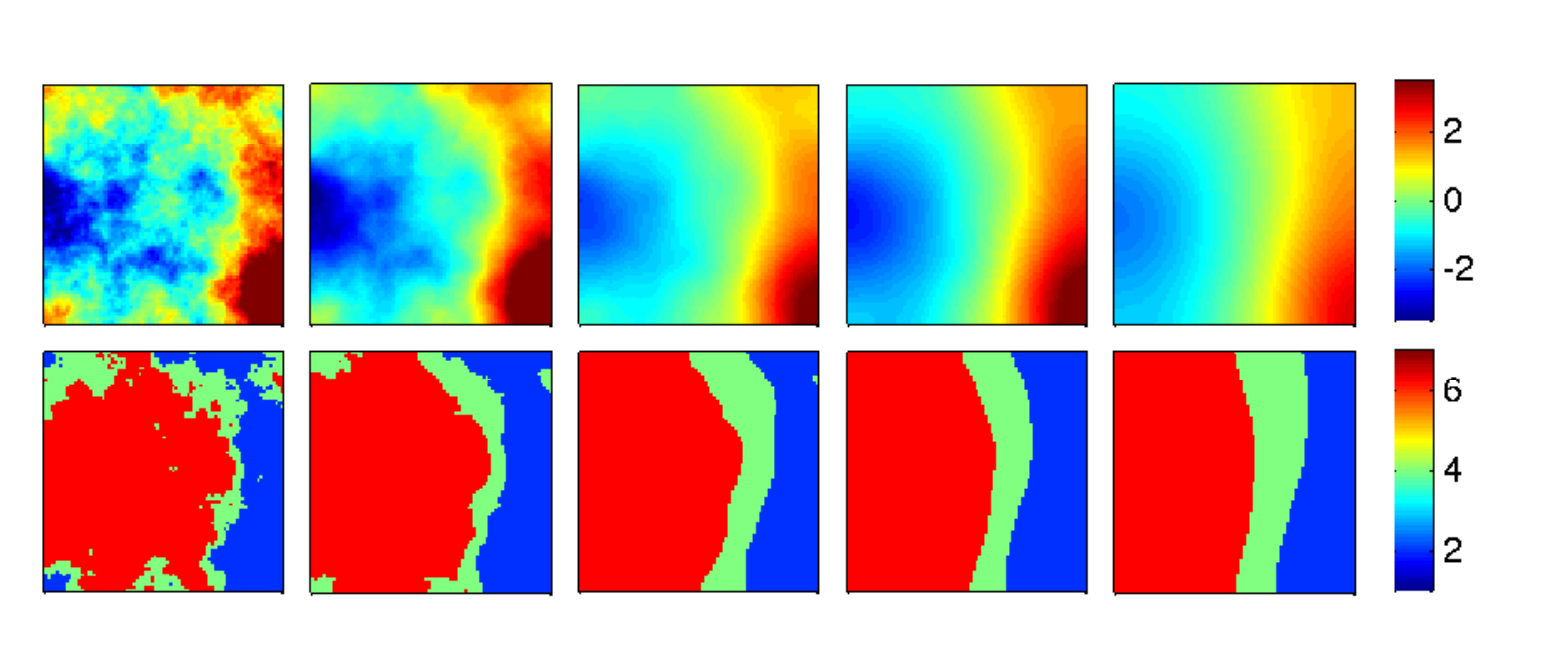}

 \caption{Identification of structural geology (layer model). Samples from the prior on the level set (first, third and fifth rows) for (from left to right) $\alpha=1.5, 2.0, 2.5, 3.0, 3.5$  in (\ref{eq:cova2}). $\kappa$ (second, fourth and sixth rows) associated to each of these samples from the level set function.}     \label{Fig12}
\end{center}
\end{figure}

\begin{figure}[htbp]
\includegraphics[scale=1.1]{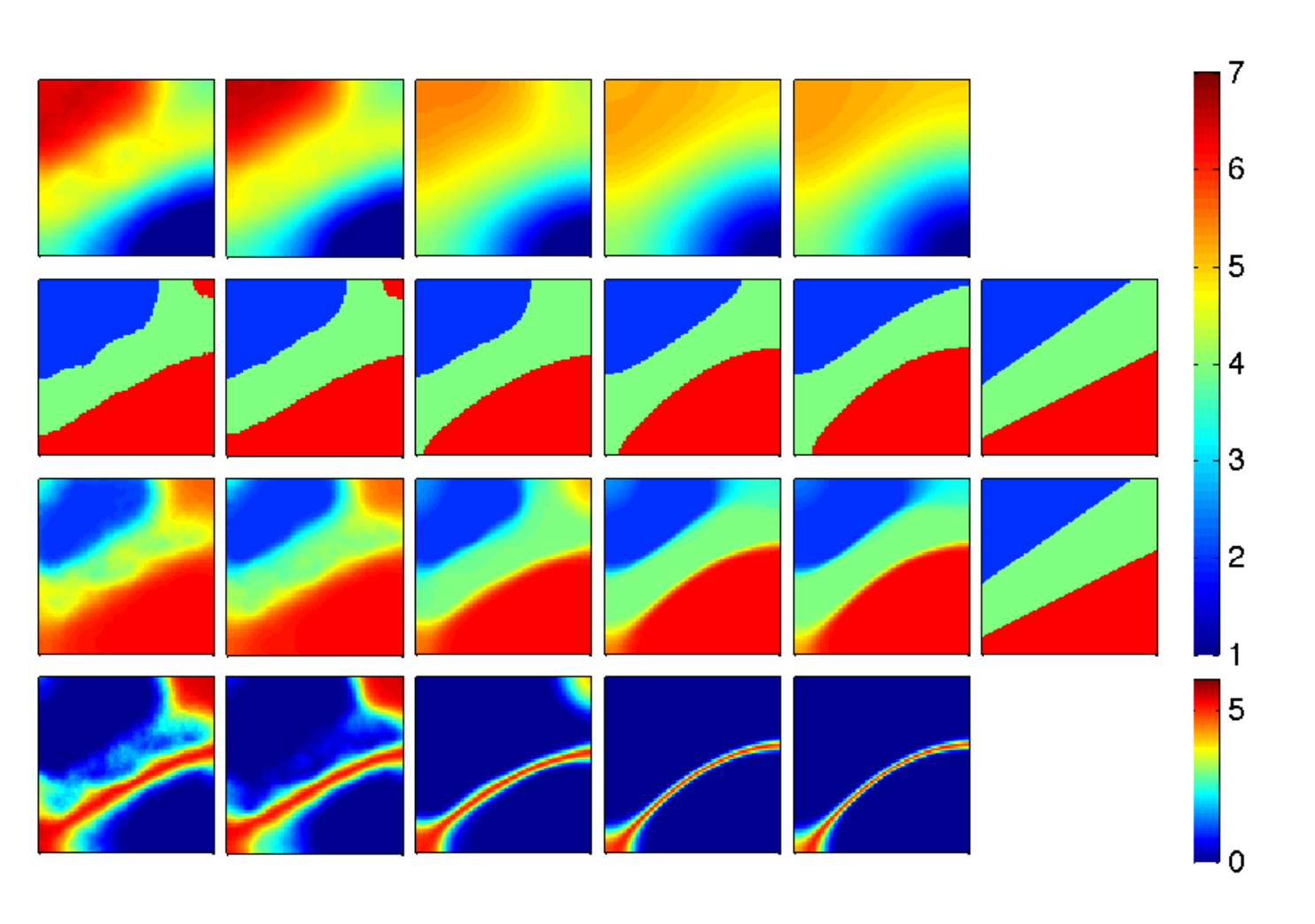}
 \caption{Identification of structural geology (layer model). MCMC results for (from left to right) $\alpha=1.5, 2.0, 2.5, 3.0, 3.5$  in the eq. (\ref{eq:cova2}). Top: MCMC mean of the level set function. Top-middle: $\kappa$ associated to the mean of the level set function (true $\kappa$ is displayed in the last column). Bottom-middle: Mean of the $\kappa$. Bottom: Variance of $\kappa$}   \label{Fig13}
\end{figure}

\begin{figure}[htbp]
\begin{center}
\includegraphics[scale=0.85]{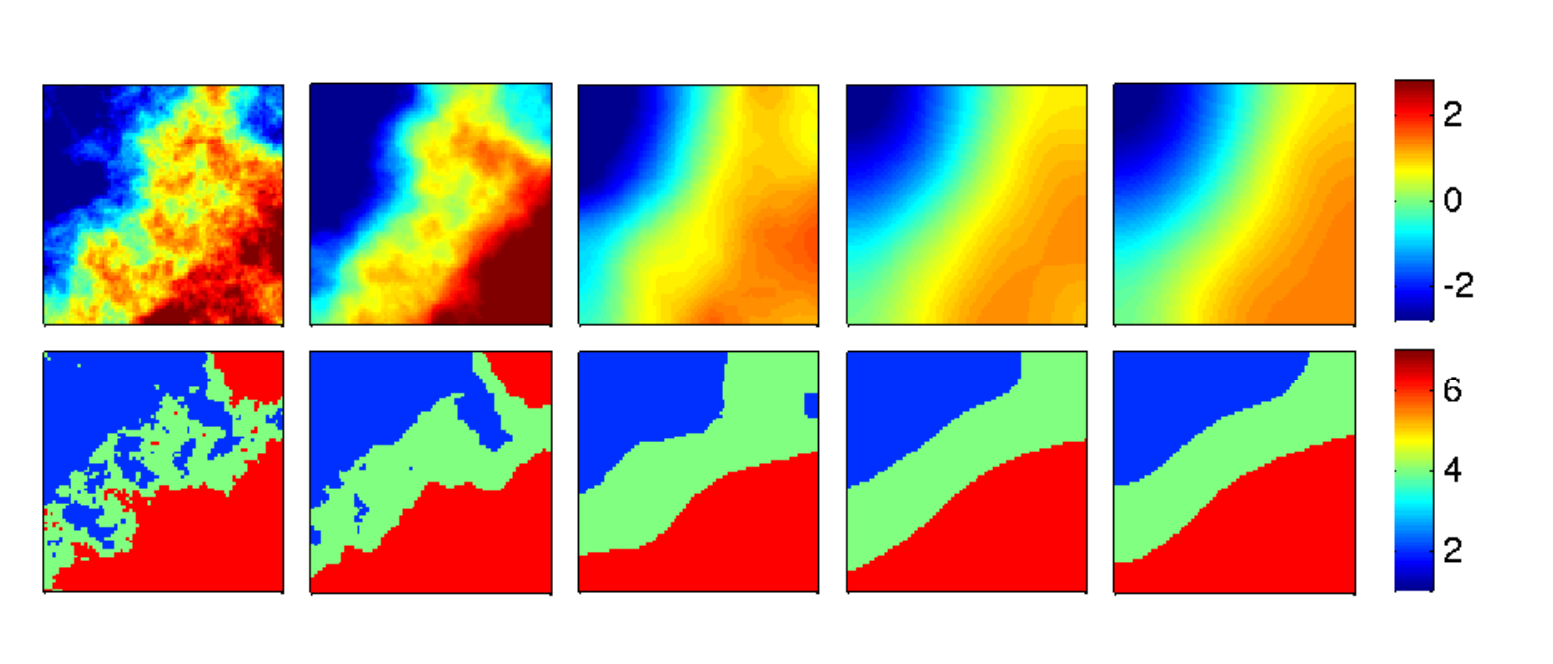}\\
\vskip-15pt
\includegraphics[scale=0.85]{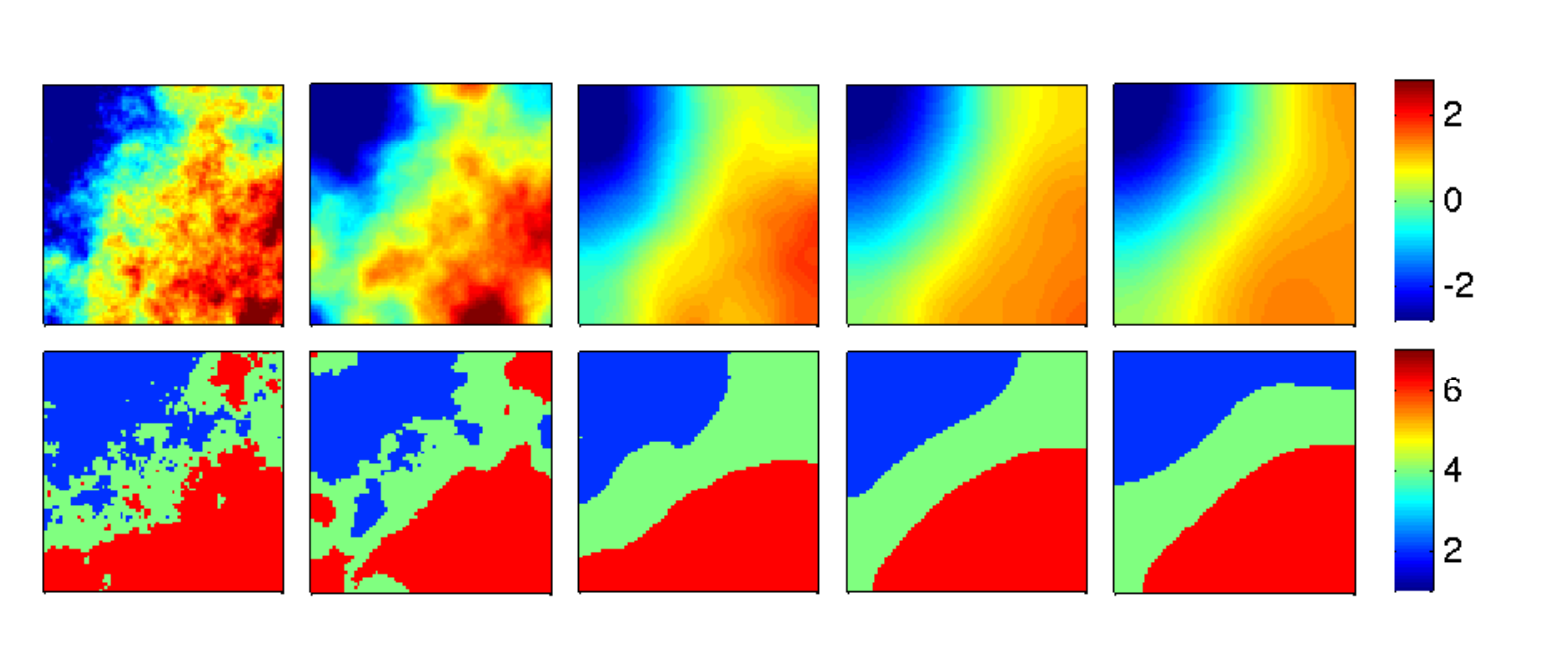}\\
\vskip-15pt
\includegraphics[scale=0.85]{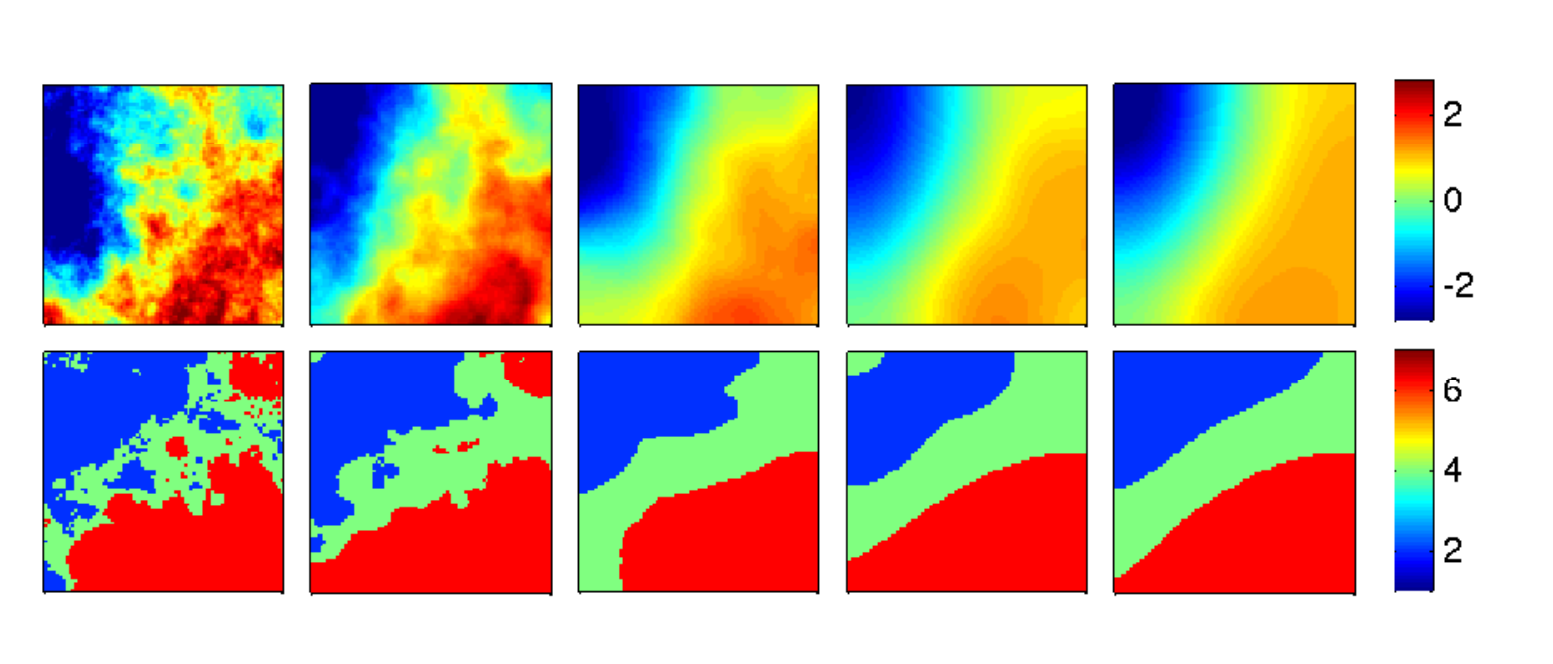}

 \caption{Identification of structural geology (layer model). Samples from the posterior on the level set (first, third and fifth rows) for (from left to right) $\alpha=1.5, 2.0, 2.5, 3.0, 3.5$  in the eq. (\ref{eq:cova2}). $\kappa$ (second, fourth and sixth rows) associated to each of these samples from the level set function.}   \label{Fig14}

\end{center}
\end{figure}

\begin{figure}[htbp]
\begin{center}
\includegraphics[scale=0.25]{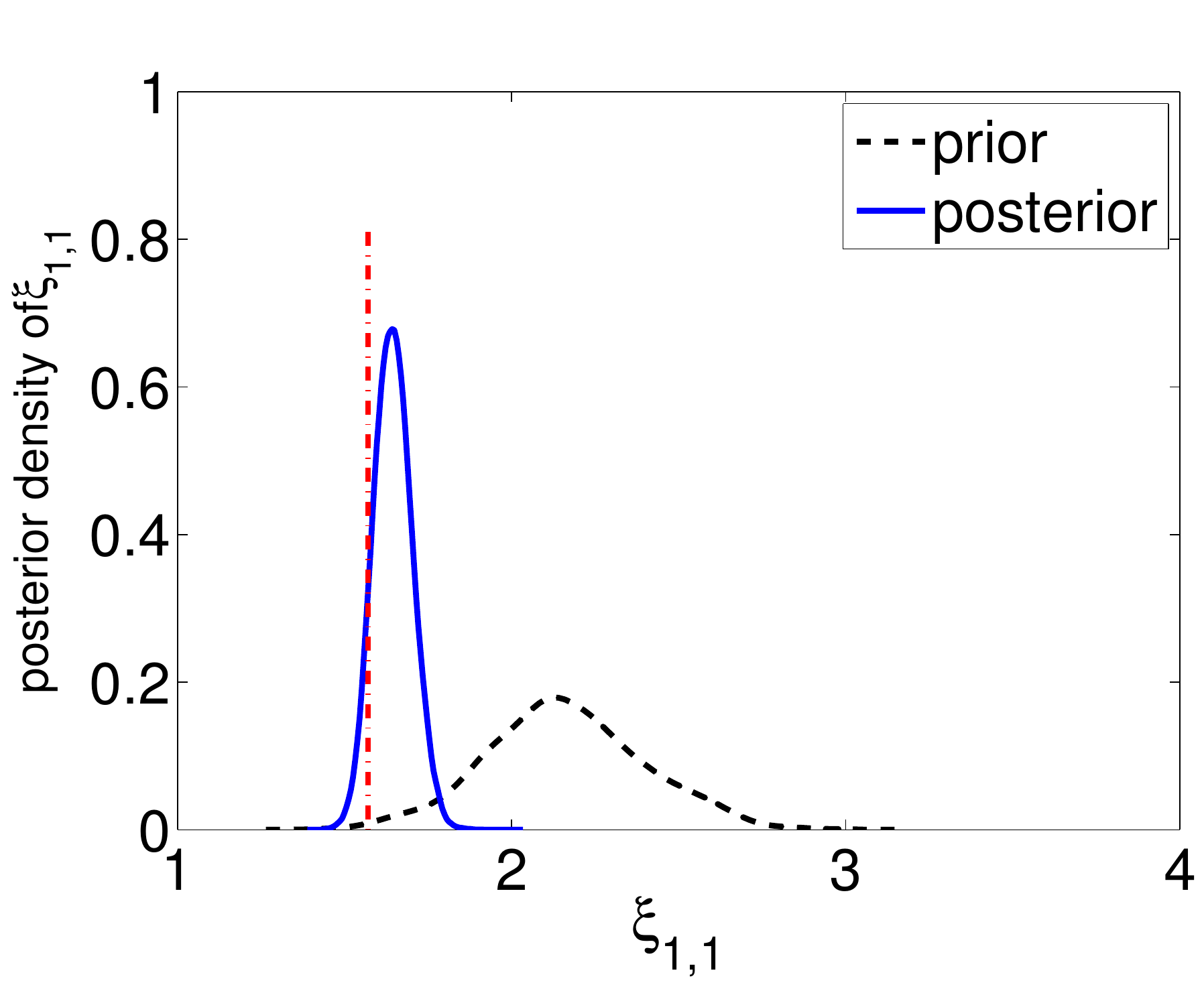}
\includegraphics[scale=0.25]{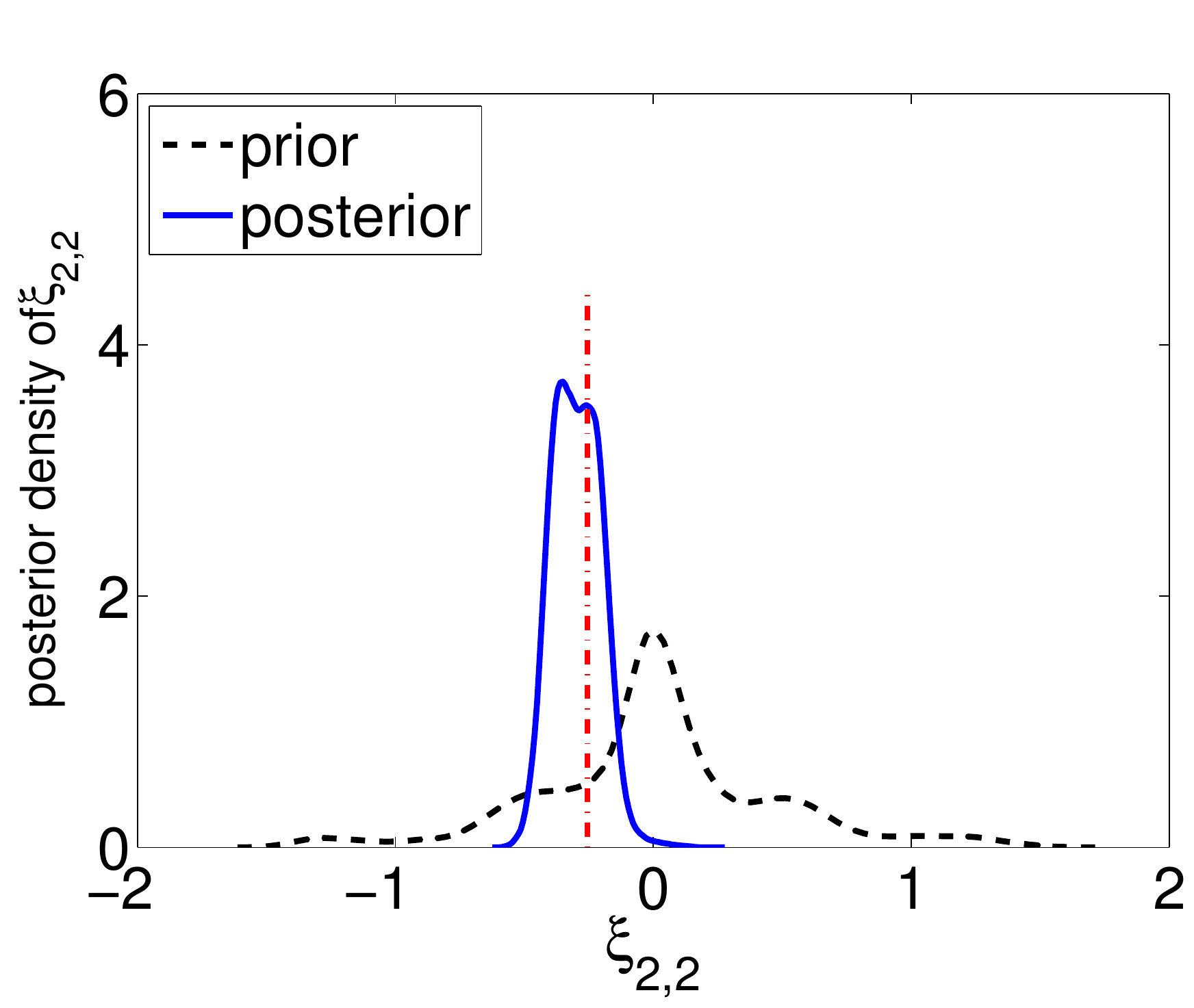}
\includegraphics[scale=0.25]{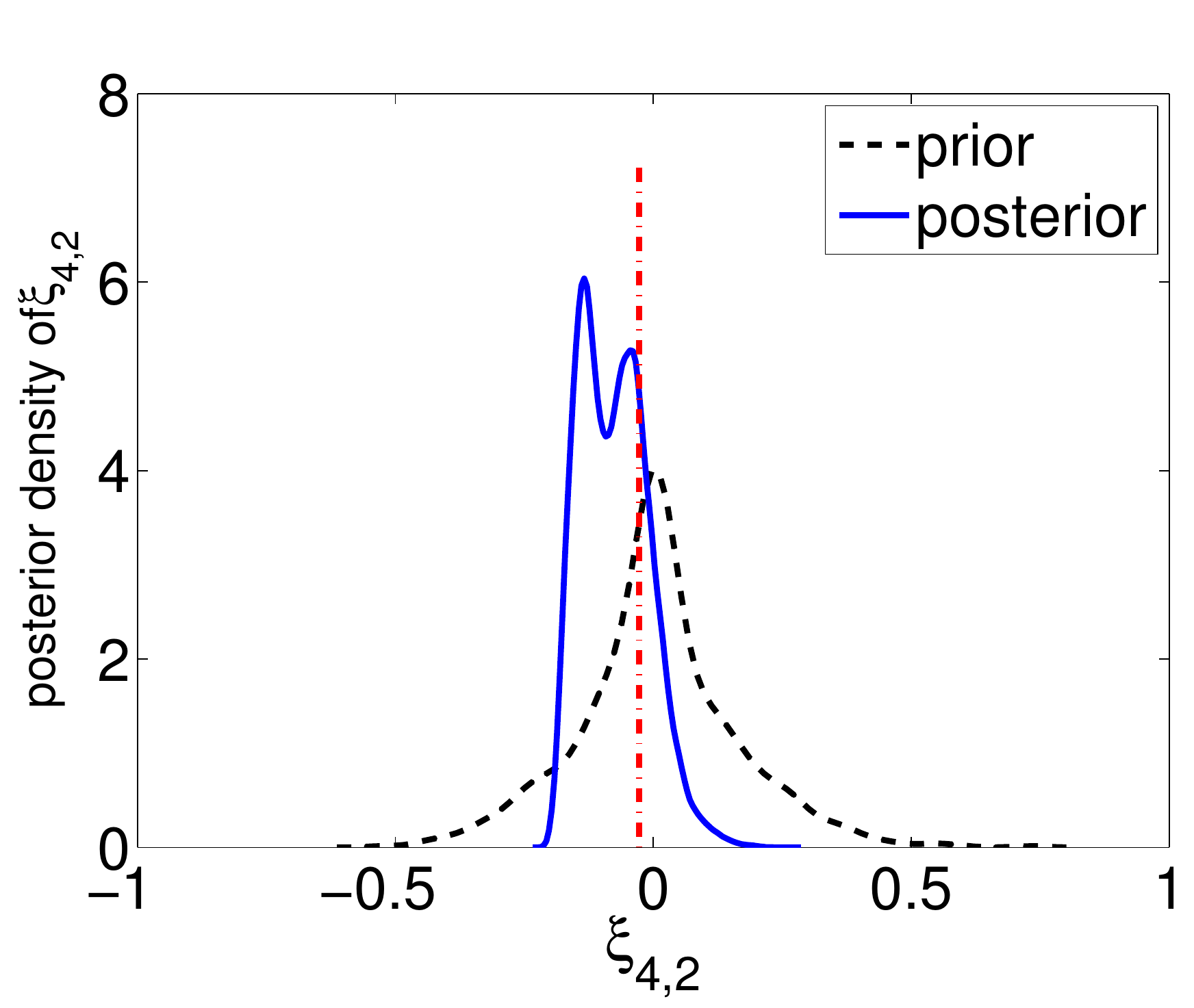}\\
\includegraphics[scale=0.25]{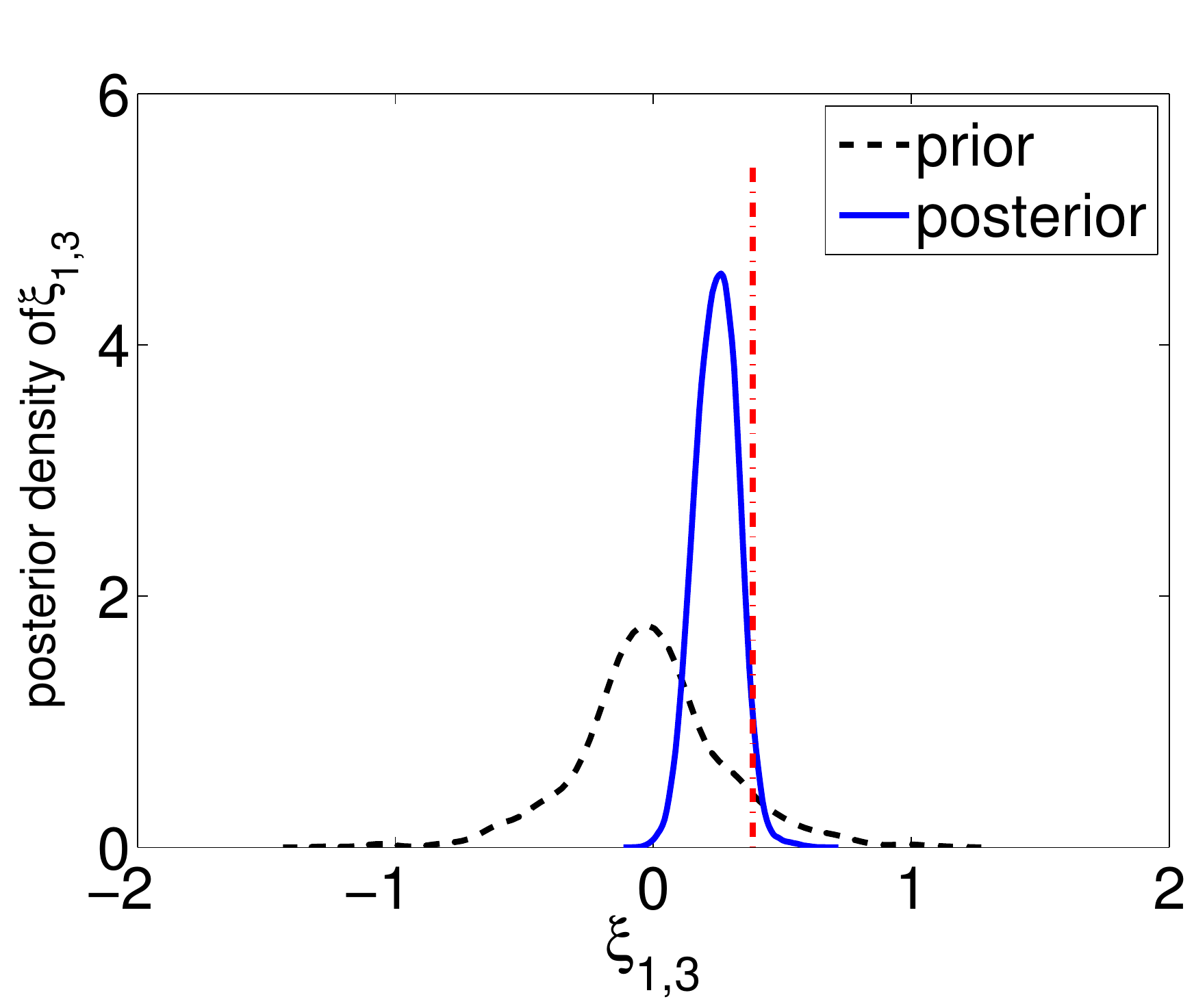}
\includegraphics[scale=0.25]{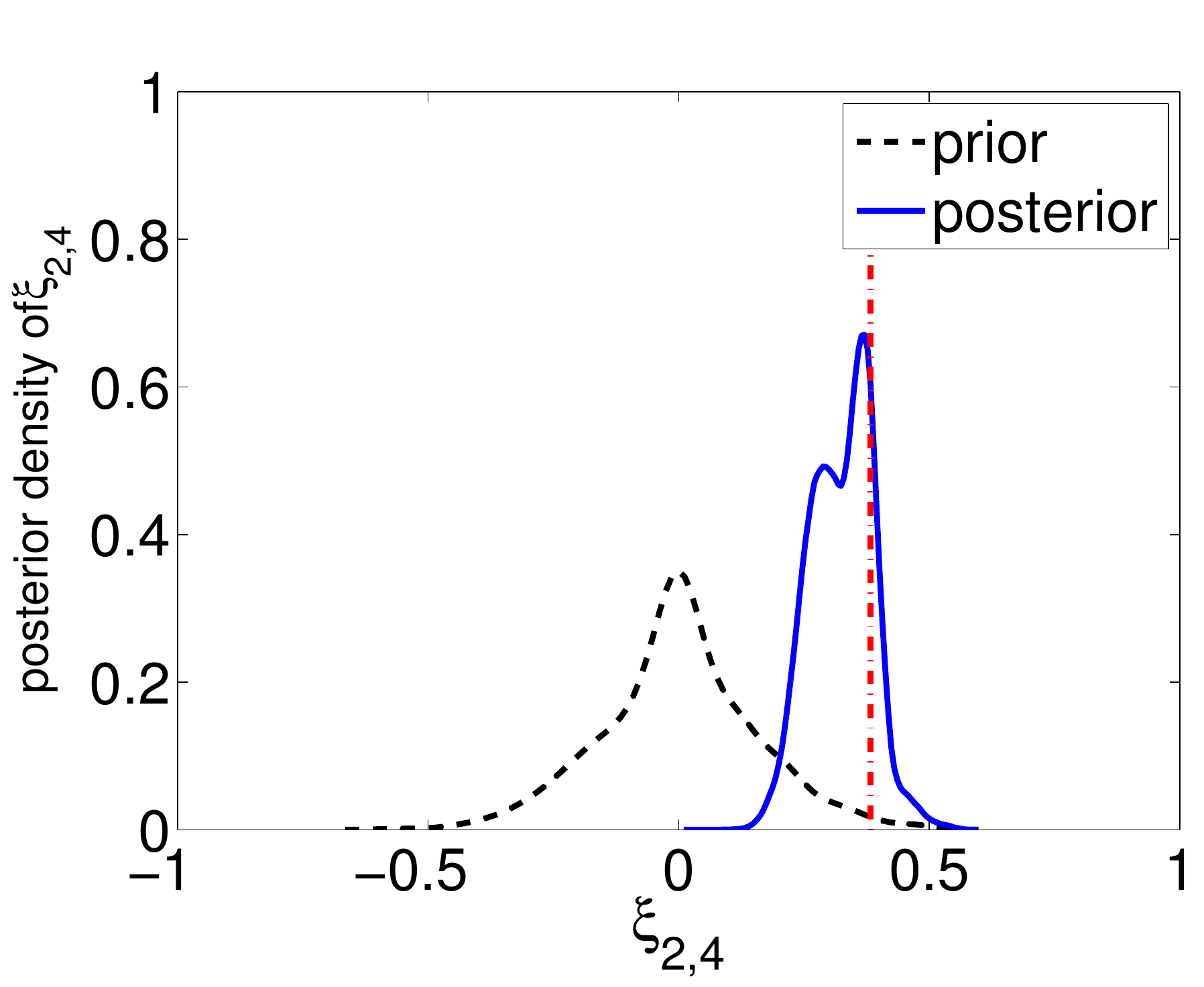}
\includegraphics[scale=0.25]{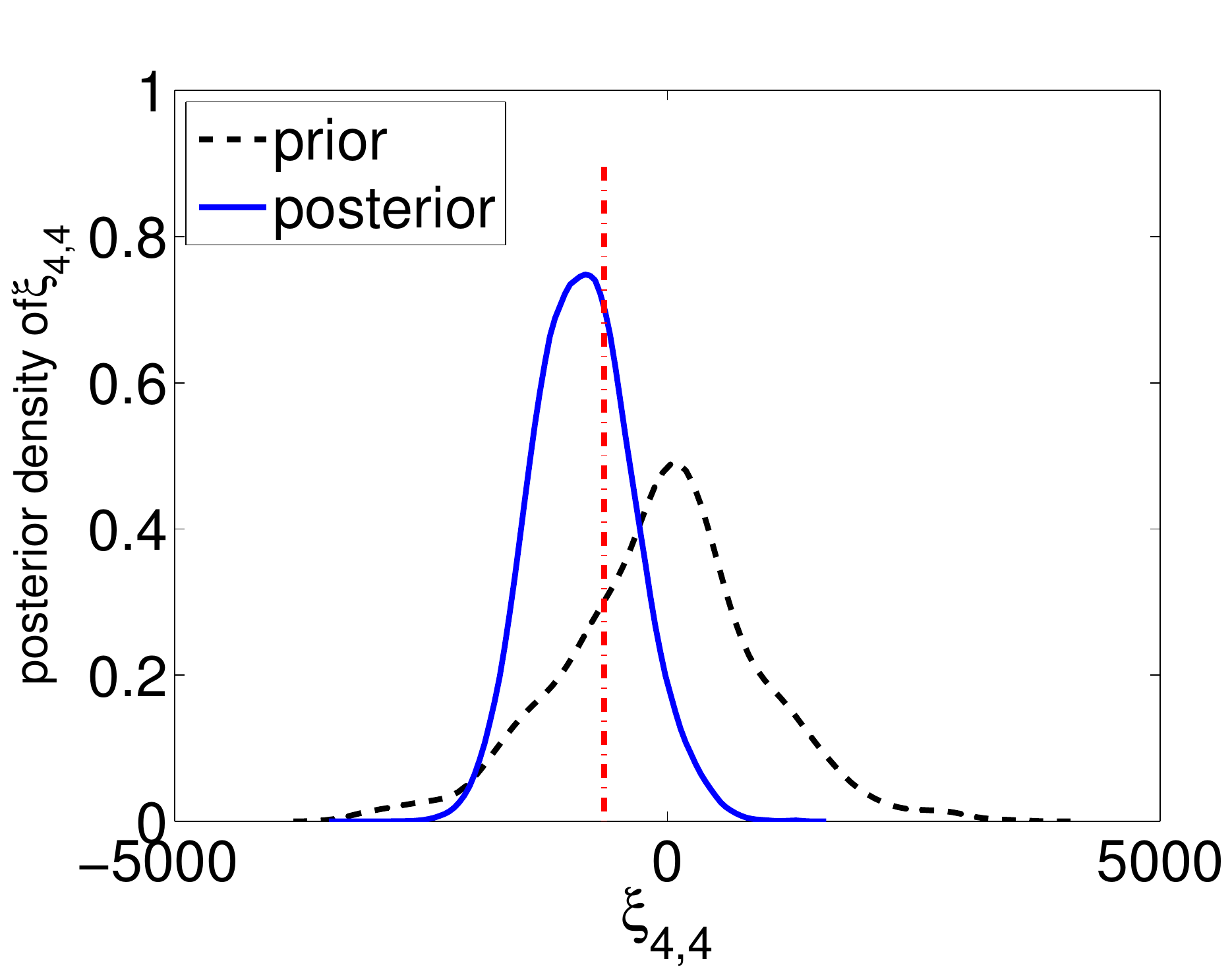}
 \caption{Identification of structural geology (layer model). Densities of the prior and posterior of some DCT coefficients of the  $\kappa$'s obtained from MCMC samples on the level set for $L=0.4$ (vertical dotted line indicates the truth).}   
   \label{Fig15}
\end{center}
\end{figure}

\section{Conclusions}
\label{sec:con}

The primary contributions of this paper are:

\begin{itemize}

\item We have formulated geometric inverse problems within the
Bayesian framework.

\item This framework leads to a well-posedness of the level set
approach, something that is hard to obtain in the context of
classical regularization techniques for level set inversion
of interfaces. 

\item The framework also leads to the use of state-of-the-art
function-space MCMC methods for sampling of the posterior
distribution on the level set function. An explicit motion
law for the interfaces is not needed: the MCMC accept-reject mechanism
implicitly moves them.

\item We have studied two examples: inverse source reconstruction
and an inverse conductivity problem. In both cases we have
demonstrated that, with appropriate choice of priors, the
abstract theory applies. We have also highlighted the behavior
of the algorithms when applied to these problems. 

\item The fact that no explicit level set equation is required helps 
to reduce potential issues arising in level set inversion, such as
reinitialization. In most computational level set approaches \cite{BO05}, 
the motion of the interface by means of the standard level set equation 
often produces level set functions that are quite flat. For the mean curvature flow problem, such fattening phenomena
is firstly observed in \cite{ES91} where the surface evolution starts from a "figure eight" shaped initial surface; in addition it has been shown to happen even if the initial surface is smooth \cite{AIC95}.  This causes 
stagnation as the interface then move slowly. Ad-hoc approaches, such as 
redistancing/reinitializing the level set function with a signed distance 
function, are then employed to restore the motion of the interface. 
In the proposed computational framework, not only does the MCMC accept-reject 
mechanism induce the motion of the intertace, but it does so in a way that 
avoids creating flat level set functions underlying the permeability. 
Indeed, we note that the posterior samples of the level set function inherit the same properties from the ones of the prior. Therefore, the probability of 
obtaining a level set function which takes any given value on a set of
positive measure is zero under the posterior, as it is under the prior.
This fact promotes very desirable, and automatic, algorithmic robustness.

\end{itemize}

Natural directions for future research include the following:

\begin{itemize}

\item The numerical results for the two examples that we consider
demonstrate the sensitive dependence of the posterior distribution
on the length-scale parameter of our Gaussian priors. It
would be natural to study automatic selection techniques
for this parameter, including hierarchical Bayesian modelling. 

\item We have assumed that the values of $\kappa_i$ on each unknown domain
$D_i$ are both known and constant. It would be interesting, and possible,
to relax either or both of these assumptions, as was done in the
finite geometric parameterizations considered in \cite{ILS14}.

\item The numerical results also indicate that initialization
of the MCMC method for the level set function 
can have significant impact on the performance of the inversion
technique; it would be interesting to study this issue more
systematically.

\item The level set formulation we use here, with a single level set function
and possibly multiple level set values $c_i$ has been used for modeling island dynamics \cite{Merriman98levelset} where a nested structure is assumed for the regions $\{D_i\}_{i=1}^n$ see Figure \ref{fig_dom_a}. However, we comment that there exist objects with non-nested regions, such as the one depicted in Figure \ref{fig_dom_b}, which can not be represented by a single level set function.
. It would be of interest to extend this work to the consideration
of vector-valued level set functions. 
In the case of binary obstacles, it is enough to represent the shape via a single level set function (cf. \cite{S96}). However, in the case of $n$-ary obstacles or even more complex geometric objects, the representation is more complicated; see the review papers \cite{DL06, DL09, TC04} for more details.

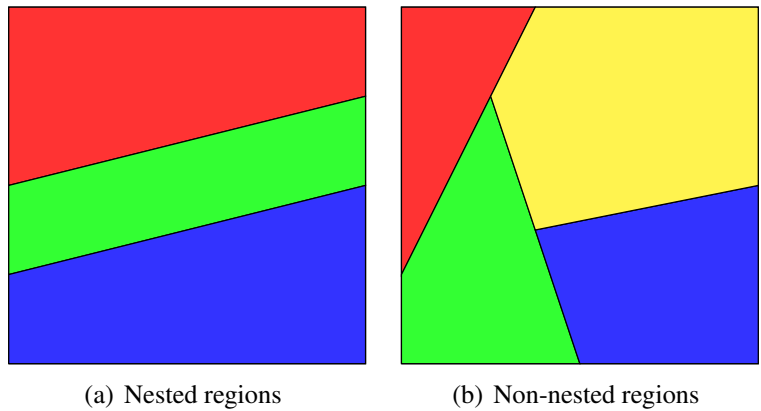
\begin{figure}[h]
\centering
\subfigure[Nested regions]{\label{fig_dom_a}
 \resizebox{0.3\textwidth}{!}{
\begin{tikzpicture}
\draw (0,0) -- (4,0) -- (4,4) -- (0,4) -- (0,0);
 \draw[fill=red!80] (0,2) -- (4,3) -- (4,4) -- (0,4) -- (0,2) --cycle;
 \draw[fill=green!80] (0,1) -- (4,2) -- (4,3) -- (0,2) --cycle;
 \draw[fill=blue!80] (0,0) -- (4,0) -- (4,2) -- (0,1) --cycle;
\end{tikzpicture}
}}
\subfigure[Non-nested regions]{\label{fig_dom_b}
 \resizebox{0.3\textwidth}{!}{
\begin{tikzpicture}
\draw (0,0) -- (4,0) -- (4,4) -- (0,4) -- (0,0);
 \draw[fill=red!80] (0,1) -- (1.5,4) -- (0,4) --cycle;
 \draw[fill=green!80] (0,0)-- (2,0)-- (1,3)-- (0,1) --cycle;
 \draw[fill=blue!80] (2,0)-- (4,0)-- (4,2)-- (1.5,1.5) --cycle;
 \draw[fill=yellow!80] (1.5,1.5) -- (4,2)--(4,4)--(1.5,4)--(1,3)--cycle;
\end{tikzpicture}
}}
\caption{Nested regions and non-nested regions}\label{fig_dom}
 \end{figure}

\item The Bayesian framework could be potentially combined with other parameterizations of unkwon geometries. For example, the pluri Gaussian approach has been used with EnKF in \cite{liu2005ensemble} to identify geologic facies. 

\end{itemize}

\vspace{0.2in}

\noindent{\bf Acknowledgments} 
YL is  is supported by EPSRC as part of the MASDOC DTC at the University of Warwick with grant No. EP/HO23364/1.
AMS is supported by the (UK) EPSRC Programme Grant EQUIP, and by the (US)
Ofice of Naval Research. 

\vspace{0.2in}

\bibliographystyle{plain}
\bibliography{ref}

\section{Appendix 1}

\begin{proof}[Proof of Theorem \ref{t:main}]
Notice that the random variable $y|u$ is distributed according to the measure $\Q_u$, which is the translate of $\Q_0$ by $\G(u)$, satisfying $\Q_u \ll \Q_0$ with Radon-Nikodym derivative
\[
\frac{\dd \Q_u}{\dd \Q_0}(y) \propto \exp\Bigl(- \Phi(u;y) \Bigr);
\]
where $\Phi: U\times \R^J \gt \R$ is the least squares function
given in \eref{Phi}.
Thus for the given data $y$, $\Phi(u; y)$ is up to a constant, the {\em negative log likelihood}. We denote by $\nu_0$ the product measure
\be\label{eq:prodmeas}
\nu_0 (du, dy) = \mu_0(d u)\Q_0(d y).
\en
Suppose that $\Phi(\cdot, \cdot)$ is $\nu_0$ measurable, then the random variable $(u, y)\in U\times Y$ is distributed according to $\nu (d u, d y)$ with
\[
\frac{\dd \nu}{\dd \nu_0}(u, y) \propto \exp\Bigl( -\Phi(u;y) \Bigr).
\] 

From Assumptions \ref{a:1}(ii) and the continuity of $\Phi(u;y)$ with 
respect to $y$, we know that $\Phi(\cdot;\cdot): U\times Y \gt \R$ is 
continuous $\nu_0-$almost surely. Then it follows from Lemma \ref{lem_asc} 
below that $\Phi(\cdot;\cdot)$ is $\nu_0$-measurable. On the other hand, by 
Assumptions \ref{a:1}(i), for $|y|_\Gamma \leq r$, we obtain the upper 
and lower bound for $Z$,
   \[
 0 < \exp(- K(r, \kappa_{\min}))   \leq Z = \int_U \exp(-\Phi(u; y))\mu_0 (du) \leq 1
   \]
Thus the measure is normalizable and applying the Bayes' Theorem 3.4
from \cite{S13} yields the existence of $\mu^y$.

Let $Z = Z(y)$ and $Z^\prime = Z(y^\prime)$ be the normalization constants for $\mu^y$ and $\mu^{y^\prime}$, i.e.
   \[
    Z = \int_U \exp(-\Phi(u; y))\mu_0 (du), \quad  Z^\prime = \int_U \exp(-\Phi(u; y^\prime))\mu_0 (du)
   \]
   We have seen above that
   \[
   1 \geq Z, Z^\prime \geq \exp(- K(r,\kappa_{\rm min})) > 0.
   \]
    From Assumptions \ref{a:1}(iii),
   \[
   |Z - Z^\prime|  \leq \int |\exp(-\Phi(u;y)) - \exp(-\Phi(u;y^\prime))| \mu_0 (\dd u) \leq \int |\Phi(u;y) - \Phi(u;y^\prime)| \mu_0 (\dd u) \\  \leq C(r) |y - y^\prime |_\Gamma
   \]
   Thus, by the definition of Hellinger distance, we have
   \[
   2 d_{{\rm Hell}} (\mu^y, \mu^{y^\prime})^2  = \int \left( Z^{-1/2} \exp\left(-\frac{1}{2}\Phi(u;y)\right) -  (Z^\prime)^{-1/2} \exp\left(-\frac{1}{2}\Phi(u;y^\prime)\right)\right)^2 \mu_0 (\dd u) \\
   \leq I_1 + I_2
   \]
   where
\begin{eqnarray*}
       I_1 & =& \frac{2}{Z} \int \left( \exp\left(-\frac{1}{2}\Phi(u;y)\right) - \exp\left(-\frac{1}{2}\Phi(u;y^\prime)\right)\right)^2 \mu_0 (\dd u) \\
       I_2  &=& 2|Z^{-1/2} - (Z^\prime)^{-1/2}|^2\int \exp(-\Phi(u; y^\prime)) \mu_0 (\dd u)
\end{eqnarray*}
   Applying (i) and (iii) in Assumptions \ref{a:1} again yields
   \[
   \frac{Z}{2} I_1 \leq C(r)|y - y^\prime|^2_{\Gamma}
   \]
   and
\[
   I_2 \leq C(r) |Z^{-1/2} - (Z^\prime)^{-1/2}|^2 \leq C(r) (Z^{-3} \vee (Z^\prime)^{-3}) |Z - Z^\prime|^2 \leq C(r)|y - y^\prime|^2_\Gamma
   \]
   Therefore the proof that the measure is Lipschitz is 
completed by combining the preceding estimates. The final statement follows
as in \cite{S10}, after noting that $f \in L^2_{\mu_0}(U;S)$ implies
that  $f \in L^2_{\mu^y}(U;S)$, since $\Phi(\cdot;y) \ge 0.$
\end{proof}

\begin{lem}\label{lem_asc}
 Let $U$ be a separable Banach space and $(U, \Sigma, \mu)$ be a complete probability space with $\sigma$-algebra $\Sigma$. If a functional $\cF: U \gt \R$ is continuous $\mu$-almost surely, i.e. $\mu(M) = 1$ where $M$ denotes the set of the continuity points of $\cF$, then $\cF$ is $\Sigma$-measurable. 
\end{lem}

\begin{proof}
 By the definition of measurability, it suffices to show that for any $c > 0$
 \[
 S := \{u\in U\ |\ \cF(u) > c\} \in \Sigma.
 \]
 One can write $S$ as $S = (S\cap M) \cup (S\setminus M)$. Since $\cF$ is continuous $\mu$-almost surely, $M$ is measurable and $\mu (M) = 1$. It follows from the completeness of the measure $\mu$ that $S \setminus M$ is measurable and $\mu(S \setminus M) = 0$. Now we claim that $S \cap M$ is also measurable. Denote $B_{\delta} (u) \subset U$ to be the ball of radius $\delta$ centered at $u \in U$. For each $v\in S \cap M$, as $\cF$ is continuous at $v$, there exists $\delta_v > 0$ such that if $v^\prime\in B_{\delta_v}(v)$, then $\cF(v^\prime) > c$. Therefore $S\cap M$ can be written as
\[
 S\cap M = M  \bigcap \bigcup_{v\in S \cap M} {B}_{\delta_{v}}\left(v \right)
 \]
 that is the intersection of the measurable set $M$ with the open set $\bigcup_{v\in S \cap M} {B}_{\delta_{v}}\left(v \right)$. So $S \cap M$ is measurable. Then it follows that $\cF$ is $\Sigma$-measurable.
 \end{proof}

\begin{proof}[Proof of Proposition \ref{thm_cont}]
``$\Longleftarrow$.''  
Let $\{u_\varepsilon\}$ denote
any approximating family of level set functions with limit $u$ as
$\varepsilon \to 0$ in $C(\OL{D};\R): \|u_\varepsilon - u\|_{C(\OL{D})} < \varepsilon \gt 0$.  Let $D_{i,\varepsilon}, D^0_{i,\varepsilon}$ be the sets defined in \eref{eq_lss} associated with the approximating level set function $u_\varepsilon$ and define $\kappa=F(u)$ by \eref{eq_lsf} and, similarly,
$\kappa_\varepsilon := F(u_\varepsilon)$. Let $m(A)$ denote the Lebesgue measure of the set $A$.

Suppose that $m (D^0_i) = 0, i = 1, \cdots, n-1$. Let $\{u_\varepsilon\}$ be the above approximating functions. We shall prove $\|\kappa_\varepsilon - \kappa\|_{L^q(D)} \gt 0$. In fact, we can write
  \begin{eqnarray*}
  \kappa_\varepsilon(x) - \kappa(x) & = \sum_{i=1}^n \sum_{j=1}^n (\kappa_i - \kappa_j) \I_{D_{i, \varepsilon} \cap D_j}(x)\\
 & = \sum_{i,j=1, i\neq j}^n (\kappa_i - \kappa_j) \I_{D_{i, \varepsilon} \cap D_j}(x).
  \end{eqnarray*}
  By the definition of $u_\varepsilon$, for any $x\in D$
  \be\label{eq_lsf_ineq}
  u(x) - \varepsilon <  u_\varepsilon(x) < u(x) + \varepsilon
  \en
  Thus for $|j- i| > 1$ and $\varepsilon$ sufficiently small, $D_{i,\varepsilon}\cap D_j = \varnothing$. For the case that $|i - j| = 1$, from \eref{eq_lsf_ineq}, it is easy to verify that

\begin{eqnarray}\label{eq_lsf_set}
   D_{i,\varepsilon}\cap D_{i+1}  \subset \widetilde{D}_{i,\varepsilon} &:=& \{x\in D\ |\ c_i \leq u(x) < c_i + \varepsilon\} \gt D^0_{i}, \ i = 1,\cdots, n-1\\
   D_{i,\varepsilon}\cap D_{i-1}  \subset \widehat{D}_{i-1,\varepsilon}&:= &\{x\in D\ |\ c_{i-1} - \varepsilon < u(x) < c_{i-1}\} \gt \varnothing ,\ i = 2, \cdots n
\end{eqnarray}

  as $\varepsilon \gt 0$. By this and the assumption that $m(D^0_i) = 0$, we have that $m(D_{i,\varepsilon}\cap D_j) \gt 0$ if $i\neq j$. Furthermore, the Lebesgue's dominated convergence theorem yields
  \[
  \|\kappa_\varepsilon - \kappa\|_{L^q(D)}^q = \sum_{i,j = 1, i\neq j}^n \int_{D_{i, \varepsilon}\cap D_j}|\kappa_i - \kappa_j|^q\dd x \gt 0
  \]
  as $\varepsilon \gt 0$. Therefore, $F$ is continuous at $u$.

``$\Longrightarrow$.'' We prove this by contradiction.  Suppose that there exists $i^\ast$ such that $m(D^0_{i^\ast}) \neq 0$. We define $u_\varepsilon := u - \varepsilon$, then it is clear that $\|u_\varepsilon - u\|_{C(\OL{D})}\gt 0$ as $\varepsilon \gt 0$. By the same argument used in proving the sufficiency,
  \[
     \|\kappa_\varepsilon - \kappa\|^q_{L^q(D)}  = \sum_{i=1}^{n-1} \int_{\widetilde{D}_{i,\varepsilon}\cup \widehat{D}_{i,\varepsilon}} |\kappa_{i+1} - \kappa_i|^q\dd x \gt \sum_{i=1}^{n-1}\int_{D^0_{i}} |\kappa_{i+1} - \kappa_i|^q \dd x \\
      > \int_{D^0_{i^\ast}} |\kappa_{i^\ast+1} - \kappa_{i^\ast}|^q \dd x > 0
  \]
  where we have used $m(D^0_{i^\ast}) \neq 0$ in the last inequality. However, this contradicts with the assumption that $F$ is continuous at $u$.
 \end{proof}

\section{Appendix 2}

Recall the Gaussian measure $\mu_0 = \N(0, \C)$ on the function space $\HH$
where $\C = \C_i, \HH = \HH_i, i = 1,2$ given in subsection \ref{ssec:pri}.
These measures can be constructed as Gaussian 
random fields. 

Let $(\Omega, \mathscr{F}, \mathbb{P})$ be a complete probability space, i.e. if $A\in \mathscr{F}$ with $\PP(A) = 0$, then $\PP(B) = 0$ for $B\subset A$. A random field on $D$ is a measurable mapping $u: D\times \Omega \gt \R$. Thus, for any $x\in D$, $u(x; \cdot)$ is a random variable in $\R$; whilst for any $\omega\in \Omega$, $u(\cdot; \omega): D\gt \R$ is a vector field. Denote by $(\R^\mathbb{N}, \B(\R^\mathbb{N}), \mathbb{P})$ the probability space of i.i.d standard Gaussian sequences equipped with product $\sigma$-algebra and product measure. In this paper, we consider $(\Omega, \mathscr{F}, \mathbb{P})$ as the completion of $(\R^\mathbb{N}, \B(\R^\mathbb{N}), \mathbb{P})$ in which case the $\sigma$-algebra $\mathscr{F}$ consists of all sets of the type $A\cup B$, where $A\in \B(\R^\mathbb{N})$ and $B \subset N\in \B(\R^\mathbb{N})$ with $\mathbb{P}(N) = 0$. Let $\omega = \{\xi_k\}_{k=1}^\infty \in \Omega = \R^\mathbb{N}$ be an i.i.d sequence with $\xi_1 \sim \N(0,1)$, 
and consider the random function $u\in \HH$ defined via the Karhunen-Lo\'eve 
expansion

\be\label{eq_rf}
   u(x; \omega) =\mathcal{T}(\omega):= \sum_{k=1}^{\infty} \sqrt{\lambda_k}\xi_k \phi_k (x),
   \en
   where $\{\lambda_k, \phi_k\}_{k=1}^\infty$ is the eigensystem of $\C$. By the theory of Karhunen-Lo\'eve expansions \cite{Bog98}, the law of the random function $u$ is identical to $\mu_0$. Recalling that $\alpha > 1$, the eigenvalues $\{\lambda_k\}$ of $\C_1$ decay like $k^{-\alpha}$ in two dimensions; whilst the eigenvalues of $\C_2$ will decay exponentially. Moreover, we assume further that $\phi_k\in U$ and that $\sup_k\|\phi_k\|_\infty < \infty$ which holds in simple geometries.
    Due to the decaying properties of the eigenvalues of $\C$, the truncated sum
\be\label{eq_rf2}
    \T_N(\omega) = \sum_{k=1}^{N} \sqrt{\lambda_k}\xi_k \phi_k
   \en
   admits a limit $\T$ in $L^2_{\PP}(\Omega; \HH)$. By the Kolmogorov
Continuity Theorem \cite{S13}, $\T$ is in fact H\"older continuous $\PP-$almost surely; in particular, $\T\in U$ $\PP$-almost surely. Then by Theorem 3.1.2 in \cite{AT07}, we have $\T_N\gt \T$ in the uniform norm of $U$, $\PP$-almost surely. Since for any $N\in \mathbb{}$, $\T_{N}: (\Omega, \mathscr{F})\gt (U, \mathscr{B}(U))$ is continuous and thus measurable, we know from the completeness of $(\Omega,\mathscr{F}, \PP)$ that the limit $\T$ is also measurable from $(\Omega, \mathscr{F})$ to $(U, \mathscr{B}(U))$ (see p30 in \cite{SS09}). The measurability of $\T$ enables us to define a new measure on $(U, \mathscr{B}(U))$ which we still denote by $\mu_0$ by the following:
 \be\label{eq_pf}   
 \mu_0 (A) = \mathbb{P}(\mathcal{T}^{-1}(A)) =  \mathbb{P}\left(\{\omega\in \Omega\ |\ \ u(\cdot; \omega)\in A\}\right
)\quad \textrm{ for } A\in \mathscr{B}(U).
   \en
   Thus $\mu_0$ is indeed the push-forward measure of $\mathbb{P}$ through $\mathcal{T}$. By definition, it is not hard to verify that $\mu_0$ is the Gaussian measure $\N(0,\C)$ on $(U, \mathscr{B}(U))$.
In addition, suppose that $B \subset N\in \mathscr{B}(U)$ with $\mu_0(N) = 0$; if we still define $\mu_0(B)$ according to \eref{eq_pf}, then $\mu_0(B) = \mathbb{P}(\mathcal{T}^{-1}(B)) = 0$ by the fact that $\mathcal{T}^{-1}(B)\subset \mathcal{T}^{-1}(N)$ and the completeness of $(\Omega, \mathscr{F}, \PP)$. Denote by $\Sigma$ the smallest $\sigma$ algebra containing $\BB(U)$ and all sets of zero measure under $\mu_0$ so that any set $E\in \Sigma$ is of the form $E = A\cup B$, where $A\in\mathscr{B}(U)$ and $B\subset N \in\mathscr{B}(U)$ with $\mu_0(N) = 0$.  Then $(U, \Sigma, \mu_0)$ is complete. 

We comment that although a Gaussian measure is usually defined as a Borel measure in the literature (see e.g. \cite{Bog98}), it is more convenient to work with a complete Gaussian measure in this paper; in particular, the completeness of $\mu_0$ is employed to show the measurability of the observational map in level set based inverse problems.

Considering a Gaussian random function $u(\cdot;\omega)$ with $\omega\in \Omega$, for any level constant $c\in \R$, we define the random level set 
\be\label{rls}
D^0_c = D^0_c(u(\cdot;\omega)) = D^0_c(\omega) := \{x\ |\ u(x; \omega) = c\}.
\en 
Recall that the measure space $(U, \Sigma, \mu_0)$ is the push-forward of $(\Omega, \mathscr{F}, \PP)$ under $\T$. We define the functional $\M_c: U \gt \R$ by
\[ 
\M_c u = m(\{x\ |\ u(x) = c\})
\]
and the composition $\RR_c := \M_c\circ \T$, as illustrated in the following commutative diagram:
\[
\begin{tikzcd}
(\Omega, \mathscr{F}, \mathbb{P}) \arrow{r}{\T}\arrow{rd}[swap]{\RR_c = \M_c\circ \T}
&(U, \Sigma, \mu_0)\arrow{d}{\M_c}\\
&(\R, \BB(\R))
\end{tikzcd}
\]

\begin{lem}\label{lem_meas1}
 For any $c\in \R$, $\M_c$ is $\Sigma$-measurable and $\RR_c$ is $\mathscr{F}$-measurable so that $m(D^0_c)$ is a random variable on both $(U, \Sigma, \mu_0)$ and $(\Omega, \mathscr{F}, \PP)$. 
\begin{proof}
 To prove $\M_c$ is $\Sigma$-measurable, we only need to check the set $A_t := \{u\in U \ |\  \M_c u \geq t\}\in \Sigma$ for any $t \in \R$. Since $\M_c$ is a non-negative map, for $t\leq 0$, it is obvious that $A_t = U$ and hence measurable. Now we claim that $A_t$ is closed in $U$ for $t > 0$. To that end, let $\{u_n\}_{n=1}^{\infty}$ be a sequence of functions in $A_t$ such that $\|u_n - u\|_U \gt 0$ for some $u\in U$ as $n\gt \infty$. We prove that $u\in A_t$. Since $\|u_n - u\|_U \gt 0$, there exists a subsequence which is still denoted by $u_n$ such that $\|u_n - u\|_U < 1/n$. By the definition of $A_t$, $u_n \in A_t$ means that $m(\{x\in \OL{D}\ | \ u_n(x) = c\}) \geq t$ for all $n$. Moreover, from the construction of $u_n$, $\{x\in \OL{D}\ |\ u_n(x) = c\}\subset B_n:= \{x\in \OL{D}\ | \ |u(x) - c| < 1/n\}$, which implies that $m(B_n) \geq t$. Noting that 
 \[
 \{x\in \OL{D}\ |\ u(x) = c\} = \cap_{n=1}^\infty B_n
 \]
 and that $B_n$ is decreasing, we can conclude that $m(\{x\in \OL{D}\ |\ u(x) = 0\}) \geq t$, i.e. $u\in A_t$. So $A_t$ is closed for $t > 0$. Then it follows from the measurability of $\T$ that $\RR_c$ is $\mathscr{F}$-measurable. Therefore $m(D^0_c)$ is a random variable on both $(U, \Sigma, \mu_0)$ and $(\Omega, \mathscr{F}, \PP)$. 
\end{proof}
\end{lem}
The following theorem demonstrates that $m(D^0_c)$ vanishes almost surely on both measure spaces above. 
\begin{proposition}\label{thm_lsm}
Consider a random function $u$ drawn from one of the Gaussian probability 
measures $\mu_0$ on $(U,\Sigma)$ given in subsection \ref{ssec:pri}. 
 For $c \in \R$, the random level set $D^0_c$ of $u$ is defined by \eref{rls}. Then 
 
 (i) $m(D^0_c) = 0, \PP$-almost surely;
 
 (ii) $m(D^0_c) = 0, \mu_0$-almost surely.
 \begin{proof}
  (i) For any fixed $x\in D$, since the point evaluation $u(x)$ acts as a bounded linear functional on $U$, $u(x; \cdot)$ is a real valued Gaussian random variable, which implies $\PP(\{\omega\ |\ u(x, \omega) = c\}) = 0$. Moreover, noting that the random field $u : D\times \Omega \gt \R$ is a measurable map, if we view $m(D^0_c)$ as a random variable on $\Omega$, then

\begin{eqnarray*}
     \E[m(D^0_c)]  = \int_{\Omega} m\left(D^0_c(\omega)\right) \dd \mathbb{P}(\omega) = \int_{\Omega} \int_{\R^d} \I_{\{x\,|\, u(x; \omega) = c\}}\dd x \dd \mathbb{P}(\omega) \\
      = \int_{\Omega} \int_{\R^d} \I_{\{(x,\, \omega)\,|\, u(x; \omega) = c\}}\dd x \dd \mathbb{P}(\omega) \stackrel{\textrm{Fubini}}{=}  \int_{\R^d} \int_{\Omega} \I_{\{(x,\, \omega)\,|\, u(x; \omega) = c\}} \dd \mathbb{P}(\omega) \dd x \\ 
     = \int_{\R^d} \int_{\Omega} \I_{\{\omega\,|\,u(x; \omega) = c\}} \dd \mathbb{P}(\omega) \dd x = \int_{\R^d} \mathbb{P}(\{\omega\,|\, u(x; \omega)= c\}) \dd x  = 0
\end{eqnarray*}
  Noting that $m(D^0_c)\geq 0$, we obtain $m(D^0_c) = 0, \mathbb{P}$-almost surely.
  
  (ii) Recall that $A_t = \{u\in U \ |\  \M_c u \geq t\}$ defined in Lemma \ref{lem_meas1} is closed in $U$ for any $t > 0$. Thus the set $A := \{u\in U \ |\ m(\{x\ |\ u(x) = c\}) = 0\}= (\cup_{k=1}^{\infty} A_{1/k})^c = \cap_{k=1}^\infty A_{1/k}^c$ is a Borel set of $U$ and measurable. Since $\mu_0$ is the push-forward measure of $\PP$ under $\T$,
  \[
  \mu_0(A) = \PP(\T^{-1}(A)) = \PP(\{\omega\ |\ m(D_c^0(\omega)) = 0\}) = 1
  \]
  where the last equality follows from (i).
 \end{proof}
\end{proposition}

\end{document}